\documentclass[11pt]{article}
\usepackage[textsize=footnotesize]{todonotes}
\usepackage{hyperref}

\usepackage[title]{appendix}
\usepackage{graphicx}
\usepackage{amsmath}         
\usepackage{amssymb}
\usepackage{amsthm}
\usepackage[countmax]{subfloat}  
\usepackage{fourier}
\usepackage{xcolor}
\usepackage{mathtools}
\usepackage{enumerate}
\usepackage{natbib}
\usepackage{url}
\usepackage{verbatim}
\usepackage{booktabs}
\usepackage{placeins}
\usepackage{multirow}
\usepackage{caption}
\usepackage{subcaption}
\usepackage{mathrsfs} 
\def\R{\mathbb R}     
\def\E{\mathbb E}
\def\V{\text{Var}}
\def\cv{\text{Cov}}

\def\cm{\text{Cum}}

\newtheorem{prop}{Proposition}[section]     
\newtheorem{thm}{Theorem}[section]
\newtheorem{lemma}{Lemma}[section]

\newtheorem{Remark}{Remark}[section]
\newtheorem{Corollary}{Corollary}[section]
\newtheorem{assumption}{Assumption}[section]
\newtheorem{properties}{Properties}[section]
\newcommand{\im}{\mathrm{i}}

\newcommand{\mynegspace}{\hspace{-0.12em}}

\DeclareMathAlphabet{\mathantt}{OT1}{antt}{li}{it}
\newcommand{\bigsnorm}[1]{\Big\rvert\mynegspace\Big\rvert\mynegspace\Big\rvert\mynegspace {#1} \Big\rvert\mynegspace \Big\rvert\mynegspace \Big\rvert}
\newcommand{\snorm}[1]{\rvert\mynegspace\rvert\mynegspace\rvert\mynegspace {#1} \rvert\mynegspace \rvert\mynegspace \rvert}

\newcommand{\innprod}[2]{\langle #1, #2 \rangle}
\newcommand{\biginnprod}[2]{\Big\langle #1, #2 \Big\rangle}

\newcommand{\opnorm}[1]{\vert\kern-0.25ex\vert\kern-0.25ex\vert #1 
    \vert\kern-0.25ex\vert\kern-0.25ex\vert}

\newcommand{\fdft}[2]{ D_N^{{u_{j_{{#1}}}},{\omega_{k_{#2}}}}}
\newcommand{\fdftc}[2]{ D_N^{{u_{j_{{#1}}}},{-\omega_{k_{#2}}}}}
\newcommand{\fdftl}[2]{ D_N^{{u_{j_{{#1}}}},{\omega_{k_{#2}-1}}}}
\newcommand{\fdftcl}[2]{ D_N^{{u_{j_{{#1}}}},{-\omega_{k_{#2}-1}}}}
\newcommand{\rnum}{\mathbb{R}}
\newcommand{\znum}{\mathbb{Z}}
\newcommand{\nnum}{\mathbb{N}}
\newcommand{\XT}[1]{X_{#1, T}}
\newcommand{\Xu}[1]{X^{(u)}_{#1}}
\newcommand{\Xuu}[2]{X^{(#1)}_{t_{#2}}}
\newcommand{\cumbp}[2]{\kappa_{{#1};t_1,\ldots,t_{{#2}-1}}}
\newcommand{\cumloc}[2]{c_{#1;t_1,\ldots,t_{{#2}-1}}(\tau_1,\ldots,\tau_{#2})}
\newcommand{\floc}[2]{f_{#1;\omega_1,\ldots,\omega_{{#2}-1}}(\tau_1,\ldots,\tau_{#2})}

\newcommand{\F}{\mathcal{F}}
\newcommand{\Eps}{\mathcal{E}}

\allowdisplaybreaks

\newcommand\tageq{\addtocounter{equation}{1}\tag{\theequation}}
\DeclareMathOperator{\Tr}{Tr}

\allowdisplaybreaks

\usepackage{authblk}
\usepackage{tcolorbox}

\begin{document}

\title{A nonparametric test for stationarity \\in functional time series}

\author[a]{Anne van Delft}
\author[b]{Vaidotas Characiejus}
\author[a]{Holger Dette}
\affil[a]{Ruhr-Universit\"at Bochum, Fakult\"at f\"ur Mathematik, 44780 Bochum, Germany}
\affil[b]{D\'epartement de math\'ematique, Universit\'e libre de Bruxelles, Belgium}

 \maketitle
\begin{abstract}
We propose a new measure for stationarity of a functional time series, which is based on an explicit representation  of the  $L^2$-distance between the spectral density operator
of a non-stationary process  and its best  ($L^2$-)approximation by a spectral density operator corresponding to a stationary  process. This distance can easily be estimated 
by sums of Hilbert-Schmidt inner products of periodogram operators (evaluated at different frequencies), and asymptotic normality of an  appropriately standardized version of the estimator 
can be established for the corresponding estimate under the null hypothesis 
and alternative. As a result we obtain 
a simple asymptotic   frequency domain  level $\alpha$ test
(using the quantiles of the normal distribution) for the hypothesis of stationarity of functional time series. 
Other applications such as asymptotic confidence intervals for a measure of stationarity or the construction 
of  tests for ``relevant deviations from   stationarity'', are also briefly mentioned.  
We demonstrate in a small simulation study that the new method has very good finite sample properties. Moreover, we apply our test to annual temperature curves.
\end{abstract}

\\

\section{Introduction}  \label{sec1}
\def\theequation{1.\arabic{equation}}
\setcounter{equation}{0}

In many applications of functional data analysis (FDA) data is recorded sequentially over time and naturally exhibits dependence. In the last years, an increasing number of authors have worked on analysing functional data from time series 
and we refer to the monographs of \cite{bosq2000} and \cite{horvkoko2012} among others. An important assumption in most of the literature is stationarity, which allows a unifying development of statistical theory. For example, stationary processes with a linear representation have among others been investigated by
\cite{Mas2000}, \cite{Bosq2002} and \cite{Dehling2005}. Prediction methods \citep[e.g.,][]{as03,adh,bosq2000} and violation of the i.i.d.\ assumption in the context of change point detection have also received a fair amount of attention \citep[e.g.,][]{Aue2009,Berkes2009,Horvath2010}. \cite{hk} provide a general framework to examine temporal dependence among functional observations of stationary processes.  Frequency domain analysis of stationary functional time series has been considered by \cite{Panaretos_Tavakoli_2013} under the assumption of functional generalizations of cumulant-mixing conditions.

In practice, it is however not clear that the temporal dependence structure is constant and hence that stationarity is satisfied. It is therefore desirable to have tests for second order stationarity or measures for deviations from stationarity for data analysis of functional time series. 
In the context of Euclidean data (univariate and multivariate) there exists a considerable amount of literature on this problem. Early work can be found in \cite{prisub1969}
who proposed testing the ``homogeneity'' of a set of evolutionary spectra.  \cite{sacneu2000} used coefficients  with respect to a Haar wavelet series expansion of  time-varying
periodograms for this purpose, see also  \cite{nason2013}  who provided an important extension of their approach
and \cite{cardnason2010} or \cite{tayecknun2014} for further applications of wavelets in the problem of testing for stationarity.  
\cite{paparoditis2009,paparoditis2010} proposed to reject the null hypothesis of second order stationarity if  the $L^2$-distance between
a local spectral density estimate  and an  estimate derived under the assumption of stationarity is large.  
\cite{detprevet2011}  suggested to estimate this distance directly by sums of periodograms evaluated at the Fourier frequencies in order to avoid
the problem  of choosing additional bandwidths [see also \cite{prevetdet2013} for an empirical process approach].
An alternative method  to investigate second order stationarity can be found in  \cite{dwisub2011} and \cite{jensub2015}, who
used the fact that  the discrete Fourier transform (DFT)  is asymptotically uncorrelated at the canonical frequencies if and only if
the time series is second-order stationary. 
Recently,  \cite{jinwang2015} proposed a double-order selection test for checking second-order stationarity of a univariate time series,
while \cite{dasnason2016} investigated an experimental empirical measure of non-stationarity based on the mathematical roughness of
the time evolution of fitted parameters of a dynamic linear model. 

On the other hand -- despite the frequently made assumption of second-order stationarity in functional data analysis -- much less work has been done investigating the stationarity of functional data. A rigorous mathematical framework 
for locally stationary functional time series has only been recently developed by \cite{vde16}, who extended the concept
of local stationarity  introduced by \cite{dahlhaus1996,dahlhaus1997} from univariate time series  to functional data. 
To our best knowledge \cite{avd16_main} is the only reference that applies this framework to test for second-order stationarity of a functional time series against smooth alternatives. These authors follow the approach of \cite{dwisub2011} and show that the functional discrete Fourier transform (fDFT) is asymptotically uncorrelated at distinct Fourier frequencies if and only if the process is functional weakly stationary. This result is then used to construct a test statistic based on an empirical covariance operator of the fDFT's, which is subsequently projected to finite dimension. The asymptotic properties of the resulting quadratic form is demonstrated to be chi-square distributed both under the null and under the alternative of functional local stationarity. Although the authors thereby provide an explicit expression for the degree of departure from weak stationarity, the test requires the specification of the parameter $M$, the number of lagged fDFT's included. This can be seen as a disadvantage since it affects the power of the test.
 
In the present paper we propose a different test which is based on an explicit representation  of the  $L^2$-distance between the spectral density operator of a non-stationary process and its best ($L^2$-)approximation by a spectral density operator corresponding to a stationary process. This measure vanishes if and only if the time series is second order stationary, and consequently a test can be obtained by rejecting the hypothesis of stationarity for large values of a corresponding estimate. The $L^2$-distance is estimated by a functional  of  sums of integrated periodogram operators, for which 
(after appropriate standardization) asymptotic normality can be established under the null hypothesis and any fixed alternative.
The resulting test of the hypothesis of stationarity is extremely simple and therefore very attractive for practitioners. The test uses the quantiles of the standard normal distribution, does  neither  require the choice of a bandwidth in order to estimate the time vary spectral density operators  nor bootstrap methods to obtain critical values.  Therefore the proposed methodology is also  very efficient from a computational point of view.

Although a similar concept has been investigated for univariate time series [see \cite{detprevet2011}], the mathematical derivation of the asymptotic normality requires several sophisticated and new tools for spectral analysis of locally stationary functional time series. In particular - in contrast to  the cited reference -  our approach does not require a linear representation of the time series by an independent sequence and we derive several new properties of the periodogram operator, which are of independent interest. Exemplary we mention 
Theorem \ref{lem:cumHSinprod} in Section \ref{sec5}, which provides  a representation of the cumulants of Hilbert-Schmidt inner products of local periodogram tensors (evaluated at different time points and different frequencies) by the trace of cumulants of simple tensors of the local functional discrete Fourier transforms.

The rest of the paper is organized as follows. In Section \ref{sec2} we introduce the main concept of local stationary functional time series, define a measure of stationarity for these processes and its  corresponding estimates. Section \ref{sec3} is devoted to the asymptotic properties of the proposed estimators and some statistical applications of the asymptotic theory. 
Besides the new  and simple test for  the hypothesis of stationarity  we also briefly discuss several other applications of the asymptotic theory, such as  confidence intervals for the measure of stationarity and  tests for precise 
 hypotheses [see \cite{bergdela1987}], which means in the present context to test if the measure of stationarity exceeds a certain threshold. Section \ref{sec5} contains a proof of the asymptotic normality, while it is demonstrated in Section \ref{sec4} by means of a small simulation study that the new test has very good  finite sample properties. In this section  we also illustrate the application of our test 
 analyzing annual temperature curves recorded at several measuring stations in Australia over the past 135 years. Finally, some of the more technical arguments, which are rather complicated, can be found in the Appendix (see Section \ref{appendix}).

\section{A measure of stationarity on the function space}  \label{sec2}
\def\theequation{2.\arabic{equation}}
\setcounter{equation}{0}

\subsection{Notation and the functional setup}
We begin by providing definitions and facts about operators used in the paper. Suppose that $\mathcal H$ is a separable Hilbert space. $\mathcal L(\mathcal H)$ denotes the Banach space of bounded linear operators $A:\mathcal H\to\mathcal H$ with the operator norm given by $\vert\kern-0.25ex\vert\kern-0.25ex\vert A\vert\kern-0.25ex\vert\kern-0.25ex\vert_\infty=\sup_{\|x\|\le1}\|Ax\|$. Each operator $A\in\mathcal L(\mathcal H)$ has the adjoint operator $A^\dagger\in\mathcal L(\mathcal H)$, which satisfies  $\langle Ax,y\rangle=\langle x,A^\dagger y\rangle$ for each $x,y\in\mathcal H$. $A\in\mathcal L(\mathcal H)$ is called self-adjoint if $A=A^\dagger$ and non-negative definite if $\langle Ax,x\rangle\ge0$ for each $x\in\mathcal H$. The conjugate of an operator $A \in \mathcal L(\mathcal{H})$, denoted by $\overline{A}$, is defined as $\overline{A}x=\overline{(A\overline{x})}$, where $\overline{x}$ denotes the complex conjugate of $x \in \mathcal{H}$.

An operator $A\in\mathcal L(\mathcal H)$ is called compact if it can be written in the form $A=\sum_{j\ge1}s_j(A)\langle e_j,\cdot\rangle f_j$, where $\{e_j:j\ge1\}$ and $\{f_j:j\ge1\}$ are orthonormal sets (not necessarily complete) of $\mathcal H$, $\{s_j(A):j\ge1\}$ are the singular values of $A$ and the series converges in the operator norm. We say that a compact operator $A\in\mathcal L(\mathcal H)$ belongs to the Schatten class of order $p\ge1$ and write $A\in S_p(\mathcal H)$ if $\vert\kern-0.25ex\vert\kern-0.25ex\vert A\vert\kern-0.25ex\vert\kern-0.25ex\vert_p^p=\sum_{j\ge1}s_j^p(A)<\infty$. The Schatten class of order $p\ge1$ is a Banach space with the norm $\vert\kern-0.25ex\vert\kern-0.25ex\vert\cdot\vert\kern-0.25ex\vert\kern-0.25ex\vert_p$. A compact operator $A\in\mathcal L(\mathcal H)$ is called Hilbert-Schmidt if $A\in S_2(\mathcal H)$ and trace class if $A\in S_1(\mathcal H)$. The space of Hilbert-Schmidt operators $S_2(\mathcal H)$ is also a Hilbert space with the inner product given by $\langle A,B\rangle_{\mathrm HS}=\sum_{j\ge1}\langle Ae_j,Be_j\rangle$ for each $A,B\in S_2(\mathcal H)$, where $\{e_j:j\ge1\}$ is an orthonormal basis.

Let $L^2_{\mathbb{C}}([0,1]^k)$ for $k\ge1$ denote the Hilbert space of equivalence classes of square integrable measurable functions $f:[0,1]^k\to\mathbb C$ with the inner product given by
\[
	\langle f,g\rangle
	=\int_{[0,1]^k}f(x)\overline{g(x)}dx
\]
for each $f,g\in L^2_{\mathbb{C}}([0,1]^k)$.% where $\overline x$ denotes the complex conjugate of $x\in\mathbb C$. 
We denote the norm of $L^2_{\mathbb{C}}([0,1]^k)$ by $\|\cdot\|_2$. $L^2_{\mathbb{R}}([0,1]^k)$ for $k\ge1$ denotes the corresponding space of real functions.

An operator $A\in\mathcal L(L^2_{\mathbb{C}}([0,1]^k))$ is Hilbert-Schmidt if and only if there exists a kernel $k_A\in L^2_{\mathbb C}([0,1]^k\times[0,1]^k)$ such that
\[
	Af(x)=\int_{[0,1]^k}k_A(x,y)f(y)dy
\]
almost everywhere in $[0,1]^k$ for each $f\in L^2_{\mathbb{C}}([0,1]^k)$ [see Theorem~6.11 of \citet{weidmann1980}]. In particular, for $A,B\in S_2(L^2_{\mathbb{C}}([0,1]^k))$, we have 
\[
	|\mkern-1.5mu|\mkern-1.5mu|A|\mkern-1.5mu|\mkern-1.5mu|_2^2=\|k_A\|_2^2	%=\int_{[0,1]^k}\int_{[0,1]^k}|k_A(x,y)|^2\;\mathrm dx\mathrm dy 
	\quad \text {and}	\quad \langle A,B\rangle_{\mathrm{HS}}
	=\langle k_A,k_B\rangle 	%=\int_{[0,1]^k}\int_{[0,1]^k}k_A(x,y)\overline{k_B(x,y)}dxdy \quad 
\]
where $k_A, k_b\in L^2_{\mathbb C}([0,1]^k\times[0,1]^k)$. For $f,g\in L^2_{\mathbb{C}}([0,1]^k)$, we define the tensor product $f\otimes g:\mathcal L(L^2_{\mathbb{C}}([0,1]^k))$ as $(f\otimes g)v =  \innprod{v}{g}f$ for all $v \in L^2_{\mathbb{C}}([0,1]^k)$. Since the mapping  $\mathcal{T}:L^2_{\mathbb{C}}([0,1]^k)\otimes L^2_{\mathbb{C}}([0,1]^k) \to S_2(L^2_{\mathbb{C}}([0,1]^k))$ defined by the linear extension of 
$
\mathcal{T}(f \otimes g) = f \otimes \overline{g}
$ is an isometric isomorphism, it defines a Hilbert-Schmidt operator with the kernel in $L^2_{\mathbb C}([0,1]^k\times[0,1]^k)$ given by $(f\otimes g)(\tau,\sigma)=f(\tau)\overline{g}(\sigma)$ for each $\tau,\sigma\in[0,1]^k$.

Finally, we define the following bounded linear mappings. For $A, B, C, \in \mathcal L(\mathcal{H})$, the Kronecker product is defined as $(A \widetilde{\bigotimes} B)C = ACB^{\dagger}$, while the transpose Kronecker product is given by $(A \widetilde{\bigotimes}_{\top} B)C = (A \widetilde{\bigotimes} \overline{B})\overline{C}^{\dagger}$. For $A, B \in S_2(\mathcal{H})$, we shall denote, in analogy to elements $a,b \in \mathcal{H}$, the Hilbert tensor product as $A \bigotimes B$. Further useful properties are provided in \ref{appendix}. 

\subsection{Locally stationary functional time series}
The second order dynamics of weakly stationary time series 
of functional data $\{X_{h}\}_{h\in\mathbb Z}$  can be completely described by the Fourier transform of the sequence of covariance operators, acting on $L^2([0,1], \mathbb C)$, i.e., 
\[
\F_{\omega} =\frac{1}{2\pi} \sum_{h \in \mathbb{Z}} \E\big( (X_{h} - \mu) \otimes (X_{0}-\mu)\big) e^{-\im  \omega h} \quad \omega \in [-\pi, \pi] \tageq \label{eq:statspo}
\]
where $\mu = \E X_0$ denotes the mean function. We will assume our data are centered and hence $\mu = 0$. When stationarity is violated, we can no longer speak of a frequency distribution over all time and hence, if it exists, this object must become time-dependent. To allow for a meaningful definition of this object if stationarity is violated, we consider a triangular array $\{X_{t,T}:1\le t\le T\}_{T\in\mathbb N}$ as a doubly indexed functional time series, where $X_{t,T}$ is a random element with values in $L^2_{\R}([0,1])$ for each $1\le t\le T$ and $T\in\mathbb N$. The processes $\{X_{t,T}:1\le t\le T\}$ are extended on $\znum$ by setting $\XT{t} = \XT{1}$ for $t <1$ and $\XT{t} = \XT{T}$ for $t >T$. Following \citet{vde16}, the sequence of stochastic processes $\{\XT{t}:t\in\mathbb Z\}$ indexed by $T\in\nnum$  is called {\it locally stationary} if for all rescaled times $u\in[0,1]$ there exists an $L^2_{\R}([0,1])$-valued strictly stationary process $\{X^{(u)}_t:t\in\mathbb Z\}$ such that
\begin{equation}\label{Local stationarity}
\Bigl\|\XT{t}-\Xu{t}\Bigr\|_{2}
\leq\big(\big|\tfrac{t}{T}-u\big|+\tfrac{1}{T}\big) \,P_{t,T}^{(u)}\qquad a.s.
\end{equation}
for all $1\leq t\leq T$, where $P_{t,T}^{(u)}$ is a positive real-valued process such that for some $\rho>0$ and $C<\infty$ the process satisfies $\E\big(\big|P_{t,T}^{(u)}\big|^\rho\big)<C$ for all $t$ and $T$ and uniformly in $u\in[0,1]$. If the second-order dynamics are changing gradually over time, the second order dynamics of the stochastic process $\{\XT{t}:t\in\mathbb Z\}_{T\in\mathbb N}$ are then completely described by the {\em time-varying spectral density operator} given by
\begin{align}\label{eq:spdens}
\F_{u,\omega} = \frac{1}{2 \pi}\sum_{h\in \mathbb{Z}}
\E\big(X^{(u)}_{t+h} \otimes X^{(u)}_{t}\big)
e^{-\im  \omega h}
\end{align}
for each $u\in[0,1]$ and $\{X_{t}^{(u)}:t\in\mathbb Z\}$. Under the technical assumptions stated in \autoref{sec3}, this object is a Hilbert-Schmidt operator and we shall denote its kernel function by $f_{u,\omega}\in L^2_\mathbb{C}([0,1]^2)$, which is twice-differentiable with respect to  $u$ and $\omega$. Note that if the process is in fact second-order stationary, then \eqref{eq:spdens} reduces to the form \eqref{eq:statspo} and hence this framework lends itself in a natural way to test for changing dynamics in the second order structure. 

In this paper, we are interested in testing the hypothesis
 \begin{equation}
H_0: \mathcal{F}_{u,\omega} \equiv \mathcal{F}_{\omega} \hspace{0.1 in} \text{a.e. on } [-\pi,\pi]\times[0,1]
\label{h0} 
\end{equation}
\centerline{versus}
 \begin{equation}\label{h1} 
 H_a: \mathcal{F}_{u,\omega} \neq \mathcal{F}_{\omega}, \text{on a subset of } [-\pi,\pi]\times[0,1] \text{ of positive Lebesgue measure},
\end{equation}
where $\mathcal{F}_{\omega}$ is an unknown non-negative definite Hilbert-Schmidt operator for each $\omega \in [-\pi,\pi]$, which does not depend on the rescaled time $u\in[0,1]$.
% Similar in spirit to \citet{bagchi2016}, 
We define the minimum distance
\begin{equation}\label{eq:mindistdef}
m^2 = \min_{\mathcal{G}}\int_{-\pi}^{\pi}\int_0^1\VERT \mathcal{F}_{u,\omega}-\mathcal{G}_{\omega}\VERT_2^2 du d\omega,
\end{equation}
where the minimum is taken over all mappings $\mathcal G:[-\pi,\pi]\to S_2(L^2_\mathbb{C}([0,1]))$. 
 Note that the hypotheses in \eqref{h0}  and \eqref{h1} can be rewritten as 
\begin{equation} \label{hequiv0}
H_0: m^2=0 \hspace{0.3 in} \text{versus} \hspace{0.3 in} H_a: m^2 > 0,
\end{equation}
and a statistical test can be obtained by rejecting 
the null hypothesis $H_0$ for large values of an appropriate estimator of $m^2$. In order to construct 
such an estimator, we  first derive  an alternative
representation of the minimum distance $m^2$.

\begin{lemma}\label{lem:mindist}
The minimum distance $m^2$ defined in \eqref{eq:mindistdef} can be expressed as
\begin{equation}\label{eq:mindistexp}
	m^2
    =\int_{-\pi}^{\pi} \int_0^1 \VERT \mathcal{F}_{u,\omega} - \widetilde{\mathcal F}_{\omega}\VERT_2^2du d\omega
\end{equation}
where the operators $\widetilde{\mathcal F}_\omega$ are defined by
\begin{equation}\label{eq:ftilde}
\widetilde{\mathcal F}_\omega:= \int_0^1 {\mathcal F}_{u,\omega} du
\end{equation}
for each $\omega\in[-\pi,\pi]$. We refer to this operator $\widetilde{\mathcal F}_\omega$ as the time-integrated local spectral density operator as it acts on $L^2([0,1],\mathbb{C})$ such that $\widetilde{\mathcal F}_\omega$ no longer depends on $u\in[0,1]$ for each $\omega\in[-\pi,\pi]$.
\end{lemma}
\begin{proof}
Since $\snorm{\cdot}_2$ is induced by the Hilbert-Schmidt inner product, we  have that
\[
	|\mkern-1.5mu|\mkern-1.5mu|\mathcal F_{u,\omega} - \mathcal G_{\omega}|\mkern-1.5mu|\mkern-1.5mu|_2^2
	=|\mkern-1.5mu|\mkern-1.5mu|\mathcal F_{u,\omega}-\widetilde{\mathcal F}_{\omega}|\mkern-1.5mu|\mkern-1.5mu|_2^2
		+\langle\mathcal F_{u,\omega}-\widetilde{\mathcal F}_{\omega},\widetilde{\mathcal F}_{\omega}-\mathcal G_{\omega}\rangle_\mathrm{HS}
        +\langle\widetilde{\mathcal F}_{\omega}-\mathcal G_{\omega},\mathcal F_{u,\omega}-\widetilde{\mathcal F}_{\omega}\rangle_\mathrm{HS}
		+|\mkern-1.5mu|\mkern-1.5mu|\widetilde{\mathcal F}_{\omega} - \mathcal G_{\omega}|\mkern-1.5mu|\mkern-1.5mu|_2^2.
\]
By linearity and the definition of the Hilbert-Schmidt inner product,
\begin{align*}
\int_0^1\langle\mathcal F_{u,\omega}-\widetilde{\mathcal F}_{\omega},\widetilde{\mathcal F}_{\omega}-\mathcal G_{\omega}\rangle_\mathrm{HS} du =
 \biggl\langle\mathcal \int_0^1 \mathcal F_{u,\omega}du  -\widetilde{\mathcal F}_{\omega},\widetilde{\mathcal F}_{\omega}-\mathcal G_{\omega}\biggr\rangle_\mathrm{HS}	&=0.
\end{align*}
A similar argument shows that $\int_{0}^1\langle \widetilde{\mathcal F}_{\omega}-\mathcal G_{\omega},\mathcal F_{u,\omega}-\widetilde{\mathcal F}_{\omega}\rangle_\mathrm{HS} du=0$.
Hence,
\[
	m^2
    =\int_{-\pi}^{\pi} \int_0^1 \VERT \mathcal{F}_{u,\omega} - \widetilde{\mathcal F}_{\omega}\VERT_2^2du d\omega + \min_{\mathcal G}\int_{-\pi}^\pi \VERT \widetilde{\mathcal{F}}_{\omega} - \mathcal{G}_{\omega}\VERT_2^2d\omega
\]
and the infimum of the second term is achieved at $\mathcal{G}_{\omega} \equiv\widetilde{\mathcal{F}}_{\omega}$. The proof is complete.
\end{proof}

\noindent 
Using the definition of the Hilbert-Schmidt norm, we can rewrite expression~\eqref{eq:mindistexp} in terms of $\mathcal F_{u,\omega}$
\begin{equation}\label{eq:M^2}
m^2 = \int_{-\pi}^{\pi}\int_0^1\snorm{\mathcal F_{u,\omega}}_2^2dud\omega - \int_{-\pi}^{\pi} \snorm{ \widetilde{\mathcal{F}}_{\omega}}^2_2 d\omega,
\end{equation}
where $\widetilde{\mathcal{F}}_{\omega}$ is given by \eqref{eq:ftilde}. The two terms in \eqref{eq:M^2} can now be easily estimated from the available data $\{X_{t,T}:1\le t\le T\}$
by sums of periodogram operators.

In order to estimate the two integrals in \eqref{eq:M^2} we split the sample into $M$ blocks with $N$ elements inside each of these blocks so that $T=MN=M(T)N(T)$ for each  $T\in\mathbb N$, where $M, N \in\mathbb N$ and $N$ is an even number. 
$M$ and $N$ will correspond to the number of terms used in a Riemann sum approximating the integrals in \eqref{eq:M^2} with respect to $du$ and $d\omega$ and therefore they have to be reasonable large. 
The number of elements in the blocks grows faster than the number of blocks, but slower than the cube number of blocks.
 The choice of the number of blocks is carefully discussed in \autoref{subsec:mse} and an empirical investigation can be found in  \autoref{sec4}.
 Throughout this paper, we make the following assumption for the asymptotic analysis. 
\begin{assumption}\label{ratesNM}
 $M\to\infty$, $N\to\infty$ as $T\to\infty$, such that 	
\[N/M\to\infty
    \quad\text{and}\quad
    N/M^3\to0 \]
\end{assumption}

For $u\in[0,1]$, $\omega\in[-\pi,\pi]$ and $N\ge1$, the functional discrete Fourier transform (fDFT) evaluated around time $u$ is defined as a random function with values in $L^2_{\mathbb C}([0,1])$ given by
\begin{equation}\label{eq:fDFT}
	D_N^{u,\omega}
    :=\frac1{\sqrt{2\pi N}}\sum_{s=0}^{N-1}X_{\lfloor uT \rfloor - N/2 +s+1,T}e^{-\im  \omega s}.
\end{equation}
The periodogram tensor is then defined by
\begin{equation}\label{eq:per}
I_N^{u,\omega} := D_N^{u,\omega}\otimes {D_N^{u,\omega}}.
\end{equation}

Let $\omega_k=2\pi k/N$ for $k=1,\ldots,N$ and $u_j=(N(j-1)+N/2)/T$ for $j=1,2,\ldots,M$ be the midpoint of each block. Observe that only the $j$-th block of the sample determines the value of $I_N^{u_j,\omega_k}$ for each $k=1,\ldots,N$. {We estimate the two terms in \eqref{eq:M^2} by 
\begin{align}
\label{eq:F_1op}\hat{F}_{1,T}&:=\frac{1}{T}\sum_{k=1}^{\lfloor N/2 \rfloor} \sum_{j=1}^M \langle I_N^{u_j,\omega_k} , I_N^{u_j,\omega_{k-1}}\rangle_{HS},
\end{align}
(note that $\hat F_{1,T}$ is real-valued for each $T\in\mathbb N$ since $\langle I_N^{u,\lambda},I_N^{u,\omega}\rangle_{HS}=|\langle D_N^{u,\lambda},D_N^{u,\omega}\rangle|^2$)
and
\begin{align}
\label{eq:F_2op}\hat{F}_{2,T}&:=\frac1N\sum_{k=1}^{\lfloor N/2 \rfloor} \bigsnorm{\frac1M\sum_{j=1}^MI_N^{u_j,\omega_k}}_2^2,
\end{align}
respectively.  It will be shown later that the estimation of $\int_{-\pi}^{\pi} \snorm{ \widetilde{\mathcal{F}}_{\omega}}^2_2 d\omega$ by \eqref{eq:F_2op} introduces a bias term
\[ \tageq \label{eq:bias}
B_{N,T}= \frac{N}{T}\int_{-\pi}^\pi \int_0^1   \Tr\Big(\mathcal{F}_{u,\omega} \widetilde{\bigotimes} \mathcal{F}_{u,\omega} \Big)du d\omega.
\]
As this term  is nonvanishing in a $\sqrt{T}$-consistent estimator under Assumption \ref{ratesNM} it has to be taken into account. We therefore define the estimator of the minimum distance $m^2$ in \eqref{eq:M^2} as 
\begin{equation} \label{mhat}
	\widehat m_T
	=4\pi(\hat{F}_{1,T}-\hat{F}_{2,T}+\hat{B}_{N,T})~,
\end{equation}
where
\begin{align} \label{biasest}
\hat{B}_{N,T}&  =\frac{1}{NM}\sum_{k=1}^{\lfloor N/2 \rfloor} \sum_{j=1}^M \Tr\Big({I}^{u_j,\omega_k} \widetilde{\bigotimes} {I}^{u_j,\omega_{k-1}} \Big)
=\frac{1}{NM}\sum_{k=1}^{\lfloor N/2 \rfloor} \sum_{j=1}^M \|{\fdft{j}{k}}\|^2_2 \|\fdft{j}{k-1}\|^2_2  
\end{align}
We prove in \autoref{sec5} (\autoref{cor:consBNT}) 
that under the conditions of \autoref{thm:dist}
\[
\sqrt{T} \big(\hat{B}_{N,T} -  B_{N,T}\big) \overset{p}{\to} 0. 
\]
The bias correction therefore does not affect the asymptotic distribution of the test statistic.
}

As in the case of a real-valued time series the periodogram tensor defined by \eqref{eq:per} is not a consistent estimator. However, the estimators $\hat F_{1,T}$ and $\hat F_{2,T}$ are consistent for the quantities appearing in the measure of stationarity defined in \eqref{eq:M^2} as they are  obtained  by averaging 
periodogram tensors with respect to different Fourier frequencies.
These heuristic arguments will be made more precise in the following section, where we state our main asymptotic results.

\section{Asymptotic normality and  statistical applications}  \label{sec3}
\def\theequation{3.\arabic{equation}}
\setcounter{equation}{0}

In this section we establish asymptotic normality of an appropriately standardized version of the statistic $\widehat m_T $ defined in \eqref{mhat} and as a by-product its consistency for estimating 
the measure of stationarity $m^2$.  For this purpose
the functional process $\{\,\XT{t}\colon t \in \znum\}_{T\in\nnum}$ is assumed to satisfy the following set of conditions.
%, which are a slight adjustment of those in \citet{avd16_main}.
\begin{assumption}\label{cumglsp}
{\em Assume $\{\,\XT{t}\colon t \in \znum\}_{T\in\nnum}$ is locally stationary 
zero-mean stochastic process as introduced in \autoref{sec2} and let $\cumbp{k}{k} : L^2([0,1]^{\lfloor k/2\rfloor}) \to  L^2([0,1]^{\lfloor \frac{k+1}{2}\rfloor})$ be a positive operator  
% a positive sequence in $L^2([0,1]^k,\rnum)$ 
independent of $T$ such that, for all $j=1,\ldots,k-1$ and some $\ell\in\nnum$,
\begin{align} 
\label{eq:kapmix}
\sum_{t_1,\ldots,t_{k-1} \in \mathbb{Z}} (1+|t_j|^\ell)\snorm{\cumbp{k}{k}}_1 <\infty.
\end{align}
Let us denote
\begin{equation}\label{eq:repstatap}
Y^{(T)}_{t}=\XT{t}-X_t^{(t/T)}
\qquad\text{and}\qquad
Y_t^{(u,v)}=\frac{\Xu{t}- X^{(v)}_{t}}{(u-v)}
\end{equation}
for $T\in\mathbb N$, $1\le t\le T$ and $u,v\in[0,1]$ such that $u\ne v$. Suppose furthermore that $k$-th order joint cumulants satisfy
\begin{enumerate}[(i)]\itemsep-.3ex
\item$\|\cm(\XT{t_1},\ldots,\XT{t_{k-1}},Y^{(T)}_{t_k}) \|_2 \le \frac{1}{T}\snorm{\kappa_{k;t_1-t_k,\ldots,t_{k-1}-t_k}}_1 $,
\item$\|\cm(\Xuu{u_1}{1},\ldots,\Xuu{u_{k-1}}{k-1},Y_{t_{k}}^{(u_k,v)}) \|_2 \le \snorm{\kappa_{k;t_1-t_k,\ldots,t_{k-1}-t_k}}_1 $,
\item$\sup_u \|\cm(\Xuu{u}{1},\ldots,\Xuu{u}{k-1},\Xuu{u}{k}) \|_2 \le \snorm{\kappa_{k;t_1-t_k,\ldots,t_{k-1}-t_k}}_1$,
\item$\sup_u \|\frac{\partial^\ell}{\partial u^\ell} \cm(\Xuu{u}{1},\ldots,\Xuu{u}{k-1},\Xuu{u}{k}) \|_2 \le \snorm{\kappa_{k;t_1-t_k,\ldots,t_{k-1}-t_k}}_1$.
\end{enumerate}}
\end{assumption}
Note that for fixed $u_0$, the process $\{\Xuu{u_0}{}\colon t\in\mathbb{Z}\}$ is strictly stationary and thus the results of \citet{vde16} imply that the {\em local $k$-th order cumulant spectral kernel} 
\begin{align}\label{eq:tvsdoker}
\floc{u_0}{k} = \frac{1}{(2 \pi)^{k-1}}\sum_{t_1,\ldots,t_{k-1} \in \mathbb{Z}} \cumloc{u_0}{k}e^{-\im \sum_{l=1}^{k-1} \omega_{l} t_l}
\end{align}
is well-defined in an $L^2$-sense, where $\omega_1,\ldots, \omega_{k-1} \in [-\pi, \pi]$ and
\begin{align}
\label{eq:cumglsp}
\cumloc{u_0}{k}=\cm\big(\Xuu{u_0}{1}(\tau_1),\ldots,\Xuu{u_0}{k-1}(\tau_{k-1}),X_0^{(u_0)}(\tau_k)\big)
\end{align}
is the corresponding local cumulant kernel of order $k$ at time $u_0$. We shall denote the corresponding operators acting on $L^2([0,1]^{k}, \mathbb{C})$ by $\F_{u_0, \omega_{k_1}, \ldots, \omega_{2k-1}}$ and $C_{u_0, \omega_{k_1}, \ldots, \omega_{2k-1}}$, respectively. For $k=2$ we obtain time-varying spectral density kernel $ f_{u,\omega}(\tau_1,\tau_2) $ - the kernel of the operator defined in \eqref{eq:spdens} - which is uniquely defined by the triangular array and  twice-differentiable with respect to  $u$ and $\omega$ if assumption (iv) holds for  $\ell=2$ (see also Lemma A.3 of \citet{avd16_main} for more details).

 The following result establishes the asymptotic normality of $\widehat m_T $  (appropriately 
standardized). The proof is postponed to \autoref{sec5}. 
\begin{thm}\label{thm:dist} 
Suppose that \autoref{ratesNM} and \autoref{cumglsp} hold. Then
\[
	\sqrt T(\hat m_T -  m^2)\xrightarrow d N(0,\nu^2)\quad\text{as}\quad T\to\infty,
\]
where the expression for the asymptotic variance $\nu^2$ is relegated to Section \ref{sec5}.
\end{thm}
Under the null, the statistic has a very succinct form 
\begin{Corollary} \label{thm:disth0}
Suppose that \autoref{ratesNM} and \autoref{cumglsp} hold. Then, under the
null hypothesis $H_0$ we have 
\[
	\sqrt T \hat m_T \xrightarrow d N(0,\nu^2_{H_0})\quad\text{as}\quad T\to\infty,
\]
where the asymptotic variance $v_{H_0}^2$ is given by
\begin{equation}\label{eq:varunderhyp}
	\nu^2_{H_0}=4\pi \int_{-\pi}^{\pi} \snorm{\widetilde{\F}_\omega}_2^4d\omega.
\end{equation}
\end{Corollary}
\noindent
Observing the equivalent representation of the hypotheses in \eqref{hequiv0} it is reasonable 
to  reject the null hypotheses \eqref{h0} 
of a stationary functional process whenever\begin{eqnarray} \label{testwn}
\widehat{m}_T~> ~  \frac{\widehat{v}_{H_0}}{\sqrt{T}} u_{1-\alpha}~,
 \end{eqnarray}
 where $u_{1-\alpha}$ denotes the $({1-\alpha})$-quantile of the standard normal distribution and $\widehat{v}_{H_0}^2$ is  an appropriate estimator of the asymptotic variance under the null hypothesis  given in \eqref{eq:varunderhyp}.  The  asymptotic variance  $v_{H_0}^2$ can be estimated by the statistic
 \begin{equation} 
 \label{varesth0} 
 \hat v_{H_0}^2 = \frac{16\pi^2}{N}\sum_{k=1}^{\lfloor N/2 \rfloor} \left[\frac{1}{M}\sum_{j=1}^M\innprod{I_N^{u_j,\omega_k}}{I_N^{u_j,\omega_{k-1}}}_{HS}\right]^2.
% \left[\int_{[0,1]^2}\left[\frac{1}{M}\sum_{j=1}^M I_N^{u_j,\omega_k}(\tau,\sigma)\overline{I_N^{u_j,\omega_{k-1}}(\tau,\sigma)}\right]d\tau d\sigma\right]^2.
 \end{equation}
 \autoref{thm:disth0} and the following result show that the test defined by \eqref{testwn} is an asymptotic level $\alpha$ test.
 The proof can be found in Section \ref{subsec:estvar}.
 
 \begin{lemma} \label{consvarest} Under the assumptions of Theorem \ref{thm:disth0}, the estimator defined in \eqref{varesth0}
 is consistent, that is
 $$
  \hat v_{H_0}^2 \to   v_{H_0}^2
 $$
 in probability as $T \to \infty$.
 \end{lemma}

\subsection{The choice of $M$ and $N$}\label{subsec:mse}

We provide some heuristic arguments on how to choose the number of blocks $M$ and the number of elements in the blocks $N$. Since we assume that $T=MN$, the choice of $M$ determines the value of $N$, and vice versa. Our test is based on the estimator of the distance $m^2$ defined by \eqref{eq:mindistdef}. One way to choose the values of $M$ and $N$ is to choose the values of $M$ and $N$ that minimize the leading terms in the asymptotic expansion of the mean squared error (MSE) of the estimator of~$m^2$.
We have that
\[
	\operatorname{MSE}(\hat m_T)
    =\operatorname{Var}\hat m_T+|\operatorname \E\hat m_T-m^2|^2 \tageq \label{eq:msedec}
    \]
with $\operatorname{Var}\hat m_T=(4\pi)^2\{\operatorname{Var}\hat F_{1,T}+\operatorname{Var}\hat F_{2,T}-2\operatorname{Cov}[\hat F_{1,T},\hat F_{2,T}]\}$ and $\operatorname \E\hat m_T=4\pi(\operatorname \E\hat F_{1,T}-\operatorname \E\hat F_{2,T})$. 
Note that we ignore the estimate $\hat{B}_{N,T}$ of the bias  defined by \eqref{biasest} as it is of lower order in \eqref{eq:msedec}.
The asymptotic expressions of the terms $\operatorname{Var}\hat F_{1,T}$, $\operatorname{Var}\hat F_{2,T}$ and $\operatorname{Cov}[\hat F_{1,T},\hat F_{2,T}]$ are given in Section~\ref{sec:cov} 
of the Appendix and the asymptotic expressions of the terms $\operatorname \E\hat F_{1,T}$ and $\operatorname \E\hat F_{2,T}$ are given in Section~\ref{sec5}.  For the moment, we assume Gaussianity to avoid dealing with the fourth order terms. The leading terms of the asymptotic expression of the MSE are double Riemann sums. $M$ and $N$ determine the error that we make by approximating double integrals by double Riemann sums.  Suppose that $g:[0,1]\times[0,\pi]\to\mathbb R$ is a Riemann integrable function. Using the error bounds for the midpoint and the right endpoint approximations of the integrals, we obtain
\begin{multline}\label{eq:doubleriemann}
	\Bigl|\frac1N\sum_{k=1}^{\lfloor N/2\rfloor}\frac1M\sum_{j=1}^M g(u_j,\omega_k)
	-\frac1{2\pi}\int_0^\pi\int_0^1g(u,\omega)dud\omega\Bigr|\\
	\le\frac1{24M^2}\cdot\frac1N\sum_{k=1}^{\lfloor N/2\rfloor}\max_{u\in[0,1]}|g_u''(u,\omega_k)|
	+\frac{\pi^2}{2N}\int_0^1\max_{\omega\in[0,\pi]}|g_\omega'(u,\omega)|du,
\end{multline}
where $u_j=(N(j-1)+N/2)/T$ for $1\le j\le M$ and $\omega_k=2\pi k/N$ for $1\le k\le\lfloor N/2\rfloor$. We do not give the complete bound of the MSE, we only explain the idea behind the bound. One of the terms in the expression of $T\operatorname{Var}\hat F_{1,T}$ is given by
\[
	R_{NM}
	=\frac{1}{T}\sum_{k=1}^{\lfloor N/2 \rfloor} \sum_{j=1}^M \innprod{\F_{u_{j},\omega_{k}}}{\F_{u_{j},\omega_{k-1}}}_{HS}\innprod{ \F_{u_{j},-\omega_{k}}}{ \F_{u_{j},-\omega_{k-1}}}_{HS}.
  \]
We have that
\begin{equation}\label{eq:riemannbound}
	R_{NM}\le
    \Bigl|R_{NM}
    -\frac1{2\pi}\int_0^\pi\int_0^1\opnorm{\mathcal F_{u,\omega}}_2^4dud\omega\Bigr|
    +\frac1{2\pi}\int_0^\pi\int_0^1\opnorm{\mathcal F_{u,\omega}}_2^4dud\omega.
\end{equation}
The second term in \eqref{eq:riemannbound} does not depend on the choice of $M$ and $N$. We use inequality \eqref{eq:doubleriemann} to bound the first term in \eqref{eq:riemannbound}. Provided that the integral is finite and the Riemann sum converges in \eqref{eq:doubleriemann}
and using similar arguments for the other terms in the mean squared error, we see that we need to minimize the expression $C_1/M^2+C_2/N$ over all possible values of $M$ and $N$, where $C_1$ and $C_2$ are two positive constants that are unknown since they depend on the time-varying spectral density operator. The right hand side of \eqref{eq:doubleriemann} is minimized when $$M=(2C_1/C_2)^{1/3}\cdot T^{1/3} \mbox{ and } N=(2C_1/C_2)^{-1/3}\cdot T^{2/3}.$$
Unfortunately, since $C_1$ and $C_2$ are unknown, we cannot determine the optimal values of $M$ and $N$. However, this suggests $M\approx T^{1/3}$ and $N\approx T^{2/3}$ might be a reasonable choice provided that $(2C_1/C_2)^{1/3}\approx1$.
We provide empirical evidence of this rule in \autoref{sec4}.

\noindent

\subsection{Statistical applications}

As Theorem \ref{thm:dist} provides the asymptotic distribution at any point of the alternative it has several 
important applications, which are briefly mentioned here.
\begin{itemize}
\item[(a)] The probability of a type II of the test \eqref{testwn} can be calculated approximately by the  formula
   \begin{eqnarray} \label{power}
 \mathbb{P} \big(\widehat{m}_T \leq    \widehat{\nu}_{H_0} u_{1-\alpha}/\sqrt T \,\big) ~  \approx ~ 
 \Phi \Bigl (\frac {\nu_{H_0}}{\nu} \ u_{1- \alpha}  -  \sqrt{T}\ \frac {m^2} {\nu} \Bigr
),
 \end{eqnarray}
where $\nu _{H_0}^2$ and $\nu^2$ are defined in Theorem \ref{thm:disth0} and \ref{thm:dist}, respectively, and  $\Phi$ is the distribution function of the standard normal distribution.
\item[(b)]  An asymptotic confidence  interval for the measure of stationarity $m^2$ is given by 
\begin{align} \label{confid}
\Big[ \max \Big\{ 0, \widehat{m}_T - \frac{\hat \nu _{H_1}}{\sqrt{T}} u_{1-\alpha/2}
\Big\}, \widehat{m}_T+\frac{\hat \nu_{H_1}}{\sqrt{T}}u_{1-\alpha/2}
\Big]~,
\end{align}
where $ \hat \nu_{H_1}^2$ denotes an estimator of the variance  in  Theorem \ref{thm:dist}. 
\item[(c)] Similarly, one can use \autoref{thm:dist} 
to  test  for {\it similarity to stationarity}  considering the hypotheses
\begin{eqnarray} 
H_\Delta :  m^2 \geq \Delta  ~~ & \text{ vs. } &  ~~~K_\Delta :  m^2 <  \Delta  ~, \label{hequiv}
\end{eqnarray} 
where $\Delta $ is a pre-specified constant
such that for a value of $m^2$  smaller than $\Delta$ the experimenter defines  the second order properties  to be similar to stationarity.  For example,  if the  functional time series deviates only slightly from second order stationarity, it is often  reasonable  to work under the assumption of stationarity as many procedures are robust against small deviations
from this assumption and procedures specifically adapted to non-stationarity usually have a larger variability. \\
An 
 asymptotic  level $\alpha$  test for these hypotheses  is obtained by rejecting
the null hypothesis, whenever 
\begin{eqnarray}  \label{prectest}
\widehat{m}_T-\Delta<\frac{ \hat \nu _{H_1} }{\sqrt{T}} u_{\alpha} \: .
\end{eqnarray}
Note that this test allows to decide  for ``approximate second order stationarity''  with a controlled type I error.
It follows from Theorem \ref{thm:dist}  and  a straightforward calculation that 
\begin{align} \label{asymeqivtest}
\lim_{T\to \infty} \mathbb{P}  \Big(  \widehat{m}_T - \Delta  < \frac{ \hat \nu _{H_1} }{\sqrt{T}} u_{\alpha} \Big)~=~
\begin{cases}
0 &  \mbox{ if }  m^2 > \Delta \\
\alpha &  \mbox{ if } m^2 = \Delta  \\
1 &  \mbox{ if }  m^2 <   \Delta
\end{cases}
~, 
\end{align}
which means that the test \eqref{prectest} is a consistent and  asymptotic level $\alpha$ test for the hypotheses \eqref{hequiv}.   For the hypotheses of a {\it relevant difference} $H :  m^2 \leq \Delta  ~  \text{ vs. }  ~K : > \Delta $ a corrsponding asymptotic level $\alpha$ test can be constructed similarly and the details are omitted.
\end{itemize}

\section{Proof of \autoref{thm:dist}}
\label{sec5}
\def\theequation{4.\arabic{equation}}
\setcounter{equation}{0}

In this section we explain the main steps in the proof of \autoref{thm:dist}. Let us recall that $T = NM$, where $N$ defines the resolution in frequency of the local fDFT and $M$ controls the number of nonoverlapping local fDFT's. To establish that $\sqrt T(\hat m_T -  m^2)\xrightarrow d N(0,\nu^2)$ as $T\to\infty$ with $v^2$ given by~\eqref{eq:asvar}, we show that
\begin{align}
\label{eq:convofe}\sqrt T[\operatorname \E\hat m_T-m^2]\to0,\\
\label{eq:convofe2}T\operatorname{Var}\hat m_T\to v^2,
\end{align}
and
\begin{equation}\label{eq:cum>2}
T^{n/2}\operatorname{cum}_n(\hat m_T)\to0
\end{equation}
for $n>2$ as $T\to\infty$.

Recall first that $\widehat m_T=4\pi(\hat{F}_{1,T}-\hat{F}_{2,T}+\hat{B}_{N,T})$ and that \autoref{cor:consBNT} implies the bias correction $\hat{B}_{N,T}$ does not affect the asymptotic distribution of $\widehat m_T$. Therefore, the distributional properties of $\sqrt T(\hat m_T-m^2)$ will follow from the joint distributional structure of $\hat{F}_{1,T}$ and $\hat{F}_{2,T}$. In particular, multilinearity of cumulants implies that we have
\begin{equation*}
	\operatorname{cum}_n(\hat m _T)
	=(4\pi)^n\operatorname{cum}_n(\hat F_{1,T}-\hat F_{2,T})
	=(4\pi)^n\sum_{x=0}^n(-1)^x{n\choose x}\operatorname{cum}_{n-x,x}(\hat F_{1,T},\hat F_{2,T}), \tageq\label{eq:cumnxF1F2}
\end{equation*}
where $\operatorname{cum}_{n-x,x}(\hat F_{1,T},\hat F_{2,T})$ denotes the joint cumulant
\[
\operatorname{cum}(\underbrace{\hat F_{1,T},\ldots,\hat F_{1,T}}_{x\ \text{times}},\underbrace{\hat F_{2,T},\ldots,\hat F_{2,T}}_{n-x\ \text{times}})
\]
for $n,x\ge0$. 

The first two moments \eqref{eq:convofe}-\eqref{eq:convofe2} can be determined by the cumulant structure of order $n=1$ and $n=2$ of $\hat{F}_{1,T}$ and $\hat{F}_{2,T}$, respectively, while \eqref{eq:cum>2} will follow from showing that 
\[
T^{n/2}\operatorname{cum}_{n-x,x}(\hat F_{1,T},\hat F_{2,T})\to 0
\]
as $T\to\infty$ for each $n>2$ and $0\le x\le n$. Detailed derivations of technical propositions together with additional background material on cumulant tensors is given in the appendix. 

The main ingredient to our proof is the following result which allows us to re-express the cumulants of $\hat{F}_{1,T}$ and $\hat{F}_{2,T}$, which consists of Hilbert-Schmidt inner products of local periodogram tensors, into the trace of cumulants of simple tensors of the local functional DFT's.

\begin{thm} \label{lem:cumHSinprod}
Let $\E\snorm{I_N^{u,\omega}}^{2n}_2 <\infty$ for some $n \in \mathbb{N}$ uniformly in $u$ and $\omega$. Then
\begin{align*}
 \operatorname{cum} &\Big(\innprod{ I_N^{u_{j_1},\omega_{k_1}}}{I_N^{u_{j_2},\omega_{{k_2}}}}_{HS}, \ldots, \innprod{ I_N^{u_{j_{2n-1}},\omega_{k_{2n-1}}}}{I_N^{u_{j_{2n}},\omega_{{k_{2n}}}}}_{HS}\Big) 
\\&\Tr\Big(\sum_{\boldsymbol{P} = P_1 \cup \ldots \cup P_G}S_{\boldsymbol{P}}\Big(\otimes_{g=1}^{G}
 \operatorname{cum}\big(\fdft{p}{p}|p\in P_g\Big)  \Big), 
\end{align*}
where the summation is over all indecomposable partitions $\boldsymbol{P} = P_1 \cup \ldots \cup P_G$ of the array 
\[\begin{matrix} 
(1,1) & (1,2) & (1,3) & (1,4) \\
(2,1) & (2,2) & (2,3) & (2,4) \\
\vdots & & &\vdots  \\
\vdots & & &\vdots  \\
(n,1) & (n,2) & (n,3) & (n,4) \\
\end{matrix} \tageq \label{tab:highcumF1}\]
where $p=(l,m)$ and $k_{p} = (-1)^{m} k_{2l-\delta{\{m \in \{1,2\}\}}}$ and $j_p= j_{2l-\delta{\{m \in \{1,2\}\}}}$ for $l \in\{1,\ldots, n\}$ and $m \in \{1,2,3,4\}$. 
Here the function $\delta_{\{A\}}$ equals 1 if event $A$ occurs and $0$ otherwise.
\end{thm}
\begin{proof}[Proof of Theorem \ref{lem:cumHSinprod}] 
First note that a sufficient condition for $\E\snorm{I_N^{u,\omega}}^p_2 <\infty$ to exist is $\E \|D_N^{u,\omega}\|^{2p}_2<\infty$ or, in terms of moments of $X$, $\E\|\XT{t}\|^{2p}_2<\infty$ for each $T\ge1$, $1\le t\le T$ and hence by \autoref{cumglsp}
\begin{align*}
\cm_n(\innprod{ I_N^{u_{j_1},\omega_{k_1}}}{I_N^{u_{j_2},\omega_{{k_2}}}}_{HS})  \le \prod_{l=1}^n \sum_{k_l=1}^{\lfloor N/2 \rfloor} \sum_{j_l=1}^M \sqrt{\E\snorm{ I_N^{u_{j_l},\omega_{k_l}}}^2_2} \sqrt{\E\snorm{ I_N^{u_{j_2},\omega_{k_2}}}^2_2} <\infty.
\end{align*}
The definition of scalar cumulants and a basis expansion yield 
\begin{align*}
\cm &\Big(\innprod{ I_N^{u_{j_1},\omega_{k_1}}}{I_N^{u_{j_2},\omega_{{k_2}}}}_{HS}, \ldots, \innprod{ I_N^{u_{j_{2n-1}},\omega_{k_{2n-1}}}}{I_N^{u_{j_{2n}},\omega_{{k_{2n}}}}}_{HS}\Big) 
\\& 
=\sum_{\nu=(\nu_1,\ldots,\nu_G)}(-1)^{G-1}\,(G-1)!\,\prod_{g=1}^G \E \prod_{(l, m) \in \nu_g} \innprod{ I_N^{u_{j_l},\omega_{k_l}}}{I_N^{u_{j_m},\omega_{{k_m}}}}_{HS}
\\& =\sum_{\nu=(\nu_1,\ldots,\nu_G)}(-1)^{G-1}\,(G-1)!\,\prod_{g=1}^G \E \prod_{(l, m) \in \nu_g}  \Tr \Big((\fdft{l}{l} \otimes \fdftc{l}{l} )\bigotimes  (\fdft{m}{m} \otimes \fdftc{m}{m} )\Big)
\\& =\sum_{\nu=(\nu_1,\ldots,\nu_G)}(-1)^{G-1}\,(G-1)!\,\prod_{g=1}^G \E \prod_{(l, m) \in \nu_g}  \Tr \Big((\fdft{l}{l} \otimes \fdftc{l}{l} \otimes  \fdftc{m}{m} \otimes \fdft{m}{m} \Big),
\end{align*}
where the summation extends over all unordered partitions $(\nu_1,\ldots,\nu_G), G=1,\ldots, n$. The fact that the expectation operator commutes with the trace operator together with property \ref{tensorprop}.4 implies this can be written as
\begin{align*}
%&=\Tr\Big(\sum_{\nu=(\nu_1,\ldots,\nu_G)}(-1)^{G-1}\,(G-1)!\,\widetilde{\bigotimes}_{g=1}^G  \E\Big[ \widetilde{\bigotimes}_{(l, m) \in \nu_g} (I_N^{\dagger,u_{j_l},\omega_{k_l}}\otimes I_N^{u_{j_{m}},\omega_{{k_m}}})\Big]\Big)
%\\&
 \Tr\Big(\sum_{\nu=(\nu_1,\ldots,\nu_G)}(-1)^{G-1}\,(G-1)!\,\widetilde{\bigotimes}_{g=1}^G \E\Big[ \widetilde{\bigotimes}_{(l,m) \in \nu_g} ((\fdft{l}{l} \otimes \fdft{l}{l} )\otimes  (\fdft{m}{m} \otimes \fdft{m}{m} ))\Big]\Big)
\end{align*}
The product theorem for cumulant tensors (equation \eqref{prodcumthm}  in the Appendix) then yields the above equals
\begin{align*}
=\Tr\Big(\sum_{\boldsymbol{P} = P_1 \cup \ldots \cup P_G}S_{\boldsymbol{P}}\Big(\otimes_{g=1}^{G}
\cm\big(\fdft{p}{p}|p\in \nu_g\Big)  \Big).  
\end{align*}
Here the summation is over all indecomposable partitions $\boldsymbol{P} = P_1 \cup \ldots \cup P_G$ of the array 
\[\begin{matrix}
(1,1) & (1,2) & (1,3) & (1,4) \\
(2,1) & (2,2) & (2,3) & (2,4) \\
\vdots & & &\vdots  \\
\vdots & & &\vdots  \\
(n,1) & (n,2) & (n,3) & (n,4) \\  
\end{matrix} \tageq \label{tab:arraythm41} \]
 where $S_{\boldsymbol{P}}$ denote the permutation operator on $\otimes_{i=1}^{4n} L^2([0,1],\mathbb{C})$ that maps the components of the tensor according to the permutation $(1,\ldots, 4n) \mapsto \boldsymbol{P}$ and where $p=(l,m)$, $k_{p} = (-1)^{m} k_{2l-\delta{\{m \in \{1,2\}\}}}$ and $j_p= j_{2l-\delta{\{m \in \{1,2\}\}}}$ for $l \in\{1,\ldots, n\}$ and $m \in \{1,2,3,4\}$. Here the function $\delta_{\{A\}}$ equals 1 if event $A$ occurs and $0$ otherwise.
\end{proof}

The following lemma shows that the cumulant tensor of the local fDFT's evaluated at the same midpoint $u_i$ and on the manifold $\sum_{j=1}^{k} \omega_{j} \equiv 0 \mod 2\pi$ can in turn be expressed in terms of higher order cumulant spectral operators. 

\begin{lemma} \label{lem:cumfixu}
If \autoref{cumglsp} is satisfied and $\sum_{j=1}^{k} \omega_{j} \equiv 0 \mod 2\pi$ then
\begin{align*} 
\bigsnorm{\cm\left(D_N^{u_i,\omega_1},\dots,D_N^{u_i,\omega_k}\right)-\frac{(2\pi)^{1-k/2}}{N^{k/2-1}}\mathcal{F}_{u_i,\omega_1,\dots,\omega_{k-1}}}_1 = O\left(N^{-k/2} \times \frac{N}{M^2}\right).
\end{align*}
\end{lemma}

When evaluated off the manifold, i.e., $\sum_{j=1}^{k} \omega_{j} \ne 0 \mod 2\pi$, the above cumulant is of lower order (see \autoref{cor:cumbound}). 
A direct consequence of the proof of \autoref{lem:cumfixu} is the following corollary
\begin{Corollary} \label{cor:cumbound}
We have for any $p \ge 1$
\begin{align}\bigsnorm{\cm\left(D_N^{u_1,\omega_1},\dots,D_N^{u_k,\omega_k}\right)}_p = O\left(N^{1-k/2}\right) \label{eq:boundonmani}\end{align}
uniformly in $\omega_1,\ldots,\omega_k$ and $u_1,\ldots, u_k$. Moreover, if $~\sum_{j=1}^{k}\omega_j~ \ne 0\mod 2\pi$ then
 \begin{align}\bigsnorm{\cm\left(D_N^{u_1,\omega_1},\dots,D_N^{u_k,\omega_k}\right)}_p =O\left(N^{-k/2}\right).\label{eq:boundoffmani}\end{align}
\end{Corollary} 

Additionally, when the local fDFT's are evaluated on different midpoints then we have the following lemma. 
\begin{lemma}\label{lem:cumdifu}
If \autoref{cumglsp} is satisfied and $|{j_1}-{j_2}|>1$ for some midpoints $u_{j_1}$ and $u_{j_2}$ then 
\begin{align*} & \bigsnorm{\cm\left(D_N^{u_{j_1},\omega_1},\dots,D_N^{u_{j_k},\omega_k}\right)}_1 = O\left(N^{-k/2} M^{-1} \right)
\end{align*}
uniformly in $\omega_1,\ldots,\omega_k$.
\end{lemma}

Proofs of these statements are relegated to Section \ref{sec:propcumlfdft} of the Appendix, where the properties of cumulant tensors of the local fDFTs are investigated in more detail. 
Using these results,  the assertions \eqref{eq:convofe}-\eqref{eq:cum>2} can now be established. 
%\subsection{Expectation} \label{sec52}
More specifically, for \eqref{eq:convofe}, \autoref{lem:cumHSinprod} implies we can write
\begin{align*}
\E \hat{F}_{1,T} &	=\frac{1}{T}\sum_{k=1}^{\lfloor N/2 \rfloor} \sum_{j=1}^M \Tr \Big( \E \big[\fdft{1}{1} \otimes \fdftc{1}{1} \otimes \fdftcl{1}{1}\otimes\fdftl{1}{1} \big]\Big).
\end{align*}
Expressing this expectation in cumulant tensors, we get
\begin{align*}
\E \hat{F}_{1,T} &	=\frac{1}{T}\sum_{k=1}^{\lfloor N/2 \rfloor} \sum_{j=1}^M \Tr\Bigg(S_{1234 }\Big( \cm \big((\fdft{1}{1}, \fdftc{1}{1},\fdftcl{1}{1},\fdftl{1}{1}) \Big)\Bigg)
\\&+\frac{1}{T}\sum_{k=1}^{\lfloor N/2 \rfloor} \sum_{j=1}^M \Tr \Bigg(S_{1234 }\Big(\cm(\fdft{1}{1}, \fdftc{1}{1}) \otimes \cm(\fdftcl{1}{1},\fdftl{1}{1}) \Big)\Bigg)
\\& +\frac{1}{T}\sum_{k=1}^{\lfloor N/2 \rfloor} \sum_{j=1}^M \Tr\Bigg( S_{1324 }\Big(\cm(\fdft{1}{1}, \fdftcl{1}{1}) \otimes \cm(\fdftc{1}{1},\fdftl{1}{1}) \Big)\Bigg)
\\&+\frac{1}{T}\sum_{k=1}^{\lfloor N/2 \rfloor} \sum_{j=1}^M \Tr \Bigg(S_{1423 }\Big(\cm(\fdft{1}{1}, \fdftl{1}{1}) \otimes \cm(\fdftc{1}{1},\fdftcl{1}{1}) \Big)\Bigg)
\end{align*}
where $S_{ijkl}$ denotes the permutation operator on $\otimes_{i=1}^{4}L^2([0,1],\mathbb{C})$ that maps the components of the tensor according to the permutation $(1,2,3,4) \mapsto (i,j,k,l)$.
By \autoref{lem:cumfixu} and \autoref{cor:cumbound}, we thus find
\[
\E \hat{F}_{1,T} 
=\frac{1}{T}\sum_{k=1}^{\lfloor N/2 \rfloor} \sum_{j=1}^M \innprod{\mathcal{F}_{u_j,\omega_{k}}}{\mathcal{F}_{u_j,\omega_{k-1}}}_{HS} +  O(M^{-2})+O(N^{-1})
\]
Similarly as for $\hat F_{1,T}$ we obtain 
\begin{align*}
 \E \hat{F}_{2,T} = &\frac{1}{NM^2}\sum_{k=1}^{\lfloor N/2 \rfloor} \sum_{j_1,j_2=1}^M\Tr \Big( \cm \big((\fdft{1}{}, \fdftc{1}{},\fdftc{2}{},\fdft{2}{}) \big)
\\&+ \frac{1}{NM^2}\sum_{k=1}^{\lfloor N/2 \rfloor} \sum_{j_1,j_2=1}^M\Tr \Bigg(S_{1234 }\Big(\cm(\fdft{1}{}, \fdftc{1}{}) \otimes \cm(\fdftc{2}{},\fdft{2}{}) \Big)\Bigg)
\\& + \frac{1}{NM^2}\sum_{k=1}^{\lfloor N/2 \rfloor} \sum_{j_1,j_2=1}^M \Tr\Bigg( S_{1324 }\Big(\cm(\fdft{1}{}, \fdftc{2}{}) \otimes \cm(\fdftc{1}{},\fdft{2}{}) \Big)\Bigg)
\\&+ \frac{1}{NM^2}\sum_{k=1}^{\lfloor N/2 \rfloor} \sum_{j_1,j_2=1}^M\Tr \Bigg(S_{1423 }\Big(\cm(\fdft{1}{}, \fdft{2}{}) \otimes \cm(\fdftc{1}{},\fdftc{2}{}) \Big)\Bigg).
\end{align*}
\autoref{cor:cumbound} and \autoref{lem:cumdifu} then yield
\begin{align*}
\frac{1}{NM^2}\sum_{k=1}^{\lfloor N/2 \rfloor} \sum_{j_1,j_2=1}^M \innprod{\F_{u_{j_1},\omega_k}}{\F_{u_{j_2}\omega_k}}_{HS}+  \frac{1}{NM^2}\sum_{k=1}^{\lfloor N/2 \rfloor} \sum_{j_1=1}^M \Tr\Big( S_{1324 }\big(\mathcal{F}_{u_j,\omega_k} \bigotimes \mathcal{F}_{u_j,\omega_k} \big)\Big)+  O(\frac{1}{T})+O(\frac{1}{M^2}). 
\end{align*}
Note that the permutation operator implies $\Tr\Big( S_{1324 }\big(\mathcal{F}_{u_j,\omega_k} \bigotimes \mathcal{F}_{u_j,\omega_k} \big)\Big)= \Tr \big(\mathcal{F}_{u_j,\omega_k} \widetilde{\bigotimes} \mathcal{F}_{u_j,\omega_k} \Big)$. Therefore, with slight abuse of notation
\begin{align*}
\lim_{N,M\to\infty}\operatorname \E\hat F_{2,T}
	=&\frac1{4\pi}\int_{-\pi}^\pi \innprod{\int_0^1{\mathcal{F}}_{u,\omega}du}{\int_0^1 {\mathcal{F}}_{u,\omega}du}_{HS} \mathrm du\mathrm d\omega
	+ \frac{N}{T}  {B_{N,T}}. 
\end{align*}
where the term $B_{N,T}$ is defined in \eqref{eq:bias}.

%As an unbiased estimator of $B_{N,T}$, we consider, similarly to $\hat{F}_{1,T}$, periodogram operators at a Fourier lag apart, i.e.,
%\begin{align} \label{biasest}
%\hat{B}_{N,T}&  =\frac{1}{NM}\sum_{k=1}^{\lfloor N/2 \rfloor} \sum_{j=1}^M \Tr\Big({I}^{u_j,\omega_k} \widetilde{\otimes} {I}^{u_j,\omega_{k-1}} \Big)
%=\frac{1}{NM}\sum_{k=1}^{\lfloor N/2 \rfloor} \sum_{j=1}^M \|{\fdft{j}{k}}\|^2_2 \|\fdft{j}{k-1}\|^2_2  
%\end{align}
In complete analogy with the derivation of $\E \hat{F}_{1,T}$, we find that the estimator $\hat B_{N,T}$ defined in \eqref{biasest}
is asymptotically unbiased, i.e., %$\frac{N}{\sqrt{T}}$ consistency of the estimator $ \hat{B}_{N,T}$, i.e., 
\[
\lim_{T \to \infty}\E\hat{B}_{N,T} = B_{N,T} \tageq \label{eq:unbBNT}
\]
Summarizing, we obtain
\[
\sqrt T\Biggl[4\pi  \E(\hat F_{1,T}-\hat{F}_{2,T}+\frac{N}{T} \hat{B}_{N,T})-\Big (\int_{-\pi}^\pi\int_0^1\snorm{\F_{u,\omega}}_2^2dud\omega-\int_{-\pi}^\pi\snorm{\widetilde{\F_\omega}}_2^2d\omega \Big )\Biggr]\to0,
\]
as $T\to\infty$, provided \autoref{ratesNM} is satisfied.
\\

{
In order to establish \eqref{eq:convofe2} and \eqref{eq:cum>2}, it is of importance to be able to determine which indecomposable partitions of the array \eqref{tab:arraythm41} are vanishing in a more structured fashion. The following two results allow us to exploit the structure of the array. The next
Lemma provides a global bound on the the cumulants that is implied by the behavior of the joint cumulants of the local fDFT's for different midpoints (\autoref{lem:cumdifu}).  For a fixed partition $P=\{P_1,\ldots,P_G\}$ of the array denote the size of the partition by $G$. 
\begin{lemma} \label{lem:highcumF1}
If \autoref{cumglsp} is satisfied then for finite $n$, 
\begin{align*}
& T^{n/2}\operatorname{cum}_{n-x,x}(\hat F_{1,T},\hat F_{2,T})
\\& \frac{1}{T^{n/2}M^{x}}\sum_{k_1,\ldots,k_n=1}^{\lfloor N/2 \rfloor} \sum_{\substack{j_1,\ldots,j_{n}\\ j_{n+1},\ldots, j_{n+x}=1}}^M\Tr\Big(\sum_{\boldsymbol{P} = P_1 \cup \ldots \cup P_G}S_{\boldsymbol{P}}\Big(\otimes_{g=1}^{G}
\operatorname{cum}\big(\fdft{p}{p}|p\in \nu_g\Big)\Big) = O(T^{1-n/2} N^{G-n-1})\end{align*}
uniformly in  $0 \le x \le n$.
\end{lemma}
\begin{proof}[Proof of Lemma \ref{lem:highcumF1}]
For a fixed partition $P=\{P_1,\ldots,P_G\}$,  let the cardinality of set $P_g$ be denoted by $|P_{g}|=\mathscr{C}_g$. 
By \eqref{eq:boundonmani} of Corollary \ref{cor:cumbound} and Lemma \ref{lem:cumdifu} an upperbound of \eqref{tab:arraythm41} is given by
\begin{align*}
O \Big ( T^{-n/2} M^{-x}\sum_{k_1,\ldots,k_n=1}^{\lfloor N/2 \rfloor}\sum_{\substack{j_1,\ldots,j_{n}\\ j_{n+1},\ldots, j_{n+x}=1}}^M \prod_{g=1}^{G} \frac{1}{N^{\mathscr{C}_g/2-1}} M^{-\delta_{\{\exists p_1, p_2 \in P_g: |j_{p_1} - j_{p_2}|>1\}}}  \Big )
\tageq \label{cumhighF1UB}
\end{align*}
Note that $|\mathscr{C}_g| \ge 2$ and that the partition must be indecomposable. We can therefore assume, without loss of generality, that row $l$ hooks with row $l+1$ for $l=1,\ldots,n-1$, i.e., within each partition there must be at least one set $P_g$ that contains an element from both rows. For fixed $j_{l}$, there are only finitely many possibilities, say $E$, for $j_{l+1}$ (Lemma \ref{lem:cumdifu}). If the set does not cover another row, then the fact that $j_{l}$ is fixed and $j_{l+1}$ are fixed, another set must contain at least an element from row $l$ or $l+1$. But since the sets must communicate there are only finitely many options for $j_{l+2}$. If, on other hand, the same set covers elements from yet another row then given a fixed $j_l$, there are again finitely many options for $j_{l+1}$ and for $j_{l+2}$. This argument can be continued inductively to find \eqref{cumhighF1UB} is of order 
\[
O(N^{n/2} M^{-n/2-x} E^n M^{1+x} N^{-2n+G})   =O(T^{1-n/2} N^{G-n-1}).\qedhere
\]
\end{proof}
Lemma \ref{lem:highcumF1} implies that for $n=2$ partitions with $G \le 2$ vanish, while for $n > 2$ all partitions of size $G \le n+1$ will vanish asymptotically. Moreover, indecomposability of the array requires to stay on the frequency manifold  (\eqref{eq:boundoffmani} of Corollary \ref{cor:cumbound}) and therefore imposes additional restrictions in frequency direction. in case $n=2$,  only those partitions of size $G \ge 2$ for which all sets are such that $\sum_{k \in P_g} \omega_{k} \equiv 0 \mod 2\pi$ will not vanish. For $n >2$, indecomposability of the partition and \autoref{cor:cumbound} also result in restrictions over frequencies $k_1, \ldots,k_n$. These restrictions are formalized in the following proposition.
\begin{prop}\label{prop:cumhighF1frequb}
For a partition of size $G=n+r_1+1$ with $r_1 \ge 1$ of the array \eqref{tab:arraythm41} with $n > 2$, only partitions with at least $r_1$ restrictions in frequency direction are indecomposable. For $n=2$, $G=n+r_1+1$ with $r_1 \ge 1$ will have at least 1 restriction in frequency direction.
\end{prop}
\begin{proof}[Proof of Proposition \autoref{prop:cumhighF1frequb}]
First consider the case $n >2$. We note that a minimal amount of restrictions will be given by those partitions in which frequencies and their conjugates are always part of the same set, i.e., in which for fixed row $l$, the first two columns are in the same set and the last two columns are in the same set. Given we need that $G \ge n+2$ and $\mathscr{C}_g \ge 2$, indecomposability of the array means that the smallest number of restrictions is given by partitions that have one large set that covers the first two or last two columns and $n-r_1$ rows and for the rest has $\frac{4n-2(n-r_1)}{2} = n+r_1$ sets with $\mathscr{C}_g =2$. This means there are no constraints in frequency in $n-r_1-1$ rows but for the array to hook there must be $r_1$ constraints in terms of frequencies in rows $n-r_1-1$ to row $n$. \\
Consider then the special case of $n=2$, for which the above argument implies a partition of size $G=3$ and thus $r_1=0$. For partitions are of size $G \ge 4$, indecomposability then requires the first row to hook with the second, which imposes at least one restriction in frequency direction since only those partitions for which $\sum_{k \in P_g} \omega_{k} \equiv 0 \mod 2\pi$ will not vanish.
\end{proof}
Together \autoref{lem:highcumF1} and Proposition \autoref{prop:cumhighF1frequb} allow to show that  all higher order cumulants vanish asymptotically and therefore establishes asymptotic normality.
\begin{thm}\label{thm:highcumF1F2} Under  \autoref{cumglsp}, we have for all $x=0,\ldots,n$ and $n>2$,
\[T^{n/2}\operatorname{cum}_{n-x,x}(\hat F_{1,T},\hat F_{2,T})\to 0 \quad \text{as } T \to \infty.\]
\end{thm}
\begin{proof}
By  \autoref{lem:highcumF1}, it is direct that all partitions of size $G \le n+1$ vanish. We therefore only have to consider the case where $G=n+r_1+1$ with $r_1 \ge 1$. In this case, Proposition \autoref{prop:cumhighF1frequb}, yields an upperbound of the joint cumulant that is of order $O(T^{1-n/2} N^{n+r_1+1-n-1}N^{-{r_1}})=O(T^{1-n/2})$. This establishes \eqref{eq:cum>2}. 
\end{proof}
Combining \autoref{thm:highcumF1F2} with \eqref{eq:unbBNT}, we immediately obtain the following result for the bias correction
\begin{Corollary} \label{cor:consBNT}
Under the conditions of \autoref{thm:dist} 
\[
\sqrt{T} \big(\hat{B}_{N,T} -  B_{N,T}\big) \overset{p}{\to} 0.
\]
\end{Corollary}
}
\medskip
Finally, for the covariance structure of $\sqrt{T}\hat F_{1,T}$ and $\sqrt{T}\hat F_{2,T}$, we find in Appendix \ref{sec:cov}
\begin{align*}
\lim_{T \to \infty}T\cv(\hat{F}_{1,T}, \hat{F}_{1,T})
    &=  \frac{2}{8\pi} \int_{-\pi}^{\pi} \int_{-\pi}^{\pi}\int_0^{1} 
\innprod{\F_{u,\omega_{1},-\omega_{1},-\omega_{2}}}{\F_{u,\omega_{1}} \bigotimes \F_{u,\omega_{2}} }_{HS}du d\omega_1 d\omega_2 \\ 
&+   \frac{2}{8\pi} \int_{-\pi}^{\pi} \int_{-\pi}^{\pi}\int_0^{1} \innprod{ \F_{u,\omega_{1},-\omega_{1},\omega_{2}}}{ \F_{u,\omega_{1}}\bigotimes \F_{u,-\omega_{2}}}_{HS}  du d\omega_1 d\omega_2 \\ 
 &+  \frac{2}{4\pi} \int_{-\pi}^{\pi} \int_0^{1} \snorm{\F^2_{u,\omega}}^2_2 du d\omega \\ 
  &+  \frac{1}{4\pi} \int_{-\pi}^{\pi} \int_0^{1}\snorm{\F_{u,\omega}}^4_2 du d\omega \\ 
 &+  \frac{1}{4\pi} \int_{-\pi}^{\pi} \int_0^{1}   \innprod{\F_{u,\omega} \widetilde{\bigotimes} \F_{u,\omega}}{\F_{u,\omega} \bigotimes \F_{u,\omega}}_{HS}du d\omega \\ 
 &+  \frac{1}{4\pi} \int_{-\pi}^{\pi} \int_0^{1} \innprod{\F_{u,\omega} \widetilde{\bigotimes}_{\top} \F_{u,-\omega}}{ \F_{u,\omega}\bigotimes\F_{u,-\omega}}_{HS} du d\omega 
\\	\lim_{T \to \infty}T \cv(\hat{F}_{2,T}, \hat{F}_{2,T})
	&= \frac{2}{8\pi} \int_{-\pi}^{\pi} \int_{-\pi}^{\pi}\innprod{\widetilde{\F}_{\omega_1,-\omega_{1},-\omega_{2}}}{{\widetilde{\F}_{\omega_{1}}} \bigotimes \widetilde{\F}_{\omega_{2}}}_{HS}  d\omega_1 d\omega_2 \\ 
&+  \frac{2}{8\pi} \int_{-\pi}^{\pi} \int_{-\pi}^{\pi}\innprod{\widetilde{\F}_{\omega_1,-\omega_{1},\omega_{2}}}{{\widetilde{\F}_{\omega_{1}}} \bigotimes \widetilde{\F}_{-\omega_{2}}}_{HS}  d\omega_1 d\omega_2 \\  
 &+  \frac{1}{4\pi} \int_{-\pi}^{\pi} \int_0^{1}  \innprod{\F^2_{u,\omega }}{\F_{u,\omega} \widetilde{\F}_{\omega}}_{HS}du d\omega \\ 
  &+  \frac{1}{4\pi} \int_{-\pi}^{\pi} \int_0^{1} \innprod{{\widetilde{\F}_{\omega} }\F_{u,\omega}}{\F_{u,\omega} {\widetilde{\F}_{u,\omega}}}_{HS}du d\omega \\ 
 &+  \frac{1}{4\pi} \int_{-\pi}^{\pi} \int_0^{1}\innprod{\F_{u,\omega} \widetilde{\bigotimes} \F_{u,\omega}}{ \widetilde{\F}_{{\omega}} \bigotimes \widetilde{\F}_{\omega}}_{HS}  du d\omega \\ 
 &+  \frac{1}{4\pi} \int_{-\pi}^{\pi} \int_0^{1}\innprod{\F_{u,\omega} \widetilde{\bigotimes}_{\top} \F_{u,-\omega}}{ \widetilde{\F}_{{\omega}} \bigotimes \widetilde{\F}_{-\omega}}_{HS}  du d\omega \\ 
\lim_{T\to\infty} T\cv(\hat{F}_{1,T},\hat{F}_{2,T})
	&= \frac{2}{8\pi} \int_{-\pi}^{\pi} \int_{-\pi}^{\pi}\int_0^{1} \innprod{\F_{u,\omega_{1},-\omega_{1},-\omega_{2}}}{\F_{u,\omega_{1}} \bigotimes \widetilde{\F}_{\omega_{2}}}_{HS} du d\omega_1 d\omega_2 \\ 
&+ \frac{2}{8\pi} \int_{-\pi}^{\pi} \int_{-\pi}^{\pi}\int_0^{1} \innprod{\F_{u,\omega_{1},-\omega_{1},\omega_{2}}}{\F_{u,\omega_{1}} \bigotimes \widetilde{\F}_{-\omega_{2}}}_{HS} du d\omega_1 d\omega_2 \\ 
 &+  \frac{2}{4\pi} \int_{-\pi}^{\pi} \int_0^{1} \innprod{\F_{u,\omega} \F_{u,\omega}}{ \F_{u,\omega} \widetilde{\F}_{\omega} }_{HS} du d\omega \\ 
 &+  \frac{1}{4\pi} \int_{-\pi}^{\pi} \int_0^{1} \innprod{\F_{u,\omega} \widetilde{\bigotimes} \F_{u,\omega}}{\F_{u,\omega} \bigotimes \widetilde{\F}_{\omega}}_{HS}du d\omega \\ 
  &+  \frac{1}{4\pi} \int_{-\pi}^{\pi} \int_0^{1} \innprod{\F_{u,\omega} \widetilde{\bigotimes}_{\top} \F_{u,-\omega}}{\F_{u,\omega} \bigotimes \widetilde{\F}_{-\omega}}_{HS}du d\omega
\end{align*}
A straightforward calculation then yields the asymptotic variance $v^2$ is simply given by\[\nu^2=\lim_{T \to \infty}\Big(16 \pi^2 \V(\hat{F}_{1,T})+16 \pi^2 \V(\hat{F}_{2,T})-32 \pi^2 \cv(\hat{F}_{1,T},\hat{F}_{2,T})\Big).\]
We therefore obtain the following expression for the asymptotic variance 
\begin{align*}
&\nu^2= 4\pi \int_{-\pi}^{\pi} \int_{-\pi}^{\pi}\int_0^{1} 
\innprod{\F_{u,\omega_{1},-\omega_{1},-\omega_{2}}}{\F_{u,\omega_{1}} \bigotimes \F_{u,\omega_{2}} }_{HS}du d\omega_1 d\omega_2\\& + 4\pi \int_{-\pi}^{\pi} \int_{-\pi}^{\pi}\innprod{\widetilde{\F}_{\omega_1,-\omega_{1},-\omega_{2}}}{{\widetilde{\F}_{\omega_{1}}} \bigotimes \widetilde{\F}_{\omega_{2}}}_{HS}  d\omega_1 d\omega_2- 8\pi \int_{-\pi}^{\pi} \int_{-\pi}^{\pi}\int_0^{1} \innprod{\F_{u,\omega_{1},-\omega_{1},-\omega_{2}}}{\F_{u,\omega_{1}} \bigotimes \widetilde{\F}_{\omega_{2}}}_{HS} du d\omega_1 d\omega_2\\&
+  4\pi \int_{-\pi}^{\pi} \int_{-\pi}^{\pi}\int_0^{1}\innprod{ \F_{u,\omega_{1},-\omega_{1},\omega_{2}}}{ \F_{u,\omega_{1}}\bigotimes \F_{u,-\omega_{2}}}_{HS}  du d\omega_1 d\omega_2  + 4\pi \int_{-\pi}^{\pi} \int_{-\pi}^{\pi}\innprod{\widetilde{\F}_{\omega_1,-\omega_{1},\omega_{2}}}{{\widetilde{\F}_{\omega_{1}}} \bigotimes \widetilde{\F}_{-\omega_{2}}}_{HS}  d\omega_1 d\omega_2 
\\& - 8\pi  \int_{-\pi}^{\pi} \int_{-\pi}^{\pi}\int_0^{1} \innprod{\F_{u,\omega_{1},-\omega_{1},\omega_{2}}}{\F_{u,\omega_{1}} \bigotimes \widetilde{\F}_{-\omega_{2}}}_{HS} du d\omega_1 d\omega_2\\
 &+  8\pi \int_{-\pi}^{\pi} \int_0^{1} \snorm{\F^2_{u,\omega}}^2_2 du d\omega +  4\pi\int_{-\pi}^{\pi} \int_0^{1}  \innprod{\F^2_{u,\omega }}{\F_{u,\omega} \widetilde{\F}_{\omega}}_{HS}du d\omega +  4\pi\int_{-\pi}^{\pi} \int_0^{1} \innprod{{\widetilde{\F}_{\omega} }\F_{u,\omega}}{\F_{u,\omega} {\widetilde{\F}_{u,\omega}}}_{HS}du d\omega \\ 
  &- 16\pi  \int_{-\pi}^{\pi} \int_0^{1} \innprod{\F_{u,\omega} \F_{u,\omega}}{ \F_{u,\omega} \widetilde{\F}_{\omega} }_{HS} du d\omega + 4\pi  \int_{-\pi}^{\pi} \int_0^{1}\snorm{\F_{u,\omega}}^4_2 du d\omega  \\ 
 &+ 4\pi  \int_{-\pi}^{\pi} \int_0^{1}   \innprod{\F_{u,\omega} \widetilde{\bigotimes} \F_{u,\omega}}{\F_{u,\omega} \bigotimes \F_{u,\omega}}_{HS}du d\omega +  4\pi  \int_{-\pi}^{\pi} \int_0^{1} \innprod{\F_{u,\omega} \widetilde{\bigotimes}_{\top} \F_{u,-\omega}}{ \F_{u,\omega}\bigotimes\F_{u,-\omega}}_{HS} du d\omega 
 \\
 &+  4\pi \int_{-\pi}^{\pi} \int_0^{1}\innprod{\F_{u,\omega} \widetilde{\bigotimes} \F_{u,\omega}}{ \widetilde{\F}_{{\omega}} \bigotimes \widetilde{\F}_{\omega}}_{HS}  du d\omega +  4\pi \int_{-\pi}^{\pi} \int_0^{1}\innprod{\F_{u,\omega} \widetilde{\bigotimes}_{\top} \F_{u,-\omega}}{ \widetilde{\F}_{{\omega}} \bigotimes \widetilde{\F}_{-\omega}}_{HS}  du d\omega
\\ 
 &-8\pi \int_{-\pi}^{\pi} \int_0^{1} \innprod{\F_{u,\omega} \widetilde{\bigotimes} \F_{u,\omega}}{\F_{u,\omega} \bigotimes \widetilde{\F}_{\omega}}_{HS}du d\omega -8\pi  \int_{-\pi}^{\pi} \int_0^{1} \innprod{\F_{u,\omega} \widetilde{\bigotimes}_{\top} \F_{u,-\omega}}{\F_{u,\omega} \bigotimes \widetilde{\F}_{-\omega}}_{HS}du d\omega \\ 
  \tageq\label{eq:asvar} 
\end{align*}

Under $H_0$ this reduces to 
\[\nu^2_{H_0}=4\pi \int_{-\pi}^{\pi} \snorm{\widetilde{\F}_\omega}_2^4d\omega.\]

\section{Finite sample properties}  \label{sec4}
\def\theequation{4.\arabic{equation}}
\setcounter{equation}{0}

 In this section, we investigate the finite sample properties of the methods proposed in this
 paper by means of a simulation study and illustrate potential applications analysing annual temperature curves.

\subsection{Simulation study} \label{sec4a}

%\subsubsection{Tests for the classical hypothesis $H_0: m^2=0$} \label{sec4a1}
For the investigation  of the finite sample performance of the test \eqref{testwn}  for the hypothesis 
$H_0: m^2=0$
 with simulated data  we consider  a  similar set-up   as  \cite{avd16_main}, who 
used a  Fourier basis representation on the interval $[0,1]$ to  generate functional data. To be precise, let $\{\psi_l\}_{l=1}^{\infty}$ be the Fourier basis functions. Consider the  $p$-th order time varying functional autoregressive process (tvFAR(p)), $(X_t, t \in \mathbb Z)$ defined as
\begin{equation}
\label{tvFAR}
X_{t}(\tau) = \sum_{t'=1}^p A_{t,t'}(X_{t-t'})(\tau) + \epsilon_t(\tau), \hspace{0.4 in} \tau \in [0,1],
\end{equation}
where $A_{t,1}, \dots, A_{t,p}$ are time-varying auto-covariance operators and $\{\epsilon_t(\tau)\}_{t\in \mathbb{Z}}$ is a sequence of mean zero innovations. We have
\begin{align}
\langle X_t, \psi_l \rangle = &\sum_{l'=1}^{\infty} \sum_{t'=1}^p\langle X_{t-t'},\psi_l\rangle \langle A_{t,t'}(\psi_l),\psi_{l'}\rangle + \langle \epsilon_t, \psi_l\rangle \nonumber\\
\approx & \sum_{l'=1}^{L_{max}} \sum_{t'=1}^p\langle X_{t-t'},\psi_l\rangle \langle A_{t,t'}(\psi_l),\psi_{l'}\rangle + \langle \epsilon_t, \psi_l\rangle.\label{sim_approx}
\end{align}
Therefore the first $L_{max}$ Fourier coefficients of the process $X_t$ are  generated using the $p$-th order vector autoregressive, VAR(p), process
$$\widetilde{X}_t = \sum_{t'=1}^p \widetilde{A}_{t,t'}\widetilde{X}_{t-t'} + \widetilde{\epsilon}_t,$$
where $\widetilde{X}_t := \left(\langle X_t, \psi_1 \rangle, \dots, \langle X_t, \psi_{L_{max}} \rangle \right)^T$ is the vector of Fourier coefficients, the $(l,l')$-th entry of $\widetilde{A}_{t,j}$ is given by $\langle A_{t,j}(\psi_l),\psi_{l'}\rangle$ and $\widetilde{\epsilon}_t := \left(\langle \epsilon_t, \psi_1 \rangle, \dots, \langle \epsilon_t, \psi_{L_{max}} \rangle \right)^T$.
The entries of the matrix $\widetilde{A}_{t,j}$ are generated as $N\big (0, \nu_{l,l'}^{(t,j)}\big )$ with $\nu_{l,l'}^{(t,j)}$ specified below. To ensure stationarity or existence of a causal solution the norms $\kappa_{t,j}$ of $A_{t,j}$ are required to satisfy certain conditions [see \cite{bosq2000} for stationary 
and \cite{vde16} for  local stationary time series, respectively]. 

If $A_{t,j} \equiv A_j$ for all $t$ in \eqref{tvFAR} and the error sequence $(\epsilon_t, t \in \mathbb Z)$ is an i.i.d. sequence, we obtain 
the stationary functional autoregressive (FAR) model of order $p$. In that case we generate the entries of the operator matrix from $N\big (0, \nu_{l,l'}^{(j)}\big )$ distributions. Functional white noise can be thought of as FAR model of order $0$.

Throughout this  section the number of Monte Carlo replications is always $1000$.
We use the  fda package from R to generate the functional data, where $L_{max}$ is taken to be $15$. The periodogram kernels are evaluated on a $100 \times 100$ grid on 
the square $[0,1]^2$ and their integrals are calculated by averaging the functional values at the grid points. The asymptotic variance under the  null hypothesis is estimated by   \eqref{varesth0}. 
In Table \ref{Sim_size}
we  report the simulated nominal levels of the test \eqref{testwn}  for the hypotheses in  \eqref{hequiv0} 
for the sample sizes $T=128$, $256$, $512$ and $1024$, where we consider the following
three (stationary) data generating processes:
\begin{itemize}
\item [(I)] The functional white noise variables $\epsilon_1, \dots, \epsilon_T$ i.i.d. with coefficient variances $\text{Var}(\langle \epsilon_t, \psi_l \rangle) = \exp((l-1)/10)$.
\item[(II)] The FAR(2) variables $X_1, \dots, X_T$ with operators specified by variances $\nu_{l,l'}^{(1)} = \exp(-l-l')$ and $\nu_{l,l'}^{(2)} = 1/(l+l'^{3/2})$ with norms $\kappa_1 = 0.75$ and $\kappa_2 = -0.4$ and the innovations  $\epsilon_1, \dots, \epsilon_T$ are as in (I).
\item[(III)] The FAR(2) variables $X_1, \dots, X_T$ as in (II) but with $\kappa_1=0.4$ and $\kappa_2 = 0.45$.
\end{itemize}
Recall that the test requires the choice of  the number $M$ of blocks, which determines the number $N$ of observations in each block by the equation $T=MN$. As mentioned before, the quantities
 $M$ and $N$ have to be reasonable large, because they correspond to the number of terms used in the Riemann sum approximating the integral with respect to $du$ and $d\omega$ in \eqref{eq:M^2}.  We investigate the effect of this choice in more detail  in the next section but for the moment consider only those combinations for which \autoref{ratesNM} is satisfied. 
Interestingly, the results reported in Table \ref{Sim_size} are rather robust with respect to this choice and we observe a reasonable approximation of the nominal level in nearly all cases under consideration, albeit the test being slightly undersized for the smalle samples sizes. 
  
\medskip
\begin{table*}[t]
\begin{center}
\caption{\it Empirical rejection probabilities (in percentage) of the test \eqref{testwn}  for the hypotheses in \eqref{hequiv0} 
 under the  null hypothesis}
\label{Sim_size}
\begin{tabular}{cccccccccccc}
\hline
& & & \multicolumn{3}{c} {I} & \multicolumn{3}{c} {II} & \multicolumn{3}{c} {III}\\
\cmidrule(lr){4-6} \cmidrule(lr){7-9} \cmidrule(lr){10-12}
 T & N & M & 10\% & 5\% & 1\%  & 10\% & 5\% & 1\%  & 10\% & 5\% & 1\% \\
\hline
128 & 32 & 4 & 6.0 & 2.8 & 0.7 & 7.4 & 3.7 & 0.7 & 6.4 & 2.5 & 0.3 \\ 
128 & 16 & 8 & 5.9 & 2.7 & 0.4 & 7.3 & 2.8 & 0.8 & 5.2 & 2.5 & 0.5 \\ 
\\
 256 & 32 & 8 & 7.0 & 3.2 & 0.5 & 7.1 & 4.1 & 0.7 & 6.8 & 3.5 & 0.7 \\ 
 256 & 16 & 16 & 7.5 & 2.9 & 0.5 & 7.4 & 3.6 & 0.7 & 7.0 & 3.0 & 0.5 \\ 
\\
512 & 64 & 8 & 7.5 & 3.1 & 0.5 & 8.6 & 4.2 & 0.3 & 7.9 & 3.5 & 0.6 \\ 
512 & 32 & 16 & 6.7 & 2.4 & 0.4 & 7.1 & 3.3 & 0.7 & 6.4 & 2.4 & 0.2 \\ 
\\
 1024 & 128 & 8 &  8.8 & 4.2 & 1.0 & 9.6 & 4.1 & 1.0 & 8.9 & 3.9 & 0.9 \\ 
1024 & 64 & 16 & 9.7 & 4.7 & 1.1 & 10.0 & 5.3 & 1.4 & 9.8 & 4.6 & 0.9 \\ 
1024 & 32 & 32 & 8.0 & 3.3 & 0.5 & 9.3 & 5.2 & 1.3 & 8.0 & 3.6 & 0.5 \\ 
\hline
\end{tabular}
\end{center}
\end{table*}

Next we investigate the performance  of the test \eqref{testwn}   under the  alternative, where we consider the (non-stationary) 
 data generating processes:
\begin{itemize}
\item [(IV)] The tvFAR(1) variables $X_1, \dots, X_T$ with operator specified by variances $\nu_{l,l'}^{(t,1)} = \nu_{l,l'}^{(1)} = \exp(-l-l')$ and norm $\kappa_1 = 0.8$, and innovations are as in (I) with a multiplicative time-varying variance
$$\sigma^2(t) = \cos\Big(\frac{1}{2} + \cos\big(\frac{2\pi t}{1024}\big) + 0.3 \sin\big(\frac{2\pi t}{1024}\big)\Big).$$
\item[(V)] The tvFAR(2) variables $X_1, \dots, X_T$ with operators as in (IV), but with time-varying norm
$$\kappa_{1,t}=1.8 \cos\left( 1.5 - \cos\left(\frac{4\pi t}{T}\right)\right)$$
and constant norm $\kappa_2 = -0.81$ and innovations are as in (I).
\item[(VI)] The  structural break FAR(2) variables $X_1, \dots, X_T$ generated as follows
\begin{itemize}
\item for $t \leq 3T/8$, the operators are as in (II) with norms $\kappa_1 = 0.7$ and $\kappa_2 = 0.2$, with innovations as in (I).
\item for $t > 3T/8$, the operators are as in (II) with norms $\kappa_1 = 0$ and $\kappa_2 = -0.2$, with innovations as in (I) but with  coefficient variances $\text{Var}(\langle \epsilon_t, \psi_l \rangle) = 2\exp((l-1)/10)$.
\end{itemize}
\end{itemize}
The results of the test \eqref{testwn} under the alternative are displayed in Table \ref{Sim_power}. We observe that the test has very good power for model IV and VI, even for small sample sizes. For model V the power is observably lower than for the other two models but is still very good and not completely unintuitive as it can be explained by its data generating mechanism; depending on the draw of the operators, the resulting process in finite samples can be highly dependent as well as show barely any dependence at all. These results therefore coincide with the findings of \cite{avd16_main}.

\begin{table*}[h]
\begin{center}
\caption{\it Empirical rejection probabilities (in percentage) of the test \eqref{testwn}  for the hypotheses in \eqref{hequiv0}   under the  alternative hypothesis.}
\label{Sim_power}
\begin{tabular}{cccccccccccc}
\hline
& & & \multicolumn{3}{c} {IV} & \multicolumn{3}{c} {V} & \multicolumn{3}{c} {VI}\\
\cmidrule(lr){4-6} \cmidrule(lr){7-9} \cmidrule(lr){10-12}
 T & N & M & 10\% & 5\% & 1\%  & 10\% & 5\% & 1\%  & 10\% & 5\% & 1\% \\
\hline
128 & 32 & 4 &  65.8 & 55.8 & 31.8 & 55.2 & 43.1 & 19.1 & 72.8 & 59.7 & 34.2 \\ 
128 & 16 & 8 &  66.7 & 57.1 & 36.9 & 46.4 & 37.9 & 24.1 & 41.6 & 30.1 & 12.6 \\ 
\\
 256 & 32 & 8 &   99.9 & 99.8 & 99.7 & 73.1 & 65.2 & 46.3 & 65.1 & 53.6 & 30.2 \\ 
 256 & 16 & 16 &  99.5 & 99.4 & 99.2 & 54.2 & 48.8 & 37.6 & 70.8 & 59.0 & 34.0 \\ 
\\
 512 & 64 & 8 & 99.9 & 99.9 & 99.9 & 89.3 & 85.1 & 71.6 & 90.6 & 82.5 & 62.2 \\ 
512 & 32 & 16 &  100.0 & 100.0 & 100.0 & 80.2 & 75.3 & 66.6 & 92.2 & 87.8 & 70.1 \\ 
\\
 1024 & 128 & 8 &  100.0 & 100.0 & 100.0 & 92.2 & 90.1 & 83.9 & 99.6 & 98.4 & 92.9 \\ 
 1024 & 64 & 16 & 100.0 & 100.0 & 99.9 & 90.2 & 88.2 & 83.5 & 99.7 & 99.1 & 96.5 \\ 
1024 & 32 & 32 &  99.9 & 99.9 & 99.9 & 81.4 & 79.8 & 74.6 & 99.3 & 98.5 & 95.9 \\ 
\hline
\end{tabular}
\end{center}
\end{table*}

\subsection{Choice of $M$ and $N$}\label{sec:choiceofMandN}

To further investigate how the choice of $M$ and $N$ affects the test's performance, we considered a simulation study with sample size equal to $T=4096$  as this allows us to vary $M$ from $M=4,8,\ldots,1024$. We note that we thus also include choices of $M$ for which assumption \eqref{ratesNM} does not hold. The study was again performed over 1000 replications of each of the above models.  
\autoref{fig:dens}{(a)-(c)} provides the estimated densities for each $M$ for model I, II and III respectively. The estimated densities of model I appear well-aligned with a standard normal for all values of $M$. The fit appears however best for $16 \le M \le 128$.  For model II and III, we clearly observe that for $M>N$ the distribution becomes skewed and flatter. This is intuitive since the assumptions underlying \autoref{thm:dist} do not hold. The difference with the standard normal curve seems to become more pronounced the stronger the dependence. From these three models the dependence is strongest for  model B. In order to quantify our observations, we computed the mean absolute error to measure the difference between the estimated density of the test statistic and the standard normal density (see \autoref{fig:dens}(d)). The results indicate that a relatively small value of $M$ compared to $N$ leads to the best approximation. $M$ should however not be too small. More specifically, a minimal error is attained with $M=32$ for model~I and model~III and with $M=16$ for model~II. 

\autoref{fig:Rejpbalt} shows the rejection probabilities for $\alpha=0.1,\, 0.05,\, 0.01$ under the three alternatives. For model~IV and model~VI we find perfect power for all choices of $M$ and all critical values. For model~V, there is some sensitivity and power seems best for $8 \le M \le 32$. As previously remarked, the sensitivity for model~V is due to its data generating mechanism. To summarize, it appears that our  test is very robust for different choices of $M$ for which Assumption \ref{ratesNM} is satisfied. The empirical study indicates in particular good performance for the range $16 <M < 64$ for $T=4096$ which corroborates with our findings in \ref{sec3}. 
% {We leave a more rigorous investigation of the optimal values of $M$ and $N$ for future work.}

\begin{figure}
\centering
\begin{subfigure}{0.55\textwidth}
  \centering
  \includegraphics[scale = 0.4]{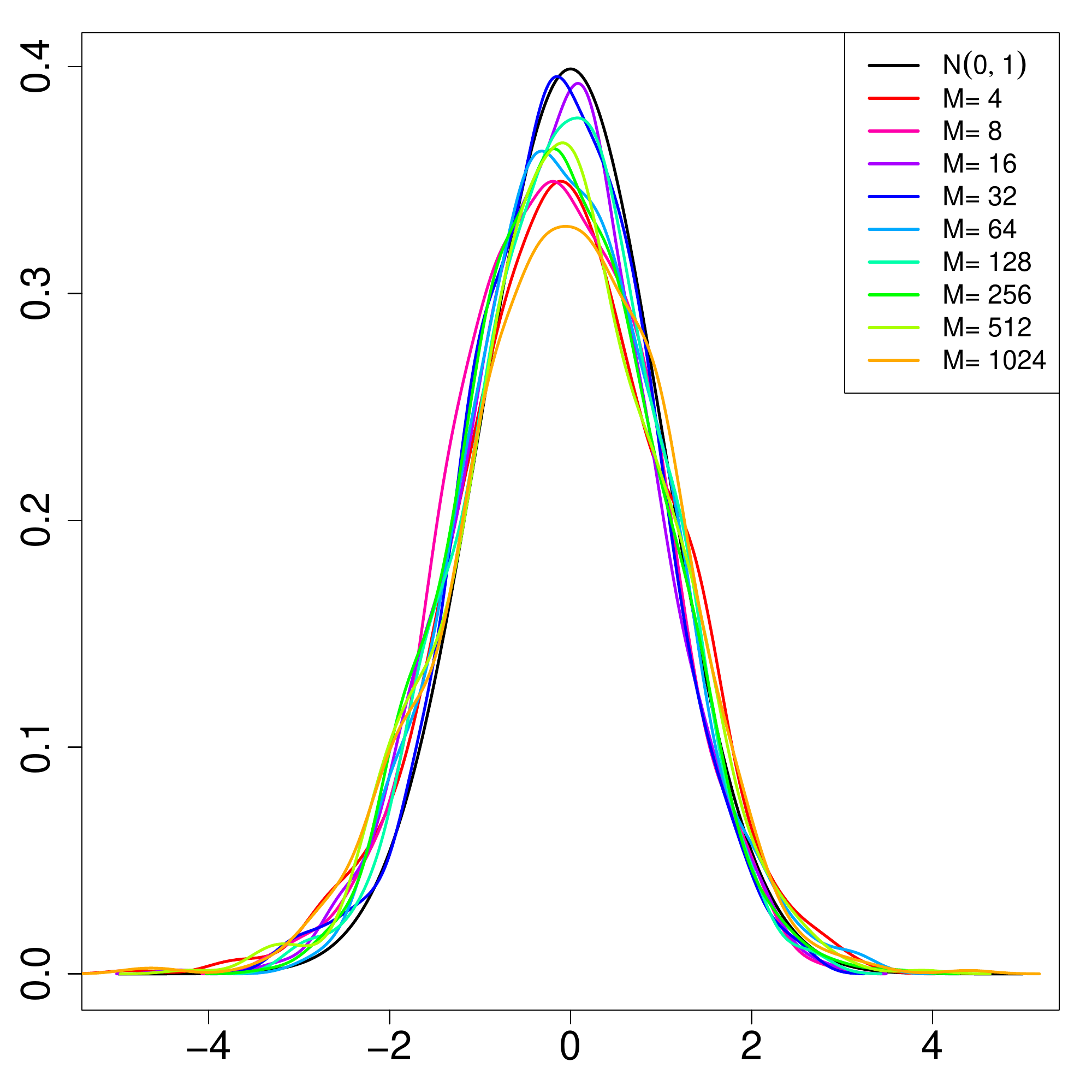}
   \caption{Model I}
\end{subfigure}%
\begin{subfigure}{.55\textwidth}
  \centering
  \includegraphics[scale = 0.4]{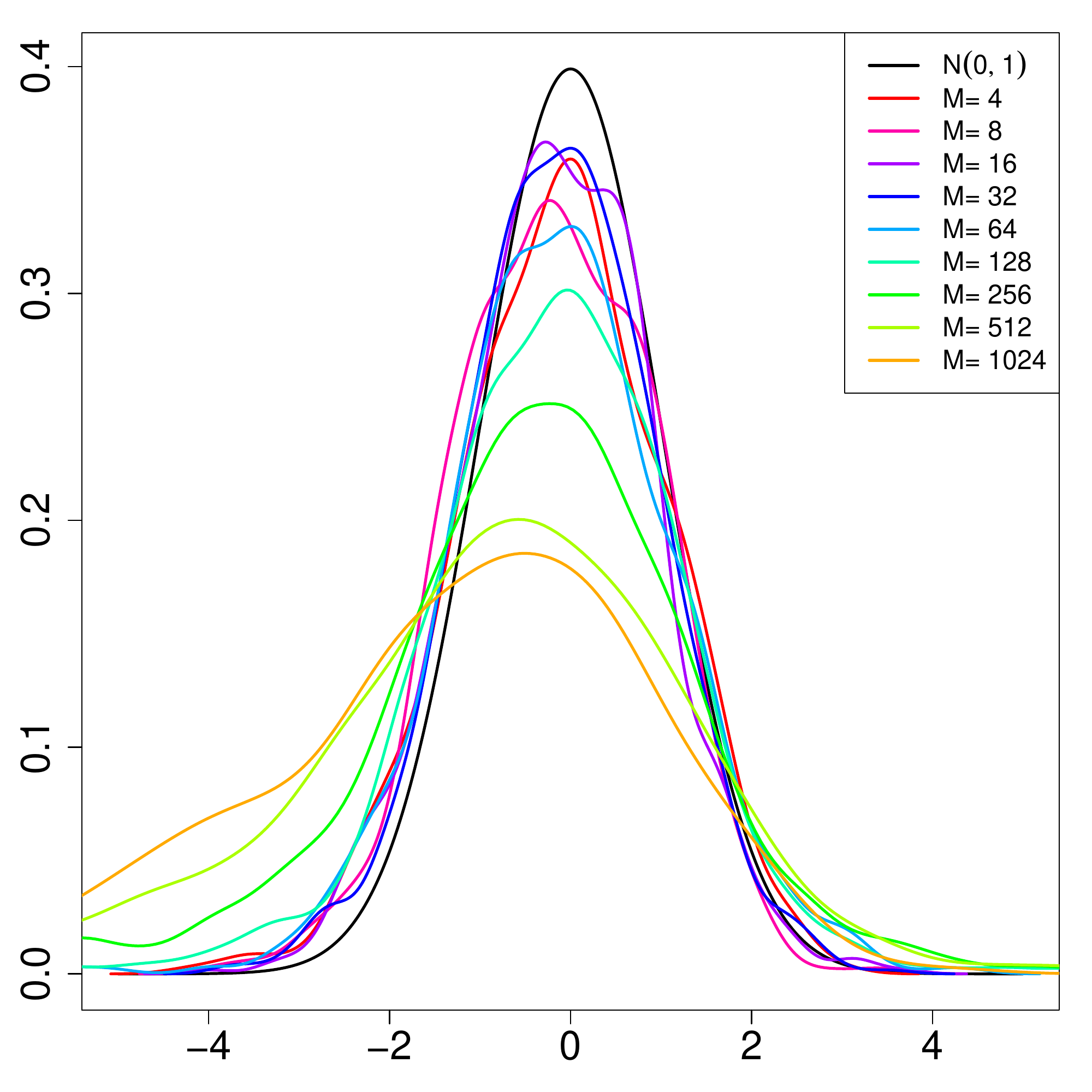}
  \caption{Model II}
\end{subfigure}
\centering
\begin{subfigure}{0.55\textwidth}
  \centering
  \includegraphics[scale = 0.4]{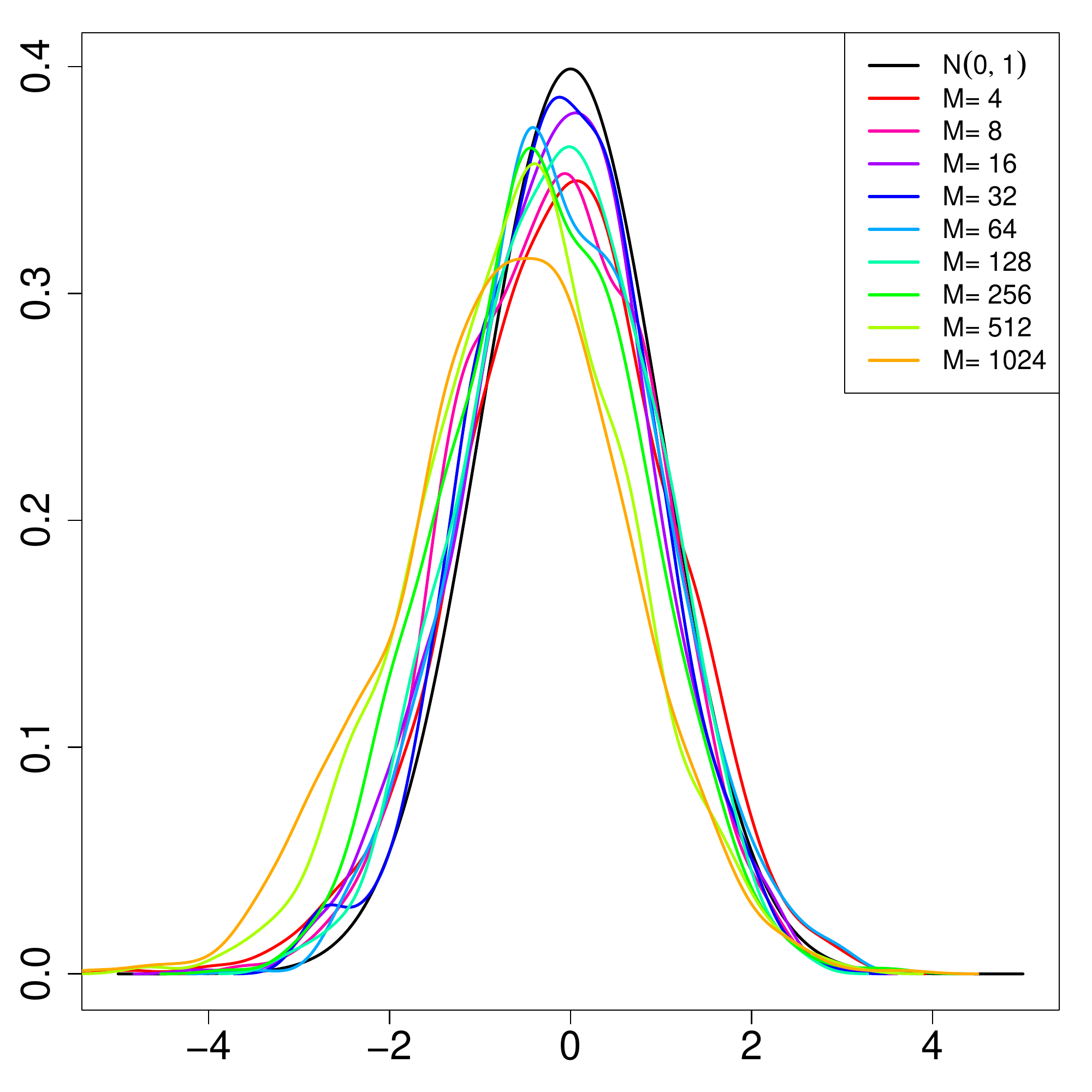}
   \caption{Model III}
\end{subfigure}%
\begin{subfigure}{.55\textwidth}
  \centering
  \includegraphics[scale = 0.4]{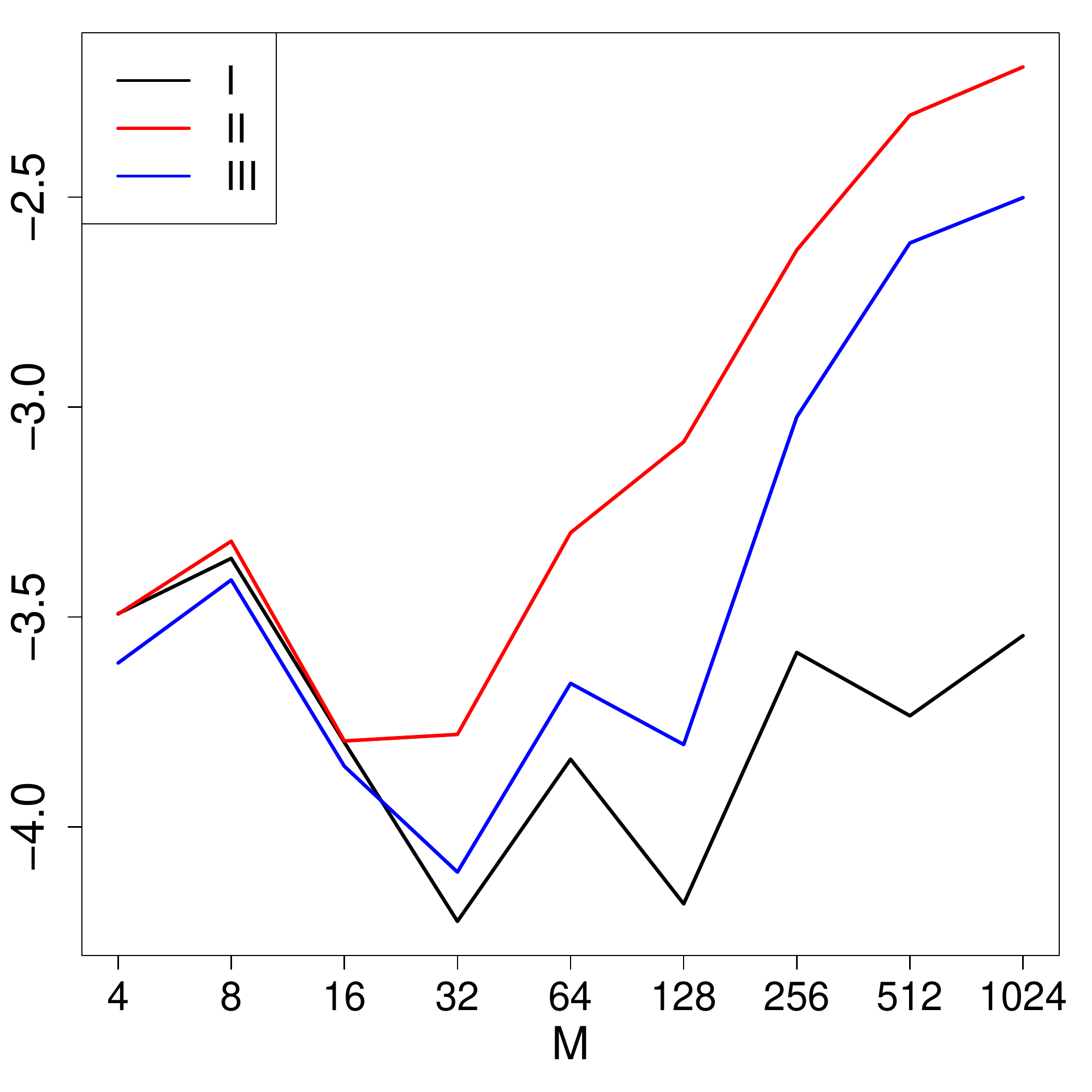}
  \caption{log MAE}
\end{subfigure}
\caption{\it (a)-(c) Estimated densities for different choices of $M$ with $T=4096$ compared to a standard normal distribution (black); (d) Natural logartihm of the mean absolute error compared to a standard normal distribution.}
\label{fig:dens}
\end{figure}

\begin{figure}
\centering
  \includegraphics[scale = 0.4]{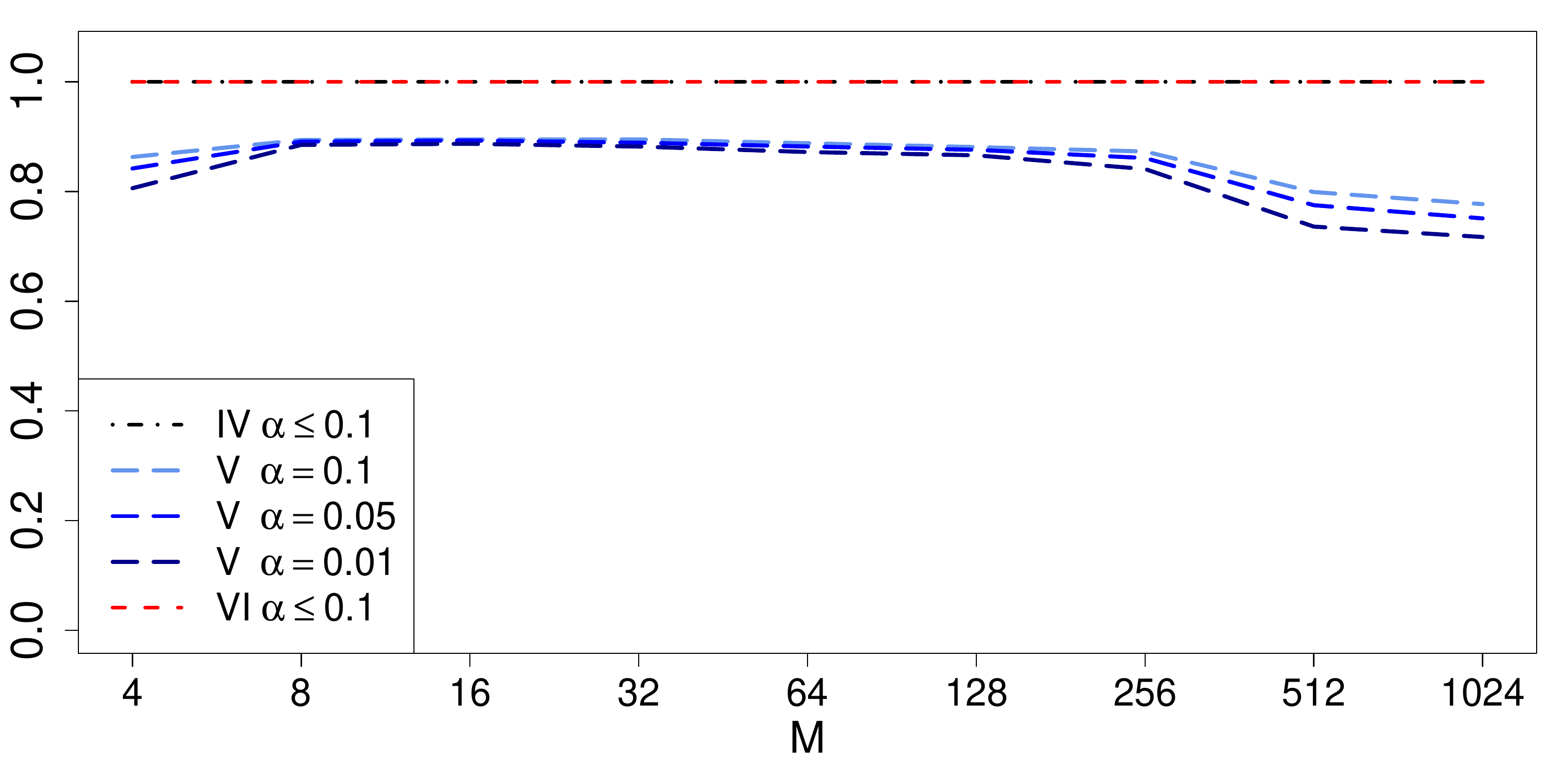}
   \caption{Rejection probabilities for the three alternative models for $T=4096$.}
  \label{fig:Rejpbalt}
\end{figure}

\subsection{Data example} \label{sec4b}
We illustrate the new methodology proposed in this  paper analyzing annual temperature curve data, recorded at several measuring stations across Australia. The recorded daily minimum temperatures for every year are treated as functional data. The locations of the measuring stations and lengths of the time series are reported in Table \ref{pval_data}. The temperature curves of Sydney and Boulia airport are exemplarily  presented in Figure \ref{fig:ts}.

\begin{figure}
\centering
\begin{subfigure}{.5\textwidth}
  \centering
  \includegraphics[scale = 0.4]{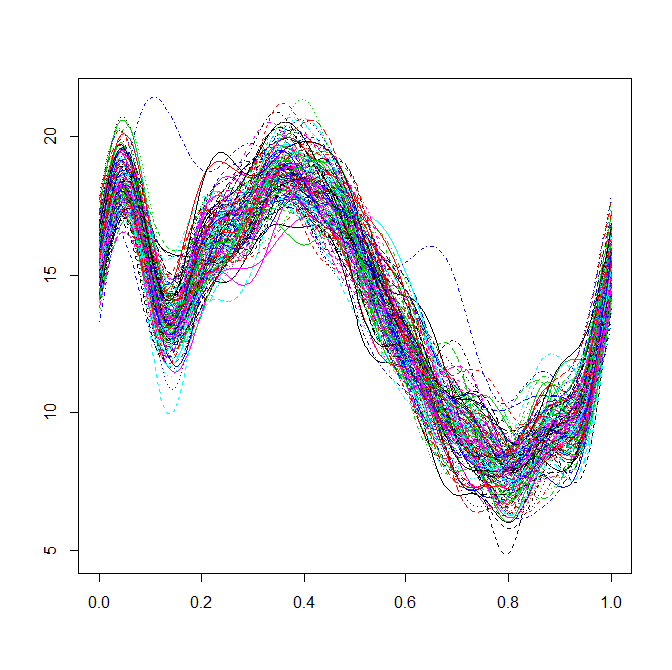}
   \caption{Sydney}
\end{subfigure}%
\begin{subfigure}{.5\textwidth}
  \centering
  \includegraphics[scale = 0.4]{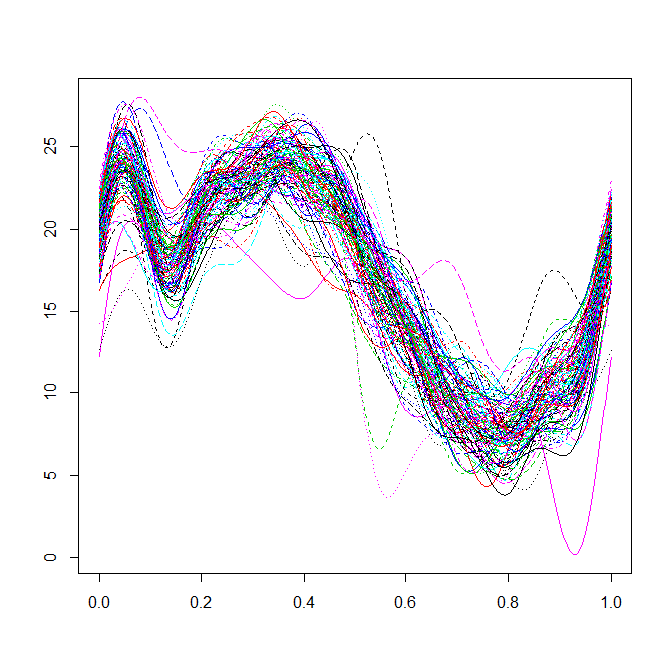}
  \caption{Boulia Airport}
\end{subfigure}
\caption{\it Time series of minimum temperature curves}
\label{fig:ts}
\end{figure}

We shall use the proposed test  in \eqref{testwn} to investigate  whether these temperature curves come from a stationary process or not. For the choice of the number of blocks, we use the above findings i.e., $M= \lceil T^{1/3}\rceil$. However, given the number of curves for each station, this value can be rather small and therefore we also consider $\lfloor T^{1/2}\rfloor$. As a fixed comparison for all curves, we take $M=8$ since the sample length is closest to $T=128$ (and often slightly larger).  The corresponding values of the test statistic  \eqref{testwn} for the hypothesis of stationarity are reported in \autoref{pval_data}. It is clear that we reject the null of stationarity in all cases at a 1\% significance level. The test therefore suggests strong evidence against the null hypothesis of stationarity for all measuring stations.

\begin{table*}
\begin{center}
\caption{\it Values of the test-statistic \eqref{testwn} for the hypothesis of stationarity of the annual temperature curve data}
\label{pval_data}
{\begin{tabular}{lcccc}
\hline
Measuring Station & T & $M=\lceil T^{1/3}\rceil$ & M=8 & $M= \lfloor T^{1/2}\rfloor$ \\
\hline
%Boulia Airport & 120& 0.00 & 0.00 & 0.00 \\  
%Cape Otway & 149& 0.00 & 0.00 & 0.00 \\  
%Gayndah Post Office & 117  & 0.00 & 0.00 & 0.00 \\  
%  Gunnedah Pool & 133 & 0.00 & 0.00 & 0.00 \\  
% Hobart & 121 & 0.00 & 0.00 & 0.00 \\  
% Melbourne & 158  & 0.00 & 0.00 & 0.00 \\    
% Robe & 129  & 0.00 & 0.00 & 0.00 \\   
% Sydney & 154  &0.00 & 0.00 & 0.00 \\ 
Boulia Airport & 120& 3.21 & 2.95 & 4.55 \\ 
Cape Otway & 149& 3.87 & 4.42 & 4.48 \\ 
Gayndah Post Office & 117  & 3.19 & 4.46 & 4.16 \\ 
 Gunnedah Pool & 133 & 4.33 & 3.72 & 5.04 \\ 
Hobart & 121 & 4.99 & 4.60 & 5.13 \\ 
 Melbourne & 158  & 2.88 & 3.68 & 4.36 \\ 
 Robe & 129  & 2.88 & 2.91 & 3.65 \\ 
Sydney & 154  &3.30 & 3.71 & 4.30 \\ 
\hline
\end{tabular}}
\end{center}
\end{table*}

\bigskip
\medskip

\noindent 
{\bf Acknowledgements.}
This work has been supported in part by the
Collaborative Research Center ``Statistical modeling of nonlinear
dynamic processes'' (SFB 823, Project A1, A7, C1) of the German Research Foundation (DFG). Vaidotas Characiejus would like to acknowledge the support of the Communaut\'e fran\c{c}aise de Belgique, Actions de Recherche Concert\'ees, Projects Consolidation 2016--2021. Anne van Delft gratefully acknowledges financial support by the contract ``Projet d'Actions de Recherche Concert{\'e}es'' No. 12/17-045 of the ``Communaut{\'e} fran\c{c}aise de Belgique''. 
% We are very grateful to two referees and an associate editor for their  constructive comments, which led to substantial improvement of  an earlier version of this manuscript.

\begin{appendices}   
\def\theequation{A.\arabic{equation}}
\setcounter{equation}{1}
\section{Auxiliary results and further proofs}
\label{appendix}

\subsection{Some properties of tensor products of operators} 

Let $\mathcal{H}_i$ for each $i=1,\ldots,n$ be a Hilbert space with inner product $\langle \cdot, \cdot \rangle$. The tensor of these is denoted by
\[
\mathcal{H}:= \mathcal{H}_1 \otimes \ldots \otimes \mathcal{H}_n = \bigotimes_{i=1}^n \mathcal{H}_i
\]

If $\mathcal{H}_i =\mathcal{H}\, \forall i$, then this is the n-th fold tensor product of $\mathcal{H}$. For $A_i \in H_i, 1\le i\le n$ the object $\bigotimes_{i=1}^n A_i$ is a multi-antilinear functional that generates a linear manifold, the usual algebraic tensor product of vector spaces $\mathcal{H}_i$, to which the scalar product  
\[
\innprod{\bigotimes_{i=1}^n A_i}{\bigotimes_{i=1}^n  B_i} = \prod_{i=1}^n\innprod{A_i}{B_i}
\]
can be extended to a pre-Hilbert space. The completion of the above algebraic tensor product is $\bigotimes_{i=1}^n \mathcal{H}_i$. \\
For $A, B, C \in \mathcal L(\mathcal{H})$ we define the following bounded linear mappings. The kronecker product is defined as $(A \widetilde{\bigotimes} B)C = ACB^{\dagger}$, while the transpose Kronecker product is given by $(A \widetilde{\bigotimes}_{\top} B)C = (A \widetilde{\bigotimes} \overline{B})\overline{C}^{\dagger}$. For $A, B, C \in S_2(\mathcal{H})$, we shall denote, in analogy to elements $a,b \in \mathcal{H}$, the Hilbert tensor product as $A \bigotimes B$. We list the following useful properties:
\begin{properties} \label{tensorprop}
Let $\mathcal{H}_i = L^2_\mathbb{C}([0,1]^k)$ for $i = 1,\ldots, n$. Then for $a_i, b_i \in \mathcal{H}_i$ and $A_i, B_i \in S_2(\mathcal{H}_i)$, we have 
\begin{enumerate}
\item $\innprod{A}{B}_{HS} =\Tr({A}{B}^{\dagger})$
\item $\innprod{\bigotimes_{i=1}^n A_i}{\bigotimes_{i=1}^n  B_i}_{{HS}} = \prod_{i=1}^n\innprod{A_i}{B_i}_{HS}$ 
\item $\innprod{a_1\otimes a_2}{b_1\otimes b_2}_{HS} =\innprod{a_1\otimes \overline{a}_2}{b_1\otimes \overline{b_2}}_{\mathcal{H}_i\otimes \mathcal{H}_i} =\innprod{a_1}{b_1}\overline{\innprod{a_2}{b_2}}$
\item If $A_i \in S_1(H)$, then $\prod_{i=1}^{n}\Tr(A_i) = \Tr(\widetilde{\bigotimes}_{i=1}^n  A_i)$
\item $\big((a_1 \otimes \overline{a}_2) \bigotimes (a_3 \otimes \overline{a}_4)\big) =\big((a_1 \otimes {a}_3) \widetilde{\bigotimes} (\overline{a}_2 \otimes \overline{a}_4)\big) =\big((a_1 \otimes {a}_4) \widetilde{\bigotimes}_{\top} {a}_2 \otimes {a}_3)\big)$
\end{enumerate}
\end{properties}

% \subsubsection{Cumulant tensors} \label{subsec:Cumtens}
Let $X$ be a random element on a probability space $(\Omega,\mathcal{A},\mathbb{P})$ that takes values in a separable Hilbert space $H$. More precisely, we endow $H$ with the topology induced by the norm on $H$ and assume that $X:\Omega\to H$ is Borel-measurable. The {\em $k$-th order cumulant tensor} is defined by \citep{vde16}
\begin{align}
\label{cumtens}
\cm \big( X_1, \ldots, X_k\big) 
= \sum_{l_1, \ldots l_k\in \nnum} \cm \big( \innprod{X_1}{\psi_{l_1}},\ldots,\innprod{X_k}{\psi_{l_k}}\big)
 (\psi_{l_1} \otimes \cdots \otimes \psi_{l_k}), 
\end{align}   
and the cumulants on the right hand side are as usual given by
\[
\cm\big(\innprod{X_1}{\psi_{l_1}},\ldots,\innprod{X_k}{\psi_{l_k}}\big)
=\sum_{\nu=(\nu_1,\ldots,\nu_p)}(-1)^{p-1}\,(p-1)!\,\prod_{r=1}^p \E\Big[\prod_{t\in\nu_r}\innprod{X_t}{\psi_{l_t}}\Big],
\]
where the summation extends over all unordered partitions $\nu$ of $\{1,\ldots,k\}$. The product theorem for cumulants \citep[Theorem 2.3.2]{brillinger} can then be generalised \citep[see e.g.][Theorem A.1]{avd16_main} to simple tensors of random elements of $H$, i.e., $X_t=\otimes_{j=1}^{J_t}X_{tj}$ with $j=1,\ldots,J_t$ and $t=1,\ldots,k$. The joint cumulant tensor is then be given by 
\begin{align} \label{prodcumthm} 
\cm(X_1,\ldots,X_k)
=\sum_{\nu=(\nu_1,\ldots,\nu_p)}S_\nu\Big(\otimes_{n=1}^{p}
\cm\big(X_{tj}|(t,j)\in\nu_n\big)\Big),
\end{align}
where $S_\nu$ is the permutation that maps the components of the tensor back into the original order, that is, $S_\nu(\otimes_{r=1}^p\otimes_{(t,j)\in\nu_r} X_{tj})=X_{11}\otimes\cdots\otimes X_{kJ_t}$.

\subsection{Bounds on cumulant tensors of local functional DFT}\label{sec:propcumlfdft}
In this section, we prove \autoref{lem:cumfixu}, \autoref{lem:cumdifu}.
% A direct consequence of \autoref{lem:cumfixu} is the following corollary
%\begin{Corollary} \label{cor:cumbound}
%We have
%$$\bigsnorm{\cm\left(D_N^{u_1,\omega_1},\dots,D_N^{u_k,\omega_k}\right)}_p = O\left(N^{1-k/2}\right)$$
%uniformly in $\omega_1,\ldots,\omega_k$ and $u_1,\ldots, u_k$. Moreover, if $~\sum_{j=1}^{k}\omega_j~ \ne 0\mod 2\pi$ then
% $$\bigsnorm{\cm\left(D_N^{u_1,\omega_1},\dots,D_N^{u_k,\omega_k}\right)}_p =O\left(N^{-k/2}\right).$$ 
%\end{Corollary}
Before we give the proofs, denote the function $\Delta^{(N)}(\omega)=\sum_{t=0}^{N-1}e^{-\im\omega t}$ for $\omega\in\rnum$. This function satisfies $|\Delta^{(N)}(\sum_{j=1}^{k} \omega_{j})|=N$ for any $\omega_1,\ldots,\omega_k$ for which their sum lies on the manifold $\omega \equiv 0 \mod 2\pi$, while it is of reduced magnitude off the manifold. For the canonical frequencies $\omega_k=2\pi k/N$ with $k\in\mathbb Z$, we moreover have
\begin{equation}\label{eq:delta}
\Delta^{(N)}(\omega_k)	=
\begin{cases}
N,&k\in N\mathbb Z;\\
0,&k\in\mathbb Z\setminus N\mathbb Z.
\end{cases}
\end{equation}

\begin{proof}[Proof of \autoref{lem:cumfixu}]
Let $p \in \{1,2\}$. Using linearity of cumulants we write 
\begin{align*} 
&\cm\left(D_N^{u_i,\omega_1},\dots,D_N^{u_i,\omega_k}\right)\\
= & \frac{1}{{(2\pi N)}^{k/2}}\sum_{s_1,\dots,s_k=0}^{N-1}\exp\left(-\im \sum_{j=1}^ks_j\omega_j\right)\cm\left(X_{\lfloor u_iT \rfloor - N/2 +s_1+1,T},\dots, X_{\lfloor u_iT \rfloor - N/2 +s_k+1,T}\right)\\
= &\frac{1}{{(2\pi N)}^{k/2}}\sum_{s_1,\dots,s_k=0}^{N-1}\exp\left(-\im \sum_{j=1}^ks_j\omega_j\right)C_{u_i-\frac{N/2-s_1-1}{T},s_1-s_k,s_2-s_k,\dots,s_{k-1}-s_k} + R_{k,M,N}^1, \tageq \label{eq:lemcumD}
\end{align*}
where $C_{u_i-\frac{N/2-s_1-1}{T}}$ denotes the local cumulant operator 
%on $L^2([0,1]^{k/2})$ 
as defined in \autoref{sec3}. Using Lemma A.2 of \citet{avd16_main} and Assumption \ref{cumglsp} of 
\begin{align*} \bigsnorm{R_{k,M,N}^1}_p \leq &\frac{1}{{(2\pi N)}^{k/2}}\sum_{s_1,\dots,s_k=0}^{N-1} \big(\frac{k}{T}+ \sum_{j=1}^{k-1}\frac{|s_j-s_k|}{T}\big)\snorm{\kappa_{k,s_1-s_k,\dots,s_{k-1}-s_k}}_p \\ & \le 
\frac{1}{{(2\pi N)}^{k/2}}\sum_{s_k=0}^{N-1} \frac{1}{T} \sum_{l_1,\dots,l_{k-1} \in \mathbb{Z}} \big(1+ \sum_{j=1}^{k-1}|l_j|\big)\snorm{\kappa_{k,l_1,\dots,l_{k-1}}}_p= O\left(N^{-k/2} M^{-1} \right).
\end{align*}
In addition, we can write the first term of \eqref{eq:lemcumD} as
\begin{align*}
= &\frac{1}{{(2\pi N)}^{k/2}}\sum_{s_1,\dots,s_k=0}^{N-1}\exp\left(-\im s_k\sum_j\omega_j\right)\exp\left(-\im \sum_{j=1}^{k-1}\omega_j(s_j-s_k)\right)C_{u_i-\frac{N/2-s_k-1}{T},s_1-s_k,s_2-s_k,\dots,s_{k-1}-s_k}\\
= &\frac{(2\pi)^{1-k/2}}{N^{k/2}}\sum_{s=0}^{N-1}e^{-\im s \sum_j\omega_j}\F_{u_i-\frac{N/2-s-1}{T},\omega_1,\dots,\omega_{k-1}}+R_{k,M,N}^2
\end{align*}
where 
\begin{align*}\bigsnorm{R_{k,M,N}^2}_p &\leq  \frac{1}{{(2\pi N)}^{k/2}}\sum_{s=0}^{N-1}\sum_{\substack{j:1,\ldots k-1:\\\vert s_j-s \vert \geq N-1}}\bigsnorm{C_{u-\frac{N/2-s-1}{T},s_1-s,s_2-s,\dots,s_{k-1}-s}}_p\\ &\leq  \frac{1}{{(2\pi N)}^{k/2}}\sum_{s=0}^{N-1}\sum_{\substack{j:1,\ldots k-1:\\\vert s_j-s \vert \geq N-1}}\snorm{\kappa_{k,s_1-s,s_2-s,\dots,s_{k-1}-s}}_p.\\
&\leq   \frac{1}{{(2\pi N)}^{k/2}}\sum_{s=0}^{N-1}\frac{1}{N^2}\sum_{\substack{j:1,\ldots k-1:\\\vert l_j \vert > N}}N^2\snorm{\kappa_{k,l_1,l_2,\dots,l_{k-1}}}_p\\
& \leq  \frac{1}{{(2\pi N)}^{k/2}}\frac{1}{N}\sum_{\substack{j:1,\ldots k-1:\\\vert l_j \vert > N}}\vert l_j\vert^{2} \snorm{\kappa_{k,l_1,l_2,\dots,l_{k-1}}}_p = o\left(N^{-k/2}\right).\end{align*}
Therefore, the cumulants satisfy
\begin{align*} 
\cm\left(D_N^{u_i,\omega_1},\dots,D_N^{u_i,\omega_k}\right)=\frac{(2\pi)^{1-k/2}}{N^{k/2}}\sum_{s=0}^{N-1}e^{-\im s \sum_j\omega_j}\F_{u_i-\frac{N/2-s-1}{T},\omega_1,\dots,\omega_{k-1}}+R_{k,M,N}^1+R_{k,M,N}^2.
\end{align*}
On the manifold $\sum_{j=1}^{k} \omega_{j} \equiv 0 \mod 2\pi$ we have that $e^{-\im s \sum_j\omega_j}=1$. Assumption \ref{cumglsp}(iv) and a Taylor expansion yield 
\begin{align*}
\cm\left(D_N^{u_i,\omega_1},\dots,D_N^{u_i,\omega_k}\right)= &\frac{(2\pi)^{k/2-1}}{N^{k/2-1}}\F_{u_i,\omega_1,\dots,\omega_{k-1}} + R_{k,M,N}^1+ R_{k,M,N}^2 + R_{k,M,N}^3,
\end{align*}
where 
\begin{align*}\snorm{R_{k,M,N}^3}_p & = \bigsnorm{\frac{(2\pi)^{k/2-1}}{N^{k/2}}\sum_{\ell =1}^{2}\sum_{s=0}^{N-1} \Big(\frac{1-N/2+s}{T}\Big)^{\ell}\frac{\partial^{\ell}}{\partial u^{\ell}} \F_{u,\omega_1,\dots,\omega_k}}_p 
\\ & \le \frac{(2\pi)^{k/2-1}}{N^{k/2}}O\Big(\frac{N}{T}+\frac{N}{M^2}\Big)\sum_{\ell =1}^{2}\sup_{u,\omega_1,\dots,\omega_k}\bigsnorm{\frac{\partial^{\ell}}{\partial u^{\ell}} \F_{u,\omega_1,\dots,\omega_k}}_p = O\left(N^{-k/2} \Big(\frac {N}{T}+\frac{N}{M^2}\Big) \right)
\end{align*}
which follows since $\sum_{s=0}^{N-1}\Big(\frac{1-N/2+s}{T}\Big)=\frac{(N-1)(1-N/2+N/2)}{T}=\frac{N}{2T}$ and similarly $\sum_{s=0}^{N-1}\Big(\frac{1-N/2+s}{T}\Big)^{2}=O\big(\frac{N^3}{T^2}\big)$. We additionally note that, off manifold, the first term of \eqref{eq:lemcumD} can be bounded in norm by
\begin{align*}
= &\bigsnorm{\frac{1}{{(2\pi N)}^{k/2}}\sum_{s_1,\dots,s_k=0}^{N-1}\exp\left(-\im s_k\sum_j\omega_j\right)\exp\left(-\im \sum_{j=1}^{k-1}\omega_j(s_j-s_k)\right)C_{u_i-\frac{N/2-s_k-1}{T},s_1-s_k,s_2-s_k,\dots,s_{k-1}-s_k}}_p\\
\le &  \frac{K}{{(2\pi N)}^{k/2}}\sum_{|l_1|\dots,|l_{k-1}|<N} \Big|\sum_{s_k=0}^{N-|l_j^{*}|}e^{-\im s_k\sum_j\omega_j} \Big|\bigsnorm{C_{u_i-\frac{N/2-s_k-1}{T},l_1,l_2,\dots,l_{k-1}}}_p
\\ \le &  \frac{K}{{(2\pi N)}^{k/2}}\sum_{|l_1|\dots,|l_{k-1}|<N} |l_j^{*}|\snorm{\cumbp{k}{l}}_p = O\left(N^{-k/2}\right),
\end{align*}
for some constant $K$ and where $j^{*}=\underset{{j=1,\ldots,k-1}}{\arg \max} |l_j|$. This finishes the proof of Lemma \ref{lem:cumfixu}.
\end{proof}

\begin{proof}[Proof of \autoref{lem:cumdifu}]
Using again the linearity of cumulants we write 
\begin{align*}
\cm & \left(D_N^{u_{j_1},\omega_1},\dots,D_N^{u_{j_k},\omega_k}\right)\\&=  \frac{1}{(2\pi N)^{k/2}}\sum_{s_1,\dots,s_k=0}^{N-1}\exp\left(-\im \sum_{v=1}^ks_v\omega_v\right)\cm\left(X_{\lfloor u_{j_1}T \rfloor - N/2 +s_1+1,T},\dots, X_{\lfloor u_{j_k}T \rfloor - N/2 +s_k+1,T}\right)\\
&=\frac{1}{(2\pi N)^{k/2}}\sum_{s_1,\dots,s_k=0}^{N-1}\exp\left(-\im \sum_{v=1}^ks_v\omega_v\right)C_{u_k',\lfloor u_{j_1}T \rfloor - \lfloor u_{j_k}T \rfloor +s_1-s_k,\dots,\lfloor u_{j_{k-1}}T \rfloor - \lfloor u_kT \rfloor +s_{k-1}-s_k}+ R_{k,M,N}^1
\end{align*}
where $u_k'= u_k -\frac{N/2-s_k-1}{T}$ and where $R_{k,M,N}^1$ is the error term derived in Lemma \ref{lem:cumfixu}. Let \[l_m =\lfloor u_{j_{m}}T \rfloor - \lfloor u_kT \rfloor +s_{m}-s_k \leftrightarrow s_m = t_{j_k} - t_{j_m} +l_m-s_k\qquad m =1,\ldots k-1\] 
Similar to the proof of \autoref{lem:cumfixu}, we note that
\[\sum_{\substack{l_1,l_2, \ldots l_{k-1},\\ |l_m|>N}} \snorm{C_{u_k', l_1,\dots,l_m,\ldots,l_{k-1}}}_p \le \sum_{\substack{l_1,l_2, \ldots l_{k-1}\\, |l_m|>N}} \frac{|l_m|^2}{N^2}\snorm{\kappa_{k,l_1,\dots,l_{k-1}}}_p=O(N^{-2}).\]
From which it follows that if $|l_m|>N$, the term
\begin{align*}
& \bigsnorm{ \frac{1}{(2\pi N)^{k/2}}\sum_{s_k=0}^{N-1}e^{-\im s_k\sum_{v=1}^{k} \omega_v} \sum_{\substack{l_1,l_2, \ldots l_{k-1},\\ |l_m|>N}}e^{-\im \sum_{v=1}^{k-1} (t_{j_v} - t_{j_k} +l_v)\omega_v}  C_{u_k', l_1,\dots,l_{k-1}}}_p
\end{align*}
is bounded by   
\begin{align*}& \frac{1}{(2\pi N)^{k/2}}\sum_{s=0}^{N-1}\sum_{\substack{{l_1,l_2,\dots,l_{k-1}}\\{|l_m| \geq N}}}\snorm{\kappa_{k,l_1,\dots, l_m, \ldots,l_{k-1}}}_2 \\& \le \frac{1}{(2\pi N)^{k/2+2}} \sum_{s=0}^{N-1} \sum_{\substack{{l_1,l_2,\dots,l_{k-1}}\\{|l_1| \geq N}}} |l_m|^2\snorm{\kappa_{k,l_1,\dots, l_m, \ldots,l_{k-1}}}_p  = O\left(N^{-k/2-1}\right). 
\end{align*}
 \end{proof}

%%%%%%%%%%%%%%%%%%%%%%%%%%%%%%%%%%%%%%

%%%%%%%%%%%%%%%%%%%%%%%%%%%%%%%%%%%%%%%%%%%%%%%%%%%%%%%

\subsection{Derivation of covariance %and higher order cumulants
} \label{sec:cov} 
\subsubsection{Covariance structure of $\sqrt{T}\hat{F}_{1,T}$}
\[
T \cv(\hat{F}_{1,T}, \hat{F}_{1,T}) 
= T \cm \Big (\frac{1}{T}\sum_{k_1=1}^{\lfloor N/2 \rfloor} \sum_{j_1=1}^M \innprod{ I_N^{u_{j_1},\omega_{k_1}}}{I_N^{u_{j_1},\omega_{k_1-1}}}_{HS},\frac{1}{T}\sum_{k_2=1}^{\lfloor N/2 \rfloor} \sum_{j_2=1}^M \overline{\innprod{I_N^{u_{j_2},\omega_{k_2}}}{I_N^{u_{j_2},\omega_{k_2-1}}}}_{HS}  \Big ) . 
\]
Using again \autoref{lem:cumHSinprod} 
\begin{align*}
	\cm_2(\sqrt{T}\hat F_{1,T}) 
=\frac{1}{T}\sum_{k_1,k_2=1}^{\lfloor N/2 \rfloor} \sum_{j_1,j_2=1}^M\Tr\Big(\sum_{\boldsymbol{P} = P_1 \cup \ldots \cup P_G}S_{\boldsymbol{P}}\Big(\otimes_{g=1}^{G} \operatorname{cum}\big(\fdft{p}{p}|p\in \nu_g\Big) \Big), \end{align*}
where $p=(l,m)$ with $k_{p} = (-1)^{l-m} k_l-\delta_{\{m \in \{3,4\}\}}$ and $j_p= j_l$ for $l \in\{1, 2\}$ and $m \in \{1,2,3,4\}$ and where $\delta_{\{A\}}$ equals 1 if event $A$ occurs and $0$ otherwise. In particular, we are interested in all indecomposable partitions of the array
\[\begin{matrix}
\underbrace{D_N^{u_{j_1},\omega_{k_1}}}_{1}&\underbrace{D_N^{u_{j_1},-\omega_{k_1}}}_2 & \underbrace{D_N^{u_{j_1},-\omega_{k_1-1}}}_3& \underbrace{D_N^{u_{j_1},\omega_{k_1-1}}}_4\\ 
\underbrace{D_N^{u_{j_2},-\omega_{k_2}}}_5 &\underbrace{D_N^{u_{j_2},\omega_{k_2}}}_6 &\underbrace{D_N^{u_{j_2},\omega_{k_2-1}}}_7& \underbrace{D_N^{u_{j_2},-\omega_{k_2-1}}}_8
\end{matrix}\]
By \autoref{lem:highcumF1}, all partitions of size $G < 3$, will be of lower order. By Proposition \autoref{prop:cumhighF1frequb}), the only partitions that remain are those that contain either one fourth-order cumulant and two second-order cumulants or those consisting only of second-order cumulants. Additionally, Corollary \ref{cor:cumbound} and Lemma \ref{lem:cumdifu} indicate that for the partitions with structure $\cm_4\cm_2 \cm_2$ to be indecomposable there must be at least one restriction in time. More restrictions in terms of frequency would mean the partition term is of lower order. \\

For the structure $\cm_4\cm_2 \cm_2$, the only significant terms are therefore  
\begin{align*}
& \Tr\Big(S_{(1256)(34)(78)}\Big(\delta_{j_1,j_2}\left[(\frac{2 \pi}{N}\F_{u_{j_1},\omega_{k_1},-\omega_{k_1},-\omega_{k_2}}+\Eps_4) \otimes (\F_{u_{j_1},-\omega_{k_1-1}}+\Eps_2) \otimes (\F_{u_{j_2},\omega_{k_2-1}}+\Eps_2)\right] \Big)\\ 
& \Tr\Big(S_{(1278)(34)(56)}\Big(\delta_{j_1,j_2}\left[ (\frac{2 \pi}{N}\F_{u_{j_1},\omega_{k_1},-\omega_{k_1},\omega_{k_2-1}}+\Eps_4) \otimes (\F_{u_{j_1},-\omega_{k_1-1}}+\Eps_2) \otimes ( \F_{u_{j_2},-\omega_{k_2}}+\Eps_2) \right]\Big)\\ 
&\Tr\Big(S_{(3456)(12)(78)}\Big(\delta_{j_1,j_2}\left[(\frac{2 \pi}{N}\F_{u_{j_1},-\omega_{k_1-1},\omega_{k_1-1},-\omega_{k_2}}+\Eps_4)\otimes (\F_{u_{j_1},\omega_{k_1}}+\Eps_2) \otimes(\F_{u_{j_2},\omega_{k_2-1}}+\Eps_2)\right] \Big) 
\\ 
& \Tr\Big(S_{(3478)(12)(56)} \Big(\delta_{j_1,j_2}\left[(\frac{2 \pi}{N}\F_{u_{j_1},-\omega_{k_1-1},\omega_{k_1-1},\omega_{k_2-1}}+\Eps_4) \otimes ( \F_{u_{j_1},\omega_{k_1}}+\Eps_2) \otimes (\F_{u_{j_2},-\omega_{k_2}}+\Eps_2)\right]\Big), 
\end{align*}
where $\Eps_k$ denotes an operator on $L^2_\mathbb{C}([0,1]^{\lfloor k/2 \rfloor})$ that satisfies $\snorm{\Eps_k}_{1}=O(N^{-k/2}\times\frac{N}{M^2})$. 
For the partitions with structure $\cm_2 \cm_2 \cm_2 \cm_2$, there must be at least one restriction in terms of time and frequency for the partition to be indecomposable. Those with more than the minimum restrictions are of lower order. For the structure $\cm_2 \cm_2 \cm_2 \cm_2$, the significant indecomposable partitions  are 
\begin{align*}
\Tr\Big(S_{(12)(37)(56)(48)}&\Big(\delta_{j_1,j_2} \delta_{k_1,k_2} \big[\F_{u_{j_1},\omega_{k_1}}\otimes \F_{u_{j_1},-\omega_{k_1-1}}\otimes \F_{u_{j_2},-\omega_{k_2}} \otimes \F_{u_{j_1},\omega_{k_1-1}}+\Eps_2\big]\Big)\\ 
\Tr\Big(S_{(12)(36)(78)(45)}& \Big(\delta_{j_1,j_2} \delta_{k_1-1,k_2} \big[\F_{u_{j_1},\omega_{k_1}}\otimes\F_{u_{j_1},-\omega_{k_1-1}}\otimes \F_{u_{j_2},\omega_{k_2-1}} \otimes \F_{u_{j_1},\omega_{k_1-1}}+\Eps_2\big]\Big)\\  
\Tr\Big(S_{(15)(26)(37)(48)}&\Big(\delta_{j_1,j_2} \delta_{k_1,k_2} \big[\F_{u_{j_1},\omega_{k_1}}\otimes \F_{u_{j_1},-\omega_{k_1}}\otimes \F_{u_{j_1},-\omega_{k_1-1}} \otimes \F_{u_{j_1},\omega_{k_1-1}}+\Eps_2\big]\Big) \\ 
\Tr\Big(S_{(15)(26)(34)(78)}&\Big(\delta_{j_1,j_2} \delta_{k_1,k_2} \big[\F_{u_{j_1},\omega_{k_1}}\otimes \F_{u_{j_1},-\omega_{k_1}}\otimes \F_{u_{j_1},-\omega_{k_1-1}}\otimes \F_{u_{j_2},\omega_{k_2-1}}+\Eps_2\big]\Big)  \\ 
\Tr\Big(S_{(18)(27)(34)(56)}&\Big(\delta_{j_1,j_2} \delta_{k_1,k_2-1} \big[\F_{u_{j_1},\omega_{k_1}}\otimes \F_{u_{j_1},-\omega_{k_1}}\otimes \F_{u_{j_1},-\omega_{k_1-1}}\otimes \F_{u_{j_2},-\omega_{k_2}}+\epsilon_2\big]\Big).
\end{align*}
Using Remark \autoref{rem:2ndstruc} and in particular equation  \eqref{eq:2ndcum4} and  \eqref{eq:2ndcum2} below, the corresponding terms of the covariance equal
\begin{align*}
	T \cv(\hat{F}_{1,T},& \hat{F}_{1,T})
	=\frac{1}{T}\sum_{k_1,k_2=1}^{\lfloor N/2 \rfloor} \sum_{j_1,j_2=1}^M\delta_{j_1,j_2}\left[ \biginnprod{\frac{2 \pi}{N}\F_{u_{j_1},\omega_{k_1},-\omega_{k_1},-\omega_{k_2}}}{\F_{u_{j_1},\omega_{k_1-1}} \bigotimes \F_{u_{j_2},\omega_{k_2-1}} }_{HS} +O(\frac{1}{T})\right]\\
	&+\frac{1}{T}\sum_{k_1,k_2=1}^{\lfloor N/2 \rfloor} \sum_{j_1,j_2=1}^M\delta_{j_1,j_2}\left[ \biginnprod{ \frac{2 \pi}{N}\F_{u_{j_1},\omega_{k_1},-\omega_{k_1},\omega_{k_2-1}}}{ \F_{u_{j_1},\omega_{k_1-1}}\bigotimes \F_{u_{j_2},-\omega_{k_2}}}_{HS} +O(\frac{1}{T})\right]\\
	&+\frac{1}{T}\sum_{k_1,k_2=1}^{\lfloor N/2 \rfloor} \sum_{j_1,j_2=1}^M\delta_{j_1,j_2}\left[ \biginnprod{\frac{2 \pi}{N}\F_{u_{j_1},-\omega_{k_1-1},\omega_{k_1-1},-\omega_{k_2}}}{\F_{u_{j_1},-\omega_{k_1}} \bigotimes \F_{u_{j_2},\omega_{k_2-1}}}_{HS} +O(\frac{1}{T})\right]\\
	&+\frac{1}{T}\sum_{k_1,k_2=1}^{\lfloor N/2 \rfloor} \sum_{j_1,j_2=1}^M\delta_{j_1,j_2}\left[ \biginnprod{(\frac{2 \pi}{N}\F_{u_{j_1},-\omega_{k_1-1},\omega_{k_1-1},\omega_{k_2-1}}}{\F_{u_{j_1},-\omega_{k_1}} \bigotimes \F_{u_{j_2},-\omega_{k_2}}}_{HS} +O(\frac{1}{T})\right]\\
	&+\frac{1}{T}\sum_{k_1,k_2=1}^{\lfloor N/2 \rfloor} \sum_{j_1,j_2=1}^M\delta_{j_1,j_2} \delta_{k_1,k_2} \left[ 
\innprod{{\F^{\dagger}_{u_{j_1},-\omega_{k_1}}} \F_{u_{j_1},-\omega_{k_1-1}}}{{\F_{u_{j_1},-\omega_{k_1-1}}}\F^{\dagger}_{u_{j_2},-\omega_{k_2}}} +O(\frac{1}{M^2})\right]\\
	&+\frac{1}{T}\sum_{k_1,k_2=1}^{\lfloor N/2 \rfloor} \sum_{j_1,j_2=1}^M\delta_{j_1,j_2} \delta_{k_1-1,k_2}  \left[ \innprod{{\F^{\dagger}_{u_{j_1},-\omega_{k_1}}}\F_{u_{j_1},-\omega_{k_1-1}}}{\F_{u_{j_1},-\omega_{k_1-1}}\F_{u_{j_2},-\omega_{k_2-1}}}_{HS}
 +O(\frac{1}{M^2})\right]
 \\& +\frac{1}{T}\sum_{k_1,k_2=1}^{\lfloor N/2 \rfloor} \sum_{j_1,j_2=1}^M\delta_{j_1,j_2} \delta_{k_1,k_2}   \left[ \innprod{\F_{u_{j_1},\omega_{k_1}}}{\F_{u_{j_1},\omega_{k_1-1}}}_{HS}\innprod{ \F_{u_{j_1},-\omega_{k_1}}}{ \F_{u_{j_1},-\omega_{k_1-1}}}_{HS}
 +O(\frac{1}{M^2})\right]
  \\& +\frac{1}{T}\sum_{k_1,k_2=1}^{\lfloor N/2 \rfloor} \sum_{j_1,j_2=1}^M\delta_{j_1,j_2} \delta_{k_1,k_2}   \left[\innprod{\F_{u_{j_1},\omega_{k_1}} \widetilde{\bigotimes} \F_{u_{j_1},\omega_{k_1}}}{\F_{u_{j_1},\omega_{k_1-1}} \bigotimes \F_{u_{j_2},\omega_{k_2-1}}}_{HS}
 +O(\frac{1}{M^2})\right]
   \\& +\frac{1}{T}\sum_{k_1,k_2=1}^{\lfloor N/2 \rfloor} \sum_{j_1,j_2=1}^M
 \delta_{j_1,j_2} \delta_{k_1,k_2-1}\left[\innprod{\F_{u_{j_1},\omega_{k_1}} \widetilde{\bigotimes}_{\top}  \F_{u_{j_1},-\omega_{k_1}}}{ \F_{u_{j_1},\omega_{k_1-1}} \bigotimes \F_{u_{j_2},-\omega_{k_2}}}_{HS}+O(\frac{1}{M^2})\right]
\end{align*}

So that, as $N, M \to \infty$,
\begin{align*}
NM \cv (\hat{F}_{1,T}, \hat{F}_{1,T})  \to 
&  \frac{2}{8\pi} \int_{-\pi}^{\pi} \int_{-\pi}^{\pi}\int_0^{1} 
\biginnprod{\F_{u,\omega_{1},-\omega_{1},-\omega_{2}}}{\F_{u,\omega_{1}} \bigotimes \F_{u,\omega_{2}} }_{HS}du d\omega_1 d\omega_2 \\ 
&+   \frac{2}{8\pi} \int_{-\pi}^{\pi} \int_{-\pi}^{\pi}\int_0^{1} \biginnprod{ \F_{u,\omega_{1},-\omega_{1},\omega_{2}}}{ \F_{u,\omega_{1}}\bigotimes \F_{u,-\omega_{2}}}_{HS}  du d\omega_1 d\omega_2 \\ 
 &+  \frac{2}{4\pi} \int_{-\pi}^{\pi} \int_0^{1} \snorm{\F^2_{u,\omega}}^2_2 du d\omega \\ 
  &+  \frac{1}{4\pi} \int_{-\pi}^{\pi} \int_0^{1}\snorm{\F_{u,\omega}}^4_2 du d\omega \\ 
 &+  \frac{1}{4\pi} \int_{-\pi}^{\pi} \int_0^{1}   \innprod{\F_{u,\omega} \widetilde{\bigotimes} \F_{u,\omega}}{\F_{u,\omega} \bigotimes \F_{u,\omega}}_{HS}du d\omega \\ 
 &+  \frac{1}{4\pi} \int_{-\pi}^{\pi} \int_0^{1} \innprod{\F_{u,\omega} \widetilde{\bigotimes}_{\top} \F_{u,-\omega}}{ \F_{u,\omega}\bigotimes\F_{u,-\omega}}_{HS} du d\omega 
\end{align*}
where we used the self-adjointness of the spectral density operator and that, for any function $g: \mathbb{R} \to \mathbb{K}$, we have $\int_{-\pi}^{\pi} g(\omega) d\omega=\int_{-\pi}^{\pi} g(-\omega) d\omega $. (From this it follows that the term 1,4 ;2,3 and 5,6 are respectively equal in the limit). 

\begin{Remark}[\textbf{A note on the permutation for the 2nd order cumulants}]\label{rem:2ndstruc}
%%%%%%%%%%%%%%%%%%%%%%%%%%%%%%%%%%%%%%%%%%%%%
%%%%%%%%%%%%%%%%%%%%%%%%%%%%%%%5
{\rm In order to give meaning to the covariance structure, we need to investigate how it `operates' as a result of the permutation that occurs due to the cumulant operation. For the second order cumulant structure, Theorem \ref{lem:cumHSinprod}] implies that the original order of the simple tensors has structure $\Tr(S_{1234}\cdot \otimes \cdot \otimes \cdot \otimes \cdot \widetilde{\bigotimes} S_{5678} \cdot \otimes \cdot \otimes \cdot \otimes \cdot )$ which leads to the following correspondence of simple tensors
\[
\begin{matrix}
1 &\leftrightarrow& 3 \\
2 &\leftrightarrow & 4\\
5 &\leftrightarrow & 7 \\
6 &\leftrightarrow & 8 \\ 
\end{matrix}
\]
Let $X \in  \mathcal{H}^{{\otimes^4}}$ and $Y, Z,  X \in  \mathcal{H}^{\otimes^2}$, then using properties \autoref{tensorprop} we find
\begin{align*}
 \Tr\Big(S_{(1256)(12)(56)} X \otimes Y \otimes Z)&= \Tr\big( X (\overline{Y} \bigotimes {Z}) ^{\dagger}\big) =\innprod{X}{\overline{Y} \bigotimes {Z}}_{HS}\\
 \Tr\Big(S_{(1256)(12)(56)} X \otimes Y \otimes Z)&= \Tr\big( X (\overline{Y} \bigotimes {Z}) ^{\dagger}\big) =\innprod{X}{\overline{Y} \bigotimes {Z}}_{HS}\\
\Tr\Big(S_{(1256)(12)(56)} X \otimes Y \otimes Z)&=  \Tr\big( X (\overline{Y} \bigotimes {Z}) ^{\dagger}\big) =\innprod{X}{\overline{Y} \bigotimes {Z}}_{HS}\\
 \Tr\Big(S_{(1256)(12)(56)}X \otimes Y \otimes Z)&= \Tr\big( X (\overline{Y} \bigotimes {Z}) ^{\dagger}\big) =\innprod{X}{\overline{Y} \bigotimes {Z}}_{HS}\tageq \label{eq:2ndcum4}\\
\end{align*}
and for $W, X, Y, Z \in X \in \mathcal{H}^{\otimes^2}$
\begin{align*}
\Tr\Big(S_{(12)(15)(56)(26)} W \otimes  X \otimes Y \otimes Z\Big) &%= \Tr(S_{(21)(15)(56)(62)}\overline{W}^{\dagger}XY\overline{Z}^{\dagger}) = \Tr(\overline{W}^{\dagger}XY\overline{Z}^{\dagger}) 
=\innprod{\overline{W}^{\dagger}X}{\overline{Z}Y^{\dagger}}_{HS}\\ 
\Tr\Big(S_{(12)(16)(56)(25)} W \otimes  X \otimes Y \otimes Z\Big)&=% \Tr\Big(S_{(21)(16)(65)(52)}\overline{W}^{\dagger}  X \overline{Y}^{\dagger} \overline{Z}^{\dagger}\Big)
 \innprod{\overline{W}^{\dagger}X}{\overline{Z}\overline{Y}}_{HS}\\
\Tr\Big(S_{(15)(26)(15)(26)} W \otimes  X \otimes Y \otimes Z\Big)&= 
%\Tr( (W\overline{Y}^{\dagger} \widetilde{\bigotimes} X\overline{Z}^{\dagger}))=
\innprod{W}{\overline{Y}}_{HS}\innprod{X}{\overline{Z}}_{HS}\\ 
\Tr\Big(S_{(15)(26)(12)(56)} W \otimes  X \otimes Y \otimes Z\Big) &=\innprod{W \widetilde{\bigotimes} \overline{X}}{(\overline{Y} \bigotimes Z)}_{HS} \\ 
\Tr\Big(S_{(16)(25)(12)(56)}  W \otimes  X \otimes Y \otimes Z\Big)&=\innprod{W \widetilde{\bigotimes}_{\top} {X}}{(\overline{Y} \bigotimes Z)}_{HS} \tageq \label{eq:2ndcum2}
\end{align*}
}
\end{Remark}

\subsubsection{Covariance structure of $\sqrt{T}\hat{F}_{2,T}$}
\[
T \cv(\hat{F}_{2,T}, \hat{F}_{2,T}) 
= T \cm(\frac{1}{N M^2}\sum_{k_1=1}^{\lfloor N/2 \rfloor} \sum_{j_1,j_2=1}^M \innprod{ I_N^{u_{j_1},\omega_{k_1}}}{I_N^{u_{j_2},\omega_{k_1}}}_{HS},\frac{1}{N M^2}\sum_{k_2=1}^{\lfloor N/2 \rfloor} \sum_{j_3,j_4=1}^M \overline{\innprod{I_N^{u_{j_3},\omega_{k_2}}}{I_N^{u_{j_4},\omega_{k_2}}}}_{HS}) 
\]
Using again \autoref{lem:cumHSinprod} 
\begin{align*}
	\cm_2(\hat F_{2,T}) 
 =\frac{1}{N^2 M^4}\sum_{k_1,k_2=1}^{\lfloor N/2 \rfloor} \sum_{\substack{j_1,j_2,\\j_3, j_4=1}}^M\Tr\Big(\sum_{\boldsymbol{P} = P_1 \cup \ldots \cup P_G}S_{\boldsymbol{P}}\Big(\otimes_{g=1}^{G} \operatorname{cum}\big(\fdft{p}{p}|p\in \nu_g\Big)\Big)
\end{align*}
where $p=(l,m)$ with $k_{p} = (-1)^{l-m} k_{l}$ and $j_p= j_{2l-\delta_{\{m \in \{1,2\}\}}}$ for $l \in\{1, 2\}$ and $m \in \{1,2,3,4\}$ and where $\delta_{\{A\}}$ equals 1 if event $A$ occurs and $0$ otherwise. That is, we are interested in all indecomposable partitions of the array
\[\begin{matrix}
\underbrace{\fdft{1}{1}}_{1}&\underbrace{\fdftc{1}{1}}_2 & \underbrace{\fdftc{2}{1}}_3& \underbrace{\fdft{2}{1}}_4\\ 
\underbrace{\fdftc{3}{2}}_5 &\underbrace{\fdft{3}{2}}_6 &\underbrace{\fdft{4}{2}}_7& \underbrace{\fdftc{4}{2}}_8
\end{matrix}\]
For the same reason as above, we only have to consider the structures $\cm_4\cm_2 \cm_2$ and $\cm_2 \cm_2 \cm_2 \cm_2$. For the structure $\cm_4\cm_2 \cm_2$, the only significant terms are again
\begin{align*}
& \Tr\Big(S_{(1256)(34)(78)}\Big(\delta_{j_1,j_3}\left[(\frac{2 \pi}{N}\F_{u_{j_1},\omega_{k_1},-\omega_{k_1},-\omega_{k_2}}+\Eps_4) \otimes (\F_{u_{j_2},-\omega_{k_1}}+\Eps_2) \otimes (\F_{u_{j_4},\omega_{k_2}}+\Eps_2)\right] \Big)\\ 
& \Tr\Big(S_{(1278)(34)(56)}\Big(\delta_{j_1,j_4}\left[ (\frac{2 \pi}{N}\F_{u_{j_1},\omega_{k_1},-\omega_{k_1},\omega_{k_2}}+\Eps_4) \otimes (\F_{u_{j_2},-\omega_{k_1}}+\Eps_2) \otimes ( \F_{u_{j_3},-\omega_{k_2}}+\Eps_2) \right]\Big)\\ 
&\Tr\Big(S_{(3456)(12)(78)}\Big(\delta_{j_2,j_3}\left[(\frac{2 \pi}{N}\F_{u_{j_2},-\omega_{k_1},\omega_{k_1},-\omega_{k_2}}+\Eps_4)\otimes (\F_{u_{j_1},\omega_{k_1}}+\Eps_2) \otimes(\F_{u_{j_4},\omega_{k_2}}+\Eps_2)\right] \Big)  \\ 
& \Tr\Big(S_{(3478)(12)(56)} \Big(\delta_{j_2,j_4}\left[(\frac{2 \pi}{N}\F_{u_{j_2},-\omega_{k_1},\omega_{k_1},\omega_{k_2}}+\Eps_4) \otimes ( \F_{u_{j_1},\omega_{k_1}}+\Eps_2) \otimes (\F_{u_{j_3},-\omega_{k_2}}+\Eps_2\right]\Big)  
\end{align*}
For the structure $\cm_2 \cm_2 \cm_2 \cm_2$, the only significant terms are in this case
\begin{align*}
&\Tr\Big(S_{(3478)(12)(56)}\big(\delta_{j_2,j_4} \delta_{k_1,k_2} \big[\F_{u_{j_1},\omega_{k_1}} \otimes \F_{u_{j_2},-\omega_{k_1}} \otimes \F_{u_{j_3},-\omega_{k_2}} \otimes \F_{u_{j_2},\omega_{k_1}}+\Eps_2\big]\big)\Big)\\ 
& \Tr\Big(S_{(12)(36)(78)(45)}\big(\delta_{j_2,j_3} \delta_{k_1,k_2} \big[\F_{u_{j_1},\omega_{k_1}} \otimes \F_{u_{j_2},-\omega_{k_1}} \otimes \F_{u_{j_4},\omega_{k_2}} \otimes \F_{u_{j_2},\omega_{k_1}} +\Eps_2\big]\big)\Big)\\  
 &\Tr\Big(S_{(15)(26)(34)(78)}\big(\delta_{j_1,j_3} \delta_{k_1,k_2} \big[\F_{u_{j_1},\omega_{k_1}}\otimes \F_{u_{j_1},-\omega_{k_1}}\otimes \F_{u_{j_2},-\omega_{k_1}}\otimes \F_{u_{j_4},\omega_{k_2}}+\Eps_2\big]\big)\Big)
\\ &\Tr\Big(S_{(18)(27)(34)(56)}\big(\delta_{j_1,j_4} \delta_{k_1,k_2} \big[\F_{u_{j_1},\omega_{k_1}}\otimes\F_{u_{j_1},-\omega_{k_1}}\otimes \F_{u_{j_2},-\omega_{k_1}} \otimes \F_{u_{j_3},-\omega_{k_2}}+\Eps_2\big]\big)\Big).
\end{align*}
Using Remark \autoref{rem:2ndstruc}, we find
\begin{align*}
	\cv(\sqrt{T}\hat F_{2,T}) 
& =\frac{1}{N M^3}\sum_{k_1,k_2=1}^{\lfloor N/2 \rfloor} \sum_{\substack{j_1,j_2,\\j_3, j_4=1}}^M \delta_{j_1,j_3}\left[ \innprod{\frac{2 \pi}{N}\F_{u_{j_1},\omega_{k_1},-\omega_{k_1},-\omega_{k_2}}}{\overline{\F_{u_{j_2},-\omega_{k_1}}} \bigotimes \F_{u_{j_4},\omega_{k_2}}}_{HS}  +O(\frac{1}{T}) \right]\\
 &+\frac{1}{N M^3}\sum_{k_1,k_2=1}^{\lfloor N/2 \rfloor} \sum_{\substack{j_1,j_2,\\j_3, j_4=1}}^M \delta_{j_1,j_4}\left[ \innprod{\frac{2 \pi}{N}\F_{u_{j_1},\omega_{k_1},-\omega_{k_1},\omega_{k_2}}}{\overline{\F_{u_{j_2},-\omega_{k_1}}} \bigotimes { \F_{u_{j_3},-\omega_{k_2}}}}_{HS} +O(\frac{1}{T})  \right]
 \\&
  +\frac{1}{N M^3}\sum_{k_1,k_2=1}^{\lfloor N/2 \rfloor} \sum_{\substack{j_1,j_2,\\j_3, j_4=1}}^M \delta_{j_2,j_3} \left[ \innprod{\frac{2 \pi}{N}\F_{u_{j_2},-\omega_{k_1},\omega_{k_1},-\omega_{k_2}}}{\overline{\F_{u_{j_1},\omega_{k_1}}} \bigotimes {\F_{u_{j_4},\omega_{k_2}}}}_{HS} +O(\frac{1}{T})  \right]
  \\&
    +\frac{1}{N M^3}\sum_{k_1,k_2=1}^{\lfloor N/2 \rfloor} \sum_{\substack{j_1,j_2,\\j_3, j_4=1}}^M \delta_{j_2,j_4}\left[ \innprod{(\frac{2 \pi}{N}\F_{u_{j_2},-\omega_{k_1},\omega_{k_1},\omega_{k_2}}}{\overline{\F_{u_{j_1},\omega_{k_1}}} \bigotimes {\F_{u_{j_3},-\omega_{k_2}}}}_{HS}+O(\frac{1}{T})  \right]
    \\&
  +\frac{1}{N M^3}\sum_{k_1,k_2=1}^{\lfloor N/2 \rfloor} \sum_{\substack{j_1,j_2,\\j_3, j_4=1}}^M \delta_{j_2,j_4} \delta_{k_1,k_2} \left[  \innprod{\overline{\F_{u_{j_1},\omega_{k_1}} }^{\dagger}\F_{u_{j_2},-\omega_{k_1}}}{\overline{\F_{u_{j_2},\omega_{k_1}}} \F_{u_{j_3},-\omega_{k_2}}^{\dagger}}_{HS} +O(\frac{1}{M^2})\right]
  \\&
  +\frac{1}{N M^3}\sum_{k_1,k_2=1}^{\lfloor N/2 \rfloor} \sum_{\substack{j_1,j_2,\\j_3, j_4=1}}^M \delta_{j_2,j_3} \delta_{k_1,k_2} \left[\innprod{\overline{\F_{u_{j_1},\omega_{k_1}} }^{\dagger}\F_{u_{j_2},-\omega_{k_1}}}{\overline{\F_{u_{j_2},\omega_{k_1}} }\overline{\F_{u_{j_4},\omega_{k_2}}}}_{HS}+O(\frac{1}{M^2})  \right]
 \\&
  +\frac{1}{N M^3}\sum_{k_1,k_2=1}^{\lfloor N/2 \rfloor} \sum_{\substack{j_1,j_2,\\j_3, j_4=1}}^M \delta_{j_1,j_3} \delta_{k_1,k_2} \left[\innprod{\F_{u_{j_1},\omega_{k_1}} \widetilde{\bigotimes} \overline{\F_{u_{j_1},-\omega_{k_1}}}}{\overline{ \F_{u_{j_2},-\omega_{k_1}}} \bigotimes \F_{u_{j_4},\omega_{k_2}}}_{HS} +O(\frac{1}{M^2}) \right]\\&
    +\frac{1}{N M^3}\sum_{k_1,k_2=1}^{\lfloor N/2 \rfloor} \sum_{\substack{j_1,j_2,\\j_3, j_4=1}}^M \delta_{j_1,j_4} \delta_{k_1,k_2} \left[  \innprod{\F_{u_{j_1},\omega_{k_1}} \widetilde{\bigotimes}_{\top} {\F_{u_{j_1},-\omega_{k_1}}}}{\overline{ \F_{u_{j_2},-\omega_{k_1}}} \bigotimes \F_{u_{j_3},-\omega_{k_2}}}_{HS}+O(\frac{1}{M^2}) \right]
\end{align*}

So that, as $N, M \to \infty$,
\begin{align*}
NM \cv (\hat{F}_{2,T}, \hat{F}_{2,T})  \to 
&  \frac{2}{8\pi} \int_{-\pi}^{\pi} \int_{-\pi}^{\pi}\innprod{\widetilde{\F}_{\omega_1,-\omega_{1},-\omega_{2}}}{{\widetilde{\F}_{\omega_{1}}} \bigotimes \widetilde{\F}_{\omega_{2}}}_{HS}  d\omega_1 d\omega_2 \\ 
&  \frac{2}{8\pi} \int_{-\pi}^{\pi} \int_{-\pi}^{\pi}\innprod{\widetilde{\F}_{\omega_1,-\omega_{1},\omega_{2}}}{{\widetilde{\F}_{\omega_{1}}} \bigotimes \widetilde{\F}_{-\omega_{2}}}_{HS}  d\omega_1 d\omega_2 \\  
 &+  \frac{1}{4\pi} \int_{-\pi}^{\pi} \int_0^{1}  \innprod{\F^2_{u,\omega }}{\F_{u,\omega} \widetilde{\F}_{\omega}}_{HS}du d\omega \\ 
  &+  \frac{1}{4\pi} \int_{-\pi}^{\pi} \int_0^{1} \innprod{{\widetilde{\F}_{\omega} }\F_{u,\omega}}{\F_{u,\omega} {\widetilde{\F}_{u,\omega}}}_{HS}du d\omega \\ 
 &+  \frac{1}{4\pi} \int_{-\pi}^{\pi} \int_0^{1}\innprod{\F_{u,\omega} \widetilde{\bigotimes} \F_{u,\omega}}{ \widetilde{\F}_{{\omega}} \bigotimes \widetilde{\F}_{\omega}}_{HS}  du d\omega \\ 
 &+  \frac{1}{4\pi} \int_{-\pi}^{\pi} \int_0^{1}\innprod{\F_{u,\omega} \widetilde{\bigotimes}_{\top} \F_{u,-\omega}}{ \widetilde{\F}_{{\omega}} \bigotimes \widetilde{\F}_{-\omega}}_{HS}  du d\omega \\ 
\end{align*}

\subsubsection{Cross-covariance $\hat{F}_{1,T}$ and $\hat{F}_{2,T}$}

Using again \autoref{lem:cumHSinprod} 
\begin{align*}
T\cm_2(\hat F_{1,T},\hat F_{2,T}) 
=\frac{T}{N^2 M^3}\sum_{k_1,k_2=1}^{\lfloor N/2 \rfloor} \sum_{\substack{j_1,j_2,j_3=1}}^M \Tr\Big(\sum_{\boldsymbol{P} = P_1 \cup \ldots \cup P_G}S_{\boldsymbol{P}}\Big(\otimes_{g=1}^{G} \operatorname{cum}\big(\fdft{p}{p}|p\in \nu_g\Big)\Big),
\end{align*}
where this time we are interested in all indecomposable partitions of the array
\[\begin{matrix}
\underbrace{\fdft{1}{1}}_{1}&\underbrace{\fdftc{1}{1}}_2 & \underbrace{D_N^{u_{j_1},-\omega_{k_1-1}}}_3& \underbrace{D_N^{u_{j_1},\omega_{k_1-1}}}_4\\ 
\underbrace{\fdftc{2}{2}}_5 &\underbrace{\fdft{2}{2}}_6 &\underbrace{\fdft{3}{2}}_7& \underbrace{\fdftc{3}{2}}_8
\end{matrix}\]
By \autoref{lem:highcumF1} and Proposition \autoref{prop:cumhighF1frequb}), we only have to consider partitions of the form $\cm_4\cm_2 \cm_2$ and $\cm_2 \cm_2 \cm_2 \cm_2$. The only significant terms of first form are again
\begin{align*}
\Tr\Big(S_{(1256)(34)(78)}\delta_{j_1,j_2}\left[(\frac{2 \pi}{N}\F_{u_{j_1},\omega_{k_1},-\omega_{k_1},-\omega_{k_2}}+\Eps_4) \otimes(\F_{u_{j_1},-\omega_{k_1-1}}+\Eps_2) \otimes ( \F_{u_{j_3},\omega_{k_2}}+\Eps_2)\right]\Big) \\ 
\Tr\Big(S_{(1278)(34)(56)}\delta_{j_1,j_3}\left[\frac{2 \pi}{N}\F_{u_{j_1},\omega_{k_1},-\omega_{k_1},\omega_{k_2}}+\Eps_4) \otimes(\F_{u_{j_1},-\omega_{k_1-1}}+\Eps_2)\otimes (\F_{u_{j_2},-\omega_{k_2}}+\Eps_2)\right]\Big) \\ 
\Tr\Big(S_{(3456)(12)(78)}\delta_{j_1,j_2}\left[\frac{2 \pi}{N}\F_{u_{j_1},-\omega_{k_1-1},\omega_{k_1-1},-\omega_{k_2}}+\Eps_4) \otimes (\F_{u_{j_1},\omega_{k_1}}+\Eps_2) \otimes (\F_{u_{j_3},\omega_{k_2}}+\Eps_2)\right] \Big)\\ 
\Tr\Big(S_{(3478)(12)(56)}\delta_{j_1,j_3}\left[(\frac{2 \pi}{N}\F_{u_{j_1},-\omega_{k_1-1},\omega_{k_1-1},\omega_{k_2}}+\Eps_4) \otimes (\F_{u_{j_1},\omega_{k_1}}+\Eps_2)\otimes (\F_{u_{j_2},-\omega_{k_2}}+\Eps_2) \right]  \Big)
\end{align*}
and for the structure $\cm_2 \cm_2 \cm_2 \cm_2$, the only significant terms are 
\begin{align*}
&\Tr\Big(S_{(12)(37)(56)(48)}\delta_{j_1,j_3} \delta_{k_1-1,k_2}\delta_{j_1,j_3} \delta_{k_1-1,k_2}  \big( \F_{u_{j_1},\omega_{k_1}}\otimes \F_{u_{j_1},-\omega_{k_1-1}} \otimes \F_{u_{j_2},-\omega_{k_2}}\otimes \F_{u_{j_1},\omega_{k_1-1}}+\Eps_2\big]\Big)\\
&\Tr\Big(S_{(12)(36)(78)(45)}\delta_{j_1,j_2} \delta_{k_1-1,k_2}\delta_{j_1,j_2} \delta_{k_1-1,k_2}\delta_{j_1,j_2} \delta_{k_1-1,k_2} \big[\F_{u_{j_1},\omega_{k_1}}\otimes \F_{u_{j_1},-\omega_{k_1-1}}\F_{u_{j_3},\omega_{k_2}}\otimes \F_{u_{j_1},\omega_{k_1-1}}+\Eps_2\big]\Big)\\   
&\Tr\Big(S_{(15)(26)(34)(78)}\delta_{j_1,j_2}\delta_{k_1,k_2})^2\delta_{j_1,j_2} \delta_{k_1,k_2} \big[\F_{u_{j_1},\omega_{k_1}}\otimes \F_{u_{j_1},-\omega_{k_1}}\otimes \F_{u_{j_1},-\omega_{k_1-1}}\otimes \F_{u_{j_2},\omega_{k_2}}+\Eps_2\big]\Big)  \\ 
&\Tr\Big(S_{(18)(27)(34)(56)}\delta_{j_1,j_3}\delta_{k_1,k_2})^2 \delta_{j_1,j_3} \delta_{k_1,k_2} \big[\F_{u_{j_1},\omega_{k_1}}\otimes\F_{u_{j_1},-\omega_{k_1}}\otimes \F_{u_{j_1},-\omega_{k_1-1}}\otimes \F_{u_{j_2},-\omega_{k_2}}+\Eps_2\big]\Big)
\end{align*}
which implies by Remark \autoref{rem:2ndstruc}
\begin{align*}
T&\cv(\hat F_{1,T},\hat F_{2,T}) \\
& =\frac{1}{N M^2}\sum_{k_1,k_2=1}^{\lfloor N/2 \rfloor} \sum_{\substack{j_1,j_2,j_3=1}}^M \delta_{j_1,j_2} \left[  \innprod{\frac{2 \pi}{N}\F_{u_{j_1},\omega_{k_1},-\omega_{k_1},-\omega_{k_2}}}{\overline{\F_{u_{j_1},-\omega_{k_1-1}}} \bigotimes { \F_{u_{j_3},\omega_{k_2}}}}_{HS}+O(\frac{1}{T}) \right]
\\ &
+\frac{1}{N M^2}\sum_{k_1,k_2=1}^{\lfloor N/2 \rfloor} \sum_{\substack{j_1,j_2,j_3=1}}^M \delta_{j_1,j_3}\left[\innprod{\frac{2 \pi}{N}\F_{u_{j_1},\omega_{k_1},-\omega_{k_1},\omega_{k_2}}}{\overline{\F_{u_{j_1},-\omega_{k_1-1}}} \bigotimes {\F_{u_{j_2},-\omega_{k_2}}}}_{HS} +O(\frac{1}{T})  \right]
 \\&
  +\frac{1}{N M^2}\sum_{k_1,k_2=1}^{\lfloor N/2 \rfloor} \sum_{\substack{j_1,j_2,j_3=1}}^M \delta_{j_1,j_2} \left[ \innprod{\frac{2 \pi}{N}\F_{u_{j_1},-\omega_{k_1-1},\omega_{k_1-1},-\omega_{k_2}}}{\overline{\F_{u_{j_1},\omega_{k_1}}} \bigotimes {\F_{u_{j_3},\omega_{k_2}}}}_{HS}+O(\frac{1}{T}) \right]
  \\&
   +\frac{1}{N M^2}\sum_{k_1,k_2=1}^{\lfloor N/2 \rfloor} \sum_{\substack{j_1,j_2,j_3=1}}^M \delta_{j_1,j_3} \left[  \innprod{\frac{2 \pi}{N}\F_{u_{j_1},-\omega_{k_1-1},\omega_{k_1-1},\omega_{k_2}}}{\overline{\F_{u_{j_1},\omega_{k_1}}} \bigotimes {\F_{u_{j_2},-\omega_{k_2}}}}_{HS}+O(\frac{1}{T}) \right]\\
& +\frac{1}{N M^2}\sum_{k_1,k_2=1}^{\lfloor N/2 \rfloor} \sum_{\substack{j_1,j_2,j_3=1}}^M \delta_{j_1,j_3} \delta_{k_1-1,k_2}\delta_{j_1,j_3} \delta_{k_1-1,k_2} \left[  \innprod{\overline{\F_{u_{j_1},\omega_{k_1}}}^{\dagger} \F_{u_{j_1},-\omega_{k_1-1}}}{\overline{ \F_{u_{j_1},\omega_{k_1-1}}}\F_{u_{j_2},-\omega_{k_2}}^{\dagger}}_{HS}+O(\frac{1}{M^2}) \right]\\
 &+\frac{1}{N M^2}\sum_{k_1,k_2=1}^{\lfloor N/2 \rfloor} \sum_{\substack{j_1,j_2,j_3=1}}^M \delta_{j_1,j_2} \delta_{k_1-1,k_2}\delta_{j_1,j_2} \delta_{k_1-1,k_2}\delta_{j_1,j_2}\left[\innprod{\overline{\F_{u_{j_1},\omega_{k_1}}}^{\dagger}\F_{u_{j_1},-\omega_{k_1-1}}}{\overline{\F_{u_{j_1},\omega_{k_1-1}}}\overline{\F_{u_{j_3},\omega_{k_2}}}}_{HS}  +O(\frac{1}{M^2})\right]
 \\&
  +\frac{1}{N M^2}\sum_{k_1,k_2=1}^{\lfloor N/2 \rfloor} \sum_{\substack{j_1,j_2,j_3=1}}^M \delta_{j_1,j_2}\delta_{k_1,k_2}\delta_{j_1,j_2} \delta_{k_1,k_2}\left[\innprod{\F_{u_{j_1},\omega_{k_1}} \widetilde{\bigotimes} \overline{\F_{u_{j_1},-\omega_{k_1}}}}{\overline{\F_{u_{j_1},-\omega_{k_1-1}}} \bigotimes \F_{u_{j_2},\omega_{k_2}}}_{HS} +O(\frac{1}{M^2}) \right]\\&
    +\frac{1}{N M^2}\sum_{k_1,k_2=1}^{\lfloor N/2 \rfloor} \sum_{\substack{j_1,j_2,j_3=1}}^M \delta_{j_1,j_3}\delta_{k_1,k_2} \delta_{j_1,j_3} \delta_{k_1,k_2}\left[  \innprod{\F_{u_{j_1},\omega_{k_1}} \widetilde{\bigotimes}_{\top} {\F_{u_{j_1},-\omega_{k_1}}}}{\overline{\F_{u_{j_1},-\omega_{k_1-1}}} \bigotimes\F_{u_{j_2},-\omega_{k_2}}}_{HS}+O(\frac{1}{M^2}) \right]\\&
\end{align*}

So that, as $N, M \to \infty$,
\begin{align*}
NM \cv (\hat{F}_{1,T}, \hat{F}_{2,T})  \to 
&  \frac{2}{8\pi} \int_{-\pi}^{\pi} \int_{-\pi}^{\pi}\int_0^{1} \innprod{\F_{u,\omega_{1},-\omega_{1},-\omega_{2}}}{\F_{u,\omega_{1}} \bigotimes \widetilde{\F}_{\omega_{2}}}_{HS} du d\omega_1 d\omega_2 \\ 
&  \frac{2}{8\pi} \int_{-\pi}^{\pi} \int_{-\pi}^{\pi}\int_0^{1} \innprod{\F_{u,\omega_{1},-\omega_{1},\omega_{2}}}{\F_{u,\omega_{1}} \bigotimes \widetilde{\F}_{-\omega_{2}}}_{HS} du d\omega_1 d\omega_2 \\ 
 &+  \frac{2}{4\pi} \int_{-\pi}^{\pi} \int_0^{1} \innprod{\F_{u,\omega} \F_{u,\omega}}{ \F_{u,\omega} \widetilde{\F}_{\omega} }_{HS} du d\omega \\ 
 &+  \frac{1}{4\pi} \int_{-\pi}^{\pi} \int_0^{1} \innprod{\F_{u,\omega} \widetilde{\bigotimes} \F_{u,\omega}}{\F_{u,\omega} \bigotimes \widetilde{\F}_{\omega}}_{HS}du d\omega \\ 
  &+  \frac{1}{4\pi} \int_{-\pi}^{\pi} \int_0^{1} \innprod{\F_{u,\omega} \widetilde{\bigotimes}_{\top} \F_{u,-\omega}}{\F_{u,\omega} \bigotimes \widetilde{\F}_{-\omega}}_{HS}du d\omega \\ \end{align*} 
  
%%%%%%%%%%%%%%%%%%%%%%%%%%%%%%%%%%%%%%%%%%%%%%%%%
\subsubsection{Limiting Variance of $\widehat{m}_T$}
\label{54}
The limiting variance of $\widehat{m}_T$ is given by 
\[\nu^2=\lim_{T \to \infty}\Big(16 \pi^2 \V(\hat{F}_{1,T})+16 \pi^2 \V(\hat{F}_{2,T})-32 \pi^2 \cv(\hat{F}_{1,T},\hat{F}_{2,T})\Big).\] The above therefore yields the following expression for the asymptotic variance 
\begin{align*}
&\nu^2= 4\pi \int_{-\pi}^{\pi} \int_{-\pi}^{\pi}\int_0^{1} 
\innprod{\F_{u,\omega_{1},-\omega_{1},-\omega_{2}}}{\F_{u,\omega_{1}} \bigotimes \F_{u,\omega_{2}} }_{HS}du d\omega_1 d\omega_2\\& + 4\pi \int_{-\pi}^{\pi} \int_{-\pi}^{\pi}\innprod{\widetilde{\F}_{\omega_1,-\omega_{1},-\omega_{2}}}{{\widetilde{\F}_{\omega_{1}}} \bigotimes \widetilde{\F}_{\omega_{2}}}_{HS}  d\omega_1 d\omega_2- 8\pi \int_{-\pi}^{\pi} \int_{-\pi}^{\pi}\int_0^{1} \innprod{\F_{u,\omega_{1},-\omega_{1},-\omega_{2}}}{\F_{u,\omega_{1}} \bigotimes \widetilde{\F}_{\omega_{2}}}_{HS} du d\omega_1 d\omega_2\\&
+  4\pi \int_{-\pi}^{\pi} \int_{-\pi}^{\pi}\int_0^{1}\innprod{ \F_{u,\omega_{1},-\omega_{1},\omega_{2}}}{ \F_{u,\omega_{1}}\bigotimes \F_{u,-\omega_{2}}}_{HS}  du d\omega_1 d\omega_2  + 4\pi \int_{-\pi}^{\pi} \int_{-\pi}^{\pi}\innprod{\widetilde{\F}_{\omega_1,-\omega_{1},\omega_{2}}}{{\widetilde{\F}_{\omega_{1}}} \bigotimes \widetilde{\F}_{-\omega_{2}}}_{HS}  d\omega_1 d\omega_2 
\\& - 8\pi  \int_{-\pi}^{\pi} \int_{-\pi}^{\pi}\int_0^{1} \innprod{\F_{u,\omega_{1},-\omega_{1},\omega_{2}}}{\F_{u,\omega_{1}} \bigotimes \widetilde{\F}_{-\omega_{2}}}_{HS} du d\omega_1 d\omega_2\\
 &+  8\pi \int_{-\pi}^{\pi} \int_0^{1} \snorm{\F^2_{u,\omega}}^2_2 du d\omega +  4\pi\int_{-\pi}^{\pi} \int_0^{1}  \innprod{\F^2_{u,\omega }}{\F_{u,\omega} \widetilde{\F}_{\omega}}_{HS}du d\omega +  4\pi\int_{-\pi}^{\pi} \int_0^{1} \innprod{{\widetilde{\F}_{\omega} }\F_{u,\omega}}{\F_{u,\omega} {\widetilde{\F}_{u,\omega}}}_{HS}du d\omega \\ 
  &- 16\pi  \int_{-\pi}^{\pi} \int_0^{1} \innprod{\F_{u,\omega} \F_{u,\omega}}{ \F_{u,\omega} \widetilde{\F}_{\omega} }_{HS} du d\omega + 4\pi  \int_{-\pi}^{\pi} \int_0^{1}\snorm{\F_{u,\omega}}^4_2 du d\omega  \\ 
 &+ 4\pi  \int_{-\pi}^{\pi} \int_0^{1}   \innprod{\F_{u,\omega} \widetilde{\bigotimes} \F_{u,\omega}}{\F_{u,\omega} \bigotimes \F_{u,\omega}}_{HS}du d\omega +  4\pi  \int_{-\pi}^{\pi} \int_0^{1} \innprod{\F_{u,\omega} \widetilde{\bigotimes}_{\top} \F_{u,-\omega}}{ \F_{u,\omega}\bigotimes\F_{u,-\omega}}_{HS} du d\omega 
 \\
 &+  4\pi \int_{-\pi}^{\pi} \int_0^{1}\innprod{\F_{u,\omega} \widetilde{\bigotimes} \F_{u,\omega}}{ \widetilde{\F}_{{\omega}} \bigotimes \widetilde{\F}_{\omega}}_{HS}  du d\omega +  4\pi \int_{-\pi}^{\pi} \int_0^{1}\innprod{\F_{u,\omega} \widetilde{\bigotimes}_{\top} \F_{u,-\omega}}{ \widetilde{\F}_{{\omega}} \bigotimes \widetilde{\F}_{-\omega}}_{HS}  du d\omega
\\ 
 &-8\pi \int_{-\pi}^{\pi} \int_0^{1} \innprod{\F_{u,\omega} \widetilde{\bigotimes} \F_{u,\omega}}{\F_{u,\omega} \bigotimes \widetilde{\F}_{\omega}}_{HS}du d\omega -8\pi  \int_{-\pi}^{\pi} \int_0^{1} \innprod{\F_{u,\omega} \widetilde{\bigotimes}_{\top} \F_{u,-\omega}}{\F_{u,\omega} \bigotimes \widetilde{\F}_{-\omega}}_{HS}du d\omega \\ 
  \tageq\label{eq:asvar} 
\end{align*}

Under $H_0$ this reduces to 
\[\nu^2_{H_0}=4\pi \int_{-\pi}^{\pi} \snorm{\widetilde{\F}_\omega}_2^4d\omega.\]
\subsection{Proof of Lemma \ref{consvarest} (consistency of variance estimate)}\label{subsec:estvar}
 
\textit{Proof of Lemma \ref{consvarest}}: We write 
\begin{align*}
\E(\hat v_{H_0}^2 ) =&\frac{16\pi^2}{N}\sum_{k=1}^{\lfloor N/2 \rfloor} \E \Big[\frac{1}{M}\sum_{j=1}^M\innprod{I_N^{u_j,\omega_k}}{I_N^{u_j,\omega_{k-1}}}_{HS}\Big]^2\\
= &\frac{16\pi^2}{N}\sum_{k=1}^{\lfloor N/2 \rfloor}\text{Var}\Big[\frac{1}{M}\sum_{j=1}^M\innprod{I_N^{u_j,\omega_k}}{I_N^{u_j,\omega_{k-1}}}_{HS}\Big] + \frac{16\pi^2}{N}\sum_{k=1}^{\lfloor N/2 \rfloor}\left(\E\Big[\frac{1}{M}\sum_{j=1}^M\innprod{I_N^{u_j,\omega_k}}{I_N^{u_j,\omega_{k-1}}}_{HS}\Big]\right)^2\\
= &\frac{16\pi^2}{NM^2}\sum_{k=1}^{\lfloor N/2 \rfloor}\sum_{j_1,j_2=1}^M\text{Cov}\left[\innprod{I_N^{u_{j_1},\omega_k}}{I_N^{u_{j_1},\omega_{k-1}}}_{HS}, \innprod{I_N^{u_{j_2},\omega_k}}{I_N^{u_{j_2},\omega_{k-1}}}_{HS}\right]\\
& + \frac{16\pi^2}{N}\sum_{k=1}^{\lfloor N/2 \rfloor}\left(\frac{1}{M}\sum_{j=1}^M \E\left[\innprod{I_N^{u_j,\omega_k}}{I_N^{u_j,\omega_{k-1}}}_{HS}\right]\right)^2. 
\end{align*} 
Using Theorem \ref{lem:cumHSinprod} we can write the first term as
\begin{align} \label{eq:varnuh0}
\frac{16\pi^2}{NM^2}\sum_{k=1}^{\lfloor N/2 \rfloor}\sum_{j_1,j_2=1}^M \Tr\Big(\sum_{\boldsymbol{P} = P_1 \cup \ldots \cup P_G}S_{\boldsymbol{P}}\Big(\otimes_{g=1}^{G} \operatorname{cum}\big(\fdft{p}{p}|p\in \nu_g\Big)\Big) 
\end{align}
where $p=(l,m)$ and $k_{p} = (-1)^{m} k_{2l-\delta{\{m \in \{1,2\}\}}}$ and $j_p= j_{2l-\delta{\{m \in \{1,2\}\}}}$ for $l \in\{1,\ldots, n\}$ and $m \in \{1,2,3,4\}$. In this case, we are interested in all indecomposible partitions of the array
\[\begin{matrix}
\underbrace{\fdft{1}{}}_{1}&\underbrace{\fdftc{1}{}}_2 & \underbrace{\fdftcl{1}{}}_3& \underbrace{\fdftl{1}{}}_7\\ 
\underbrace{\fdftcl{2}{}}_5 &\underbrace{\fdftl{2}{}}_6 &\underbrace{\fdft{2}{}}_7& \underbrace{\fdftc{2}{}}_8.
\end{matrix}\] Indecomposability immediately implies that there must be one restriction in time. Using the results in \autoref{sec5} and a similar argument as in \autoref{sec:cov} 
%using \autoref{cor:cumbound} 
will show that all of these are at most order $O(\frac{1}{M})$ and hence will vanish as $M \to \infty$. For example,
\begin{align*}
&\Tr\Big(S_{(18)(27)(36)(45)}\delta_{j_1,j_2}\big( \F_{u_{j_1},\omega_{k}}\otimes \F_{u_{j_1},-\omega_{k-1}} \otimes \F_{u_{j_1},-\omega_{k}}\otimes \F_{u_{j_1},\omega_{k-1}}+\Eps_2\big]\Big)\\ 
&\Tr\Big(S_{(12)(78)(36)(45)}\delta_{j_1,j_2}\big( \F_{u_{j_1},\omega_{k}}\otimes \F_{u_{j_2},\omega_{k}} \otimes \F_{u_{j_1},-\omega_{k}}\otimes \F_{u_{j_1},\omega_{k-1}}+\Eps_2\big]\Big)\\ 
& \vdots
\end{align*}
Using again Theorem \autoref{lem:cumHSinprod}, we can prove similar to the proof of $\E \hat{F}_{1,T}$, that the second term converges to $4\pi \int_{-\pi}^{\pi} \Big( \int^1_0 \snorm{\F_{u,\omega}}^2_2 du \Big)^2 d\omega$.
Under $H_0$, we have  ${\F}_{\omega,u}  \equiv {\F}_{\omega}$ and it follows therefore that the second term  converges to $4\pi \int_{-\pi}^{\pi} \snorm{\widetilde{\F}_\omega}_2^4d\omega$ if the null is satisifed.

For the variance of the estimator, we write 
 \begin{equation}\label{eq:varofestvar} 
\text{Var}( \hat v_{H_0}^2) = \E [\hat v_{H_0}^2]^2-\big(\E [\hat v_{H_0}^2] \big)^2.
%= \frac{16\pi^2}{N}\sum_{k=1}^{\lfloor N/2 \rfloor} \left[\frac{1}{M}\sum_{j=1}^M\innprod{I_N^{u_j,\omega_k}}{I_N^{u_j,\omega_{k-1}}}_{HS}\right]^2.
 \end{equation}
Under $H_0$, the above derivation yields that the second term of \eqref{eq:varofestvar} converges to $\big(4\pi \int_{-\pi}^{\pi} \snorm{\widetilde{\F}_\omega}_2^4d\omega)^2$. Consider then decomposing the first term of \eqref{eq:varofestvar} as
 \begin{align*}
 \E [\hat v_{H_0}^2]^2= &\frac{2^8 \pi^4}{N^2}\sum_{k_1,k_2}\E\left[\frac{1}{M^2}\sum_{j_1,j_2}\innprod{I_N^{u_{j_1},\omega_{k_1}}}{I_N^{u_{j_1},\omega_{k_1-1}}}_{HS}\innprod{I_N^{u_{j_2},\omega_{k_2}}}{I_N^{u_{j_2},\omega_{k_2-1}}}_{HS}\right]^2\\
= &\frac{2^8 \pi^4}{N^2}\sum_{k_1,k_2}\text{Var}\left[\frac{1}{M^2}\sum_{j_1,j_2}\innprod{I_N^{u_{j_1},\omega_{k_1}}}{I_N^{u_{j_1},\omega_{k_1-1}}}_{HS}\innprod{I_N^{u_{j_2},\omega_{k_2}}}{I_N^{u_{j_2},\omega_{k_2-1}}}_{HS}\right]  \tageq  \label{eq:estvardec1}\\
&+\frac{2^8 \pi^4}{N^2}\sum_{k_1,k_2}\left[\E\left(\frac{1}{M^2}\sum_{j_1,j_2}\innprod{I_N^{u_{j_1},\omega_{k_1}}}{I_N^{u_{j_1},\omega_{k_1-1}}}_{HS}\innprod{I_N^{u_{j_2},\omega_{k_2}}}{I_N^{u_{j_2},\omega_{k_2-1}}}_{HS}\right)\right]^2 \tageq  \label{eq:estvardec2}
\end{align*}
We consider\eqref{eq:estvardec1} and \eqref{eq:estvardec2} separately. Using the product theorem for cumulants \eqref{eq:estvardec1} equals
\begin{align*}
%= &\frac{2^8\pi^4}{N^2M^4}\sum_{k_1,k_2}\sum_{j_1,j_2,j_3,j_4}\text{Cov}\left(\innprod{I_N^{u_{j_1},\omega_{k_1}}}{I_N^{u_{j_1},\omega_{k_1-1}}}_{HS}\innprod{I_N^{u_{j_2},\omega_{k_2}}}{I_N^{u_{j_2},\omega_{k_2-1}}}_{HS},\right.\\
%& \hspace{1.4 in}\left.\innprod{I_N^{u_{j_3},\omega_{k_1}}}{I_N^{u_{j_3},\omega_{k_1-1}}}_{HS}\innprod{I_N^{u_{j_4},\omega_{k_2}}}{I_N^{u_{j_4},\omega_{k_2-1}}}_{HS}\right)\\
&\frac{2^8\pi^4}{N^2M^4}\sum_{k_1,k_2}\sum_{j_1,j_2,j_3,j_4}\text{Cum}\Big(\innprod{I_N^{u_{j_1},\omega_{k_1}}}{I_N^{u_{j_1},\omega_{k_1-1}}}_{HS}\,,\,\overline{\innprod{I_N^{u_{j_3},\omega_{k_1}}}{I_N^{u_{j_3},\omega_{k_1-1}}}}_{HS}\Big)\\
& \phantom{\frac{2^8\pi^4}{N^2M^4}\sum_{k_1,k_2}\sum_{j_1,j_2,j_3,j_4}}
\times \text{Cum}\Big(\innprod{I_N^{u_{j_2},\omega_{k_2}}}{I_N^{u_{j_2},\omega_{k_2-1}}}_{HS}\,,\,\overline{\innprod{I_N^{u_{j_4},\omega_{k_2}}}{I_N^{u_{j_4},\omega_{k_2-1}}}}_{HS}\Big)\\
& + \frac{2^8\pi^4}{N^2M^4}\sum_{k_1,k_2}\sum_{j_1,j_2,j_3,j_4}\text{Cum}\Big(\innprod{I_N^{u_{j_1},\omega_{k_1}}}{I_N^{u_{j_1},\omega_{k_1-1}}}_{HS}\,,\,\overline{\innprod{I_N^{u_{j_4},\omega_{k_2}}}{I_N^{u_{j_4},\omega_{k_2-1}}}}_{HS}\Big)\\
& \phantom{\frac{2^8\pi^4}{N^2M^4}\sum_{k_1,k_2}\sum_{j_1,j_2,j_3,j_4}}
\times \text{Cum}\Big(\innprod{I_N^{u_{j_2},\omega_{k_2}}}{I_N^{u_{j_2},\omega_{k_2-1}}}_{HS}\,,\,\overline{\innprod{I_N^{u_{j_3},\omega_{k_1}}}{I_N^{u_{j_3},\omega_{k_1-1}}}}_{HS}\Big)
\end{align*}
 \autoref{lem:cumHSinprod} then shows that indecomposability of the first of these terms implies the restrictions $k_1=k_2$ and $\{j_1,j_2\} \cap \{j_3,j_4\} = \emptyset$, while indecomposability of the second implies the constraints $\{j_1,j_2\} \cap \{j_3,j_4\} = \emptyset$ only. Therefore, \eqref{eq:estvardec1} is of order $O(\frac{1}{NM^2}+\frac{1}{M^2})$ and hence converges to zero as $N,M \to \infty$. Finally, it is straightforward to show using a similar argument that \eqref{eq:estvardec2} equals
\begin{align*}
%&+\frac{2^8 \pi^4}{N^2}\sum_{k_1,k_2}\left[\frac{1}{M^2}\sum_{j_1,j_2} \E\left[ \innprod{I_N^{u_{j_1},\omega_{k_1}}}{I_N^{u_{j_1},\omega_{k_1-1}}}_{HS}\innprod{I_N^{u_{j_2},\omega_{k_2}}}{I_N^{u_{j_2},\omega_{k_2-1}}}_{HS}\right]\right]^2\\
&\frac{2^8 \pi^4}{N^2}\sum_{k_1,k_2}\Big[\frac{1}{M^2}\sum_{j_1,j_2} \text{Cum}\Big( \innprod{I_N^{u_{j_1},\omega_{k_1}}}{I_N^{u_{j_1},\omega_{k_1-1}}}_{HS},\overline{\innprod{I_N^{u_{j_2},\omega_{k_2}}}{I_N^{u_{j_2},\omega_{k_2-1}}}_{HS}}\Big) \\& 
\phantom{\frac{2^8 \pi^4}{N^2}\sum_{k_1,k_2}\Big[\frac{1}{M^2}\sum_{j_1,j_2} \text{Cum}}+ \E \innprod{I_N^{u_{j_1},\omega_{k_1}}}{I_N^{u_{j_1},\omega_{k_1-1}}}_{HS} \E\innprod{I_N^{u_{j_2},\omega_{k_2}}}{I_N^{u_{j_2},\omega_{k_2-1}}}_{HS}\Big) \Big]^2 \\
&=\frac{2^8 \pi^4}{N^2}\sum_{k_1,k_2}\Big[O(\frac{1}{M})+\frac{1}{M^2}\sum_{j_1,j_2} \E \innprod{I_N^{u_{j_1},\omega_{k_1}}}{I_N^{u_{j_1},\omega_{k_1-1}}}_{HS} \E\innprod{I_N^{u_{j_2},\omega_{k_2}}}{I_N^{u_{j_2},\omega_{k_2-1}}}_{HS}\Big) \Big]^2
\\ 
&=\frac{2^8 \pi^4}{N^2}\sum_{k_1,k_2}\Big[\frac{1}{M}\sum_{j_1} \snorm{\F_{u_{j_1},\omega_{k_1}}}^2_2 \frac{1}{M}\sum_{j_2}\snorm{\F_{u_{j_2},\omega_{k_2}}}^2_2 + O(\frac{1}{M}) \Big]^2. 
\end{align*}
Under $H_0$,the latter converges to $\big(4\pi \int_{-\pi}^{\pi} \snorm{\widetilde{\F}_\omega}_2^4d\omega)^2$. Altogether, the above derivation shows that $%&=\frac{2^8 \pi^4}{N^2}\sum_{k_1,k_2}\Big[ \snorm{\widetilde{\F}_{\omega_{k_1}}}^2_2 \snorm{\widetilde{\F}_{\omega_{k_2}}}^2_2  \Big]^2 \to 
 \E [\hat v_{H_0}^2]^2\to \big(4\pi \int_{-\pi}^{\pi} \snorm{\widetilde{\F}_\omega}_2^4d\omega)^2$
as $N, M \to \infty$. Since the second term of \eqref{eq:varofestvar} converges to the same limit, we thus find $\text{Var}( \hat v_{H_0}^2) \to 0$ and consequently $\hat{v}_{H_0}^2 \stackrel{p}{\to} v_{H_0}^2$.
%%%%%%%%%%%%%%%%%%%%%%%%%%%%%%%%%%%%%%%%%%%%%%%%%%%%%%

%%%%%%%%%%%%%%%%%%%%%%%%%%%%%%%%%%%%%%%%%%%%%%%%%%%%%%

\end{appendices}

%\begin{comment}

%\end{comment}

%\bibliography{References}

\end{document}